\documentclass[11pt,notitlepage]{article}
\usepackage[active]{srcltx}
\usepackage{xcolor}
\usepackage{mathrsfs}
\usepackage{amsmath}
\usepackage{amsthm}
\usepackage{amssymb}
\usepackage{graphicx}
\usepackage[letterpaper]{geometry}
\usepackage{setspace}
\usepackage[authoryear]{natbib}
\PassOptionsToPackage{normalem}{ulem}
\usepackage{ulem}
\usepackage[]
 {hyperref}

\makeatletter

\providecommand{\tabularnewline}{\\}
\providecolor{lyxadded}{rgb}{0,0,1}
\providecolor{lyxdeleted}{rgb}{1,0,0}
\DeclareRobustCommand{\mklyxadded}[1]{\textcolor{lyxadded}\bgroup#1\egroup}
\DeclareRobustCommand{\mklyxdeleted}[1]{\textcolor{lyxdeleted}\bgroup\mklyxsout{#1}\egroup}
\DeclareRobustCommand{\mklyxsout}[1]{\ifx\\#1\else\sout{#1}\fi}

\theoremstyle{definition}
\newtheorem{defn}{\protect\definitionname}[section]
\theoremstyle{plain}
\newtheorem{assumption}{\protect\assumptionname}
\theoremstyle{plain}
\newtheorem{thm}{\protect\theoremname}[section]
\theoremstyle{plain}
\newtheorem{lem}{\protect\lemmaname}[section]

\usepackage{float}
\usepackage{subfig}
\usepackage{geometry}

\makeatother

\providecommand{\assumptionname}{Assumption}
\providecommand{\definitionname}{Definition}
\providecommand{\lemmaname}{Lemma}
\providecommand{\theoremname}{Theorem}

\begin{document}
\title{Uniform Validity of the Subset Anderson-Rubin Test under Heteroskedasticity
and Nonlinearity\hspace{2pt}\thanks{We thank participants of the Econometrics Workshop at Notre Dame 2024
and seminar participants at the Federal Reserve Bank of New York for
helpful comments. The views expressed in this paper are the sole responsibility
of the authors and to not necessarily reflect the views of the Federal
Reserve Bank of San Francisco or the Federal Reserve System.}}
\author{Atsushi Inoue\hspace{0.5pt}\thanks{Department of Economics, Vanderbilt University (\protect\protect\href{mailto:atsushi.inoue@vaderbilt.edu}{atsushi.inoue@vaderbilt.edu}).}\and
\`{O}scar Jord\`{a}\hspace{0.5pt}\thanks{Federal Reserve Bank of San Francisco; Department of Economics, University
of California, Davis; and CEPR (\protect\protect\href{mailto:oscar.jorda@sf.frb.org}{oscar.jorda@sf.frb.org};
\protect\protect\href{mailto:ojorda@ucdavis.edu}{ojorda@ucdavis.edu}).}\and Guido M. Kuersteiner\hspace{2pt}\thanks{Department of Economics, University of Maryland (\protect\protect\href{mailto:gkuerste@umd.edu}{gkuerste@umd.edu}).}}
\maketitle
\begin{abstract}
{\normalsize We consider the Anderson-Rubin (AR) statistic for a general
set of nonlinear moment restrictions. The statistic is based on the
criterion function of the continuous updating estimator (CUE) for
a subset of parameters not constrained under the Null. We treat the
data distribution nonparametrically with parametric moment restrictions
imposed under the Null. We show that subset tests and confidence intervals
based on the AR statistic are uniformly valid over a wide range of
distributions that include moment restrictions with general forms
of heteroskedasticity. We show that the AR based tests have correct
asymptotic size when parameters are unidentified, partially identified,
weakly or strongly identified. We obtain these results by constructing
an upper bound that is using a novel perturbation and regularization
approach applied to the first order conditions of the CUE. Our theory
applies to both cross-sections and time series data and does not assume
stationarity in time series settings or homogeneity in cross-sectional
settings.}{\normalsize\par}
\end{abstract}
\bigskip{}
\noindent{} \emph{JEL classification codes:} C11, C12, C22, C32, C44.

\bigskip{}
 \bigskip{}
\noindent{}

\section{Introduction}

\global\long\def\diag{\operatorname{diag}}%
\global\long\def\tr{\operatorname{tr}}%
\global\long\def\vecd{\operatorname{vec_{D}}}%
\global\long\def\vec{\operatorname{vec}}%
\global\long\def\plim{\operatorname{plim}}%
\global\long\def\cov{\operatorname{Cov}}%
\global\long\def\var{\operatorname{Var}}%
\global\long\def\argmax{\operatorname{argmax}}%
\global\long\def\argmin{\operatorname{argmin}}%
\global\long\def\rank{\operatorname{rank}}%

We analyze Anderson-Rubin (AR) statistics based on the continuous
updating estimator (CUE) for parametric non-linear moment conditions
and under a null hypothesis that restricts only a subset of the parameters.
We extend results about the uniform validity of these statistics by
Guggenberger, Kleibergen, Mavroeidis and Chen \citeyearpar{GuggenbergerKleibergenMavroeidisChen2012},
GKMC12, who considered the case of linear models in iid settings and
under homoskedasticity. Our CUE framework is non-linear, allows for
heteroskedasticity, general dependence and heterogeneity. In time
series environments stationarity is not assumed but we do rule out
trending behavior including the presence of unit roots. We make no
parametric assumptions about the data generating process. The null
hypothesis only imposes a set of possibly non-linear moment restrictions
that are parameterized by a finite dimensional parameter. Under our
assumptions, estimated parameters may be partially or completely unidentified
and there may not be a valid limiting distribution for estimated parameters
even in situations where the parameters are identified. The setting
also includes classical weak and non-identified scenarios where the
limiting distribution of parameter estimates is non-standard and/or
depends on nuisance parameters that cannot be estimated. This means
that Wald type inference is potentially highly unreliable. Our inference
procedure using the Anderson-Rubin statistic is simple to implement
and only requires conventional chi-square based critical values. Our
conditions center around objects that the researcher observes and
controls. These include the data used in the analysis and the moment
conditions that are being specified for inference. On the other handwe
treat the unobserved data generating process non-parametrically.

Our theory contributes to the literature on inference when parameters
are unidentified, partially or weakly identified. \citet{Phillips1989}
is an early investigation of inference in partially identified models.
The weak instrument literature to which \citet{Dufour1997}, \citet{Staiger1997},
\citet{Stock2000}, \citet{Kleibergen2002}, \citet{Kleibergen2005}
and \citet{Moreira2003} made important early contributions focuses
on inference that is robust to a lack of identification. \citet{Dufour1997}
and \citet{Staiger1997} are the first to note that the Anderson-Rubin
(AR) statistic, introduced by \citet{Anderson1949}, can be used to
construct valid confidence intervals in weakly identified settings.
\citet{Kleibergen2002} develops a modified LM statistic for linear
instrumental variables models under homoskedasticity that reduces
to the AR statistic in the just identified case. \citet{Stock2000}
propose a generalized version of the AR statistic based on the CUE.
\citet{Stock2000} and \citet{Dufour2005} consider subvector tests
in settings where only parameters that are part of the null hypothesis
may be underidentified. \citet{Kleibergen2005} proposes the K-statistic
which is based on the first order conditions of CUE to account for
conditional heteroskedasticity and serial correlation. Of these papers,
\citet{Stock2000} consider moment conditions and test statistics
that are similar to our setting. Their theoretical analysis is limited
to weak instrument asymptotic sequences which are not sufficient to
establish uniform validity. \citet{Kleibergen2005} considers subvector
tests for a general set of moment conditions as we are but imposes
strong identification of the parameters not subject to the null hypothesis.
His asymptotic analysis includes cases where the parameters restricted
under the null are unidentified or weakly identified but is pointwise
rather than uniform. As GKMC12, we do not impose any restrictions
on the degree of identification of either set of parameters and provide
a uniform asymptotic analysis. There is a large subsequent literature
on inference robust to identification failure including Andrews, Moreira
and Stock (2006)\nocite{Andrews2006}, \citet{Chaudhuri2011}, \citet{Andrews2012a},
\citet{Andrews2016}, \citet{Andrews2019}, Andrews, Marmer and Yu
\citeyearpar{Andrews2019a}, Andrews, Cheng and Guggenberger \citeyearpar{Andrews2020}
and Guggenberger, Kleibergen and Mavroeidis \citeyearpar{Guggenberger2024}.

Our theoretical results are closest to Guggenberger, Kleibergen, Mavroeidis
and Chen \citeyearpar{GuggenbergerKleibergenMavroeidisChen2012},
GKMC12, who show that the AR statistic applied to a subset of the
parameters in linear instrumental variables models has uniformly valid
size over a class of distributions that include completely unidentified
scenarios. Their results are valid for iid data with homoskedastic
errors. Guggenberger, Kleibergen and Mavroeidis \citeyearpar{Guggenberger2019},
GKM19, develop a more powerful subvector test under homoskedasticity
but require numerically computed critical values. \citet{Andrews2017a}
proposes a two step subvector test that is robust to heteroskedasticity
but requires a data dependent significance level in the second step.
Guggenberger, Kleibergen and Mavroeidis \citeyearpar{Guggenberger2024},
GKM24, extend GKM19 to allow for conditional heteroskedasticity by
introducing a model selection procedure that selects between the AR
statistic for linear models analyzed in GMK19 and the AR/AR test of
\citet{Andrews2017a}. There does not seem to exist theoretical results
that show the uniform validity of a simple one step subvector test
robust to non-identification of the untested parameters for general
heteroskedasticity, temporal dependence and non-stationarity. We show
that a test based on the Anderson-Rubin testing principle formulated
for the CUE is uniformly valid. The AR test is based on conventional
critical values for the chi-square distribution and thus easy to implement.
Our assumptions are high level and do not assume a specific model,
nor identification of any of the estimated or tested parameters. \citet{Andrews2019}
consider tests under general identification failure and heteroskedasticity
but only consider subvector tests for the case when the parameters
not subject to the null are strongly identified. Of the four types
of regularity conditions that Andrews and Guggenberger (2019,p. 1709)
allow to fail, we account for two failures, namely that the moment
functions have multiple solutions and that the Jacobian of the moment
functions are column rank deficient. On the other hand, and unlike
\citet{Andrews2019}, we do maintain that the moment functions have
full rank covariance matrices and that parameters under the null are
in the interior of the parameter space. We also consider valid confidence
intervals that are based on inverting the AR statistic. Inverting
the AR statistic to construct a confidence interval dates back to
\citet{Anderson1949} and was more recently investigated in the context
of weakly identified or unidentified models by \citet{Dufour1997},
\citet{Staiger1997} and \citet{Stock2000}.

We illustrate the assumptions of our theory using examples of simultaneous
equations models with heteroskedasticity and temporal dependence as
well as a time series application with inference for impulse response
analysis using the local projections proposed by \citet{Jorda2005}.
To the extent that the robust testing literature considered time series
settings, stationarity has been maintained to the best of our knowledge.
We allow for general forms of non-stationarity for processes with
well defined second moments. While this rules out trending forms of
non-stationarity such as unit roots and non-stationary fractionally
integrated processes we do allow for time-changing relationships between
outcomes, policy variables and controls.

The paper is organized as follows. Section \ref{sec:Tests} defines
the inference problem and the test statistic we consider. Section
\ref{sec:Theory} discusses high level assumptions and presents the
main theoretical result establishing uniform validity. The section
also discusses the main steps in our proof strategy. Section \ref{sec:Examples}
presents two examples that illustrate the high level assumptions we
impose. Section \ref{sec:Monte-Carlo-Experiments} presents Monte
Carlo evidence using the New Keynesian Phillips curve model considered
by \citet{Kleibergen2009}. Proofs as well as some Lemmas are contained
in the Appendix.

\section{Test Statistics \label{sec:Tests}}

Consider a probability space $\left(\Omega,\mathcal{F},P\right)$
on which a vector valued double array\footnote{A double array does not restrict the index $t$ for a given $n,$
while a triangular array imposes the restriction $t\leq n$.} of random variables $\chi_{n,t}\in\mathbb{R}^{d_{\chi}}$ is defined
and where $n$ is an integer value representing sample size and $t$
is the observation index. We do not assume a parametric data generating
process but instead consider sequences of induced probability measures
$P_{\chi_{n,t}}$. We are interested in testing a null hypothesis
that imposes constraints on $P_{\chi_{n,t}}$ that can be expressed
in terms of moment conditions and that are parametrized by finite
dimensional parameters $\beta$ and $\gamma$. More specifically,
the null hypothesis imposes the constraint $H_{0}:\beta=\beta_{n,0}$
while $\gamma$ remains unconstrained under the null. Constraints
on moments are expressed using moment functions $g\left(\chi_{n,t};\beta,\gamma\right)$
where $g\left(.\right):\mathbb{R}^{d_{\chi}}\times\mathbb{R}^{d_{\beta}+d_{\gamma}}\rightarrow\mathbb{R}^{d}$.
For ease of notation we define $g_{t}\left(\beta,\gamma\right)\equiv g\left(\chi_{n,t};\beta,\gamma\right)$
and suppress the dependence on $\chi_{n,t}$. We do assume that the
moment condition
\begin{equation}
n^{-1}\sum_{t=1}^{n}E\left[g_{t}\left(\beta_{n,0},\gamma_{n,0}\right)\right]=0\label{eq:MomentCond}
\end{equation}
holds under $H_{0}$ for at least one value $\gamma_{n,0}$. In a
strongly identified setting $\beta_{n,0}$ and $\gamma_{n,0}$ are
the unique solutions to \eqref{eq:MomentCond}. However, we do allow
for partially, weakly and unidentified cases where the values $\beta_{n,0}$
and $\gamma_{n,0}$ may not be the only parameter values that satisfy
\eqref{eq:MomentCond}.

Our framework allows for general dependence patters, including time
series data. We distinguish between outcome variables $y_{n,t}$,
exogenous or endogenous variables $x_{n,t}$ associated with the parameters
$\beta$ in linear models, additional endogenous, exogenous or predetermined
variables $w_{n,t}$ associated with auxiliary parameters $\gamma$
and instruments $z_{n,t}$ that contain exogenous components of $x_{nt}$
and $w_{n,t}$. When moment conditions are based on non-linear functions
the association between parameters $\beta$, $\gamma$ and specific
covariates is often less strong and is not required for our results.
In linear models where all variables in $x_{n,t}$ and $w_{n,t}$
are exogenous $z_{n,t}$ may be composed only of these variables.
The vector $\chi_{n,t}$ contains all the distinct random variables
contained in $y_{n,t},x_{n,t},w_{n,t}$ and $z_{n,t}$. We define
$X_{n,t}:=\left(x_{n,t}',w_{n,t}'\right)'$ with $X_{n,t}\in\mathbb{R}^{d_{x}}$.
We assume that we observe a sample $\chi_{n,t}$ for $t=1,...,n.$

The empirical analog to the moment condition is based on sample averages
\begin{equation}
\hat{g}_{n}\left(\beta,\gamma\right)=n^{-1}\sum_{t=1}^{n}g_{t}\left(\beta,\gamma\right).\label{eq:gn_hat}
\end{equation}
When evaluated at $\theta_{n,0}=\left(\beta_{n,0}',\gamma_{n,0}'\right)'$
we use the shorthand notation $\hat{g}_{n,0}=\hat{g}_{n}\left(\beta_{n,0},\gamma_{n,0}\right).$
The moment conditions \eqref{eq:MomentCond} may be general scores
of $Z$-estimators or may be moment functions for a GMM estimator
with instruments $z_{t}$ such that 
\begin{equation}
g_{t}\left(\beta,\gamma\right)=q\left(y_{n,t},X_{n,t},\beta,\gamma\right)\otimes z_{n,t}\label{eq:GMM_Moments}
\end{equation}
and where $q\left(.\right):\mathbb{R}^{d_{x}+1}\times\mathbb{R}^{d_{\beta}+d_{\gamma}}\rightarrow\mathbb{R}^{d_{q}}$
is a $d_{q}$-dimensional vector of functions. The total number of
moment restrictions is $d=d_{q}d_{z}$.

Test statistics for $H_{0}$ are based on the CUE criterion function
evaluated under $H_{0}$. The CUE was proposed by \citet{Sargan1958}
and analyzed by Hansen, Heaton and Yaron \citeyearpar{HansenHeatonYaron1996}
and more recently \citet{Donald2000} and \citet{Newey2004}. \citet{Stock2000}
were the first to consider the CUE criterion function for hypothesis
testing in a weakly identified setting involving possibly non-linear
moment restrictions.

The CUE criterion function depends on an estimate of the variance
of $\hat{g}_{n}\left(\beta,\gamma\right)$. Following the literature
on heteroskedasticity and autocorrelation robust standard errors (HAC),
see Newey and West (1987, 1994)\citet{Newey1987}, \citet{Newey2004},
\citet{Andrews1991}, \citet{Andrews1992}, and in particular \citet{Jong2000},
we define the matrix 
\[
\Omega_{n}\left(\beta,\gamma\right)=n^{-1}\sum_{t=1}^{n}\sum_{s=1}^{n}E\left[g_{t}\left(\beta,\gamma\right)g_{s}\left(\beta,\gamma\right)'\right].
\]
In addition, when evaluated at a sequence of pseudo-true values $\beta_{n,0}$
and $\gamma_{n,0}$ the covariance matrix $\Omega_{n,0}$ is defined
as 
\[
\Omega_{n,0}=\Omega_{n}\left(\beta_{n,0},\gamma_{n,0}\right)=\var\left(n^{-1/2}\sum_{t=1}^{n}g_{t}\left(\beta_{n,0},\gamma_{n,0}\right)\right)
\]
 where the last equality holds because the moment conditions $n^{-1/2}\sum_{t=1}^{n}E\left[g_{t}\left(\beta_{n,0},\gamma_{n,0}\right)\right]=0$
are satisfied for $\beta_{n,0}$ and $\gamma_{n,0}$.

The matrix $\Omega_{n,0}$ can be replaced by a sample analog following
constructions in the HAC literature, in particular De Jong and Davidson
(2000) who explicitly account for triangular arrays that are central
to our theoretical discussion. To define the estimators of the GMM
weight matrix and long run variances, we adopt the following definition
of the class of kernel functions $k\left(.\right)$ with bandwidth
sequence $a_{n}$ from \citet{Jong2000}, see also \citet{Andrews1991}.
\begin{defn}
Let $\mathscr{\mathcal{\mathscr{K}}}$ be the class of kernel functions
$k\left(.\right):\mathbb{R}\rightarrow\left[-1,1\right]$ such that
$k\left(0\right)=0,$ $k\left(x\right)=k\left(-x\right)$ for all
$x\in\mathbb{R},$ $\int_{-\infty}^{\infty}\left|k\left(x\right)\right|dx<\infty$
and for $\psi\left(\xi\right)=\int_{-\infty}^{\infty}k\left(x\right)e^{i\xi x}dx$,
$\psi\left(\xi\right)\geq0.$
\end{defn}
Then 
\begin{equation}
\hat{\Omega}_{n}\left(\beta,\gamma\right)=n^{-1}\sum_{t=1}^{n}\sum_{s=1}^{n}k\left(\frac{t-s}{a_{n}}\right)g_{t}\left(\beta,\gamma\right)g_{s}\left(\beta,\gamma\right)'.\label{eq:Omega_hat}
\end{equation}
The continuous updating estimator for $\theta=\left(\beta',\gamma'\right)'$
is now given by 
\[
\hat{\theta}_{CUE}=\left(\hat{\beta}'_{CUE},\hat{\gamma}'_{CUE}\right)'=\textrm{argmin}_{\beta,\gamma}n^{-1}Q_{n}\left(\beta,\gamma\right)
\]
where $Q_{n}\left(\beta,\gamma\right)=n\hat{g}_{n}\left(\beta,\gamma\right)'\hat{\Omega}_{n}^{-1}\left(\beta,\gamma\right)\hat{g}_{n}\left(\beta,\gamma\right).$
The subset AR statistic for the null hypothesis $H_{0}$ based on
the CUE is given as 
\begin{equation}
\textrm{AR}_{\textrm{C}}\left(\beta_{n,0}\right)=\textrm{min}_{\gamma}Q_{n}\left(\beta_{n,0},\gamma\right).\label{eq:CUAR}
\end{equation}
Let $d_{\gamma}$ be the dimension of $\gamma$ and assume that $d-d_{\gamma}>0.$
Under regularity conditions, and when $\gamma$ is strongly identified
it is well known from \citet{Stock2000}, \citet{Newey2004} and \citet{Kleibergen2005}
that $\textrm{AR}\left(\beta_{n,0}\right)$ has a limiting $\chi_{d-d_{\gamma}}^{2}$distribution.
When the AR statistic is formulated based on the LIML estimator for
linear models and iid observations with homoskedastic errors GKMC12
show that \eqref{eq:CUAR} is bounded by a random variable with a
limiting $\chi_{d-d_{\gamma}}^{2}$ distribution uniformly over sequences
of data distributions where $\beta$ and $\gamma$ may not be identified.
A nominal size $\alpha$ test of $H_{0}$ then is based on rejecting
the null when 
\[
\textrm{AR}_{\textrm{C}}\left(\beta_{n,0}\right)>c_{1-\alpha,\chi_{d-d_{\gamma}}^{2}}
\]
where $c_{1-\alpha,\chi_{d-d_{\gamma}}^{2}}$ is the $1-\alpha$ quantile
of the $\chi_{d-d_{\gamma}}^{2}$ distribution. Based on Anderson
and Rubin (1949) we can construct a confidence region with coverage
rate of at least $1-\alpha$ by inverting the AR statistic 
\[
\textrm{CI}=\left\{ \beta\in\mathbb{R}^{d_{\beta}}|\mathit{\textrm{AR}_{\textrm{C}}}\left(\beta\right)\leq c_{1-\alpha,\chi_{d-d_{\gamma}}^{2}}\right\} .
\]
The confidence set CI is the set of all values $\beta$ for which
the AR statistic does not reject $H_{0}$. AR based confidence sets
were analyzed in the context of weakly identified or non-identified
linear IV methods by \citet{Dufour1997}, and for weakly identified
instrumental variables by \citet{Staiger1997} and \citet{Stock2000}.
GKMC12 are the first to show uniform validity of the subset AR test
over all sequences of linear models, including weakly identified,
partially identified and unidentified cases.

A key ingredient to uniform validity of the AR statistic in more general
settings than linear models with homoskedastic errors is that the
statistic is based on the CUE. The CUE implicitly orthogonalizes the
moment functions $\hat{g}_{n}\left(\theta\right)$ from the influence
of estimating nuisance parameters $\gamma$. Orthogonalization lies
at the core of the $C\left(\alpha\right)$ test of \citet{Neyman1959}
and was shown in \citet{Donald2000} to be a feature of the first
order conditions of the CUE. Let $\hat{\gamma}_{n}=\hat{\gamma}_{n}\left(\beta_{n,0}\right)=\argmin_{\gamma}Q_{n}\left(\beta_{n,0},\gamma\right).$
If $\hat{\gamma}_{n}$ is not unique then any value for $\gamma$
that minimizes $Q_{n}\left(\beta_{n,0},\gamma\right)$ can be chosen.
We show that the only effect of estimating $\gamma$ by $\hat{\gamma}_{n}$
on the asymptotic distribution $\textrm{AR}_{\textrm{C}}\left(\beta_{n,0}\right)=Q_{n}\left(\beta_{n,0},\hat{\gamma}_{n}\right)$
is a reduction in the degrees of freedom of the resulting $\chi^{2}$
distribution. This result does not hold when first stage estimators
for $\gamma$ such as two stage least squares or GMM are plugged into
$Q_{n}\left(\beta_{n,0},\gamma\right).$

\section{Theoretical Results\label{sec:Theory}}

We start by stating detailed regularity conditions we impose. Let
$\left(\Omega,\mathcal{F\mathit{,P}}\right)$ be a probability space.
Let $\chi_{n}=\left\{ \chi_{n,t}:t=1,2,...;n=1,2,...\right\} $ be
a double array of random vectors of dimension $d_{\chi}$. Let $y_{n,t}$
and $x_{n,t}$ be double arrays of random variables with elements
from $\chi_{n,t}.$ Let $w_{n,t}$ be a $d_{w}\times1$ vector and
$z_{n,t}$ a $d_{z}\times1$ vector of elements from $\chi_{n,t}$.
Assume that $x_{n,t}$ and $\text{\ensuremath{w_{n,t}} }$ contain
different elements, while $z_{n,t}$ may contain $x_{n,t}$ and $\text{\ensuremath{w_{n,t}} }$
as well as additional variables. The exact configuration of $z_{n,t}$
depends on whether $x_{n,t}$ and $\text{\ensuremath{w_{n,t}} }$
are believed to be endogenous or exogenous. Let $\lambda_{min}\left(A\right)$
denote the smallest eigenvalue of the matrix $A$, let $\left\Vert A\right\Vert _{2}=\tr\left(A'A\right)^{1/2}$
be the Frobenius norm and $\left\Vert A\right\Vert =\sup_{\left\Vert x\right\Vert =1}\left\Vert Ax\right\Vert $
be the operator norm where $\left\Vert .\right\Vert $ applied to
a vector in $\mathbb{R}^{k}$ is the usual Euclidean norm. Let $\chi_{n}:\Omega\rightarrow\mathbb{R}^{\infty}$
and denote by $\chi_{n}$$\rightsquigarrow$$\chi$ weak convergence
of $\chi_{n}$ to a random element $\chi$ in $\mathbb{R}^{\infty}$,
\citeauthor{vanderVaartWellner1996} (1996, Definition 1.3.3).

Let $\mathscr{X}$ be a collection of measurable processes $\chi$.
Each $\chi\in\mathscr{X}$ induces a probability distribution $P_{\chi}$.
We consider sequences of $\chi_{n}=\left(\chi_{n,1},\chi_{n,2},...,\chi_{n,n},...\right)$
of processes $\chi_{n}\in\mathscr{X}$ that contain subsequences that
converge to a limiting random element $\chi$ in $\mathbb{R}^{\infty}$.
This is formalized in Assumption \ref{assu:Asy_Tightness}.
\begin{assumption}
\label{assu:Asy_Tightness}\textup{Every sequence $\chi_{n}\in\mathscr{X}$
for $n\geq1$ is measurable and asymptotically tight.}\footnote{A sequence $\chi_{n}$ is asymptotically tight if for every $\varepsilon>0$
there exists a compact $K$ with $\delta$-enlargement $K^{\delta}=\left\{ \eta\in\mathbb{R}^{\infty}|d\left(\eta,K\right)<\delta\right\} $
for a metric $d\left(.,.\right)$ on $\mathbb{R}^{\infty}$ such that
$\liminf_{n}P\left(\chi_{n}\in K^{\delta}\right)\geq1-\varepsilon\text{ }$
for every $\delta>0$, see \citeauthor{vanderVaartWellner1996} (1996,
p.21). Note that inner probability is replaced by $P$ because we
only consider measurable sequences.}
\end{assumption}
\phantom{.}

Assumption \ref{assu:Asy_Tightness} and Prohorov's theorem (see \citeauthor{vanderVaartWellner1996},
1996, Theorem 1.3.9) imply that for every subsequence $n_{k}$ there
is a further subsequence of $n_{k}$ denoted by $n_{k_{l}}$ such
that $\chi_{n_{k_{l}}}$$\rightsquigarrow$$\chi$ and where the limit
$\chi$ may depend on the particular subsequence. In our notation
we suppress the dependence of limits $\chi$ on the subsequence. For
ease of notation we omit references to specific subsequences and refer
to converging subsubsequences $n_{k_{l}}$ simply as converging subsequences
$n_{k}$. Assumption \ref{assu:Asy_Tightness} implies that all finite
dimensional distributions of $\chi_{n_{k}}$ converge along converging
subsequences. To the extent that moments of $\chi_{n}$ exist, moments
related to finite dimensional distributions also converge along converging
subsequences. Assumption \ref{assu:Asy_Tightness} neither implies
that ordinary limits (as opposed to subsequential limits) exist, nor
does it rule out the existence of such ordinary limits. See \citet{Andrews2020}
for a related assumption that makes the dependence of the limit $\chi$
on the subsequence explicit.

Assumption \ref{assu:Asy_Tightness} is a nonparametric generalization
of semiparametric assumptions in GKMC12 and \citet{Andrews2020} and
is similar to assumptions imposed in \citet{Andrews2019}. These papers
define classes of data generating processes and sequences taking values
in these model classes that converge to limiting data generating processes
as the sample size increases. Special cases of such classes can be
found in the weak instrument literature that follows the approach
of \citet{Staiger1997} and where only particular sequences leading
to so called weak instrument limits are considered. In order to analyze
uniformity, the setting in GKMC12 is significantly more general as
it accounts for completely unidentified, partially identified as well
as fully identified data generating processes.

We do not assume that $\chi_{n}$ is generated by a particular parametric
or semiparametric model. We only assume that under the null hypothesis
a set of finite dimensional moment restrictions parameterized by a
set of finite dimensional parameters hold. The null hypothesis is
specified in terms of restrictions on a subset of the parameter vector.
We do not specifically consider the case where the null hypothesis
restricts the entire parameter vector since in that scenario the AR
statistic trivially has uniformly correct size under our assumptions
and results. The next assumption postulates that the moment conditions
\eqref{eq:MomentCond} hold along parameter sequences under the null
hypothesis.
\begin{assumption}
\label{assu:H0}Let $\theta=\left(\beta',\gamma'\right)'$ and let
$\Theta\subseteq\mathbb{R}^{d_{\gamma}+d_{\beta}}$ be the parameter
space. Assume that $\Theta$ is compact. Let $\textrm{int}\Theta$
be the interior of $\Theta.$ There exists a random function $g_{t}\left(.\right):\mathbb{R}^{d_{\beta}+d_{\gamma}}\rightarrow\mathbb{R}^{d}$,
two times continuously differentiable with respect to $\gamma$ and
for each $n\geq1$ a vector $\theta_{n,0}=\left(\beta_{n,0}',\gamma_{n,0}'\right)'\in\textrm{int}\Theta$
with $H_{0}:\beta=\beta_{n,0}$ such that
\begin{equation}
n^{-1/2}\sum_{t=1}^{n}E\left[g_{t}\left(\beta_{n,0},\gamma_{n,0}\right)\right]=0.\label{eq:Mom_Cond}
\end{equation}
Along converging subsequences $n_{k}$ assume that $\lim_{k\rightarrow\infty}\theta_{n_{k},0}=\theta_{0}$
for some $\theta_{0}$\textup{$\in\textrm{int}\Theta$ }and where
$\theta_{0}$ may depend on the subsequence.
\end{assumption}
In applications the null hypothesis is usually formulated for a fixed
constant $\beta_{0}$ such that $H_{0}:\beta=\beta_{0}$. Our theory
covers a fixed null as a special case. Without assuming homogeneity
or stationarity $\gamma_{n,0}$ typically varies with $n$ even if
$\beta_{n,0}=\beta_{0}$ is fixed. Note that $\theta_{n,0}$ is not
necessarily unique under our assumptions and in general there are
infinitely many solutions to the moment condition. Our theory is not
intended to directly test for violations of Assumption \ref{assu:H0}
which imposes both functional form restrictions summarized by $g_{t}\left(.\right)$
as well as a restriction on the parameter $\beta_{0}.$ In other words,
and in line with existing theory, a test of $H_{0}$ is interpreted
as testing $\beta=\beta_{0}$ while taking for granted that the moment
condition in \eqref{eq:Mom_Cond} is satisfied for some $\theta.$
However, since our theoretical analysis is only conducted under the
null we do not take a stand on whether \eqref{eq:Mom_Cond} holds
for some $\theta$ under the alternative.
\begin{assumption}
\label{assu:Omega_Conv}Let $\hat{\Omega}_{n}\left(\theta\right)$
be defined in \eqref{eq:Omega_hat}, $\Omega_{n}\left(\theta\right)=n^{-1}\sum_{t=1}^{n}\sum_{s=1}^{n}E\left[g_{t}\left(\beta,\gamma\right)g_{s}\left(\beta,\gamma\right)'\right]$.
For all converging subsequences $n_{k}$, a sequence of bandwidth
parameters $a_{n}$ and a sequence $\hat{\theta}_{n}$ such that\textup{
$\hat{\theta}_{n_{k}}\rightarrow_{p}\theta_{0}$} there exists a symmetric
matrix $\Omega\left(\theta\right)$ of continuous functions of $\theta$
with $\inf_{\theta\in\Theta}\lambda_{\min}\left(\Omega\left(\theta\right)\right)\geq K_{\Omega}>0$
such that \\
i) \textup{$\Omega_{n_{k}}\left(\theta_{n_{k},0}\right)-\Omega\left(\theta_{0}\right)\rightarrow0$
and $\hat{\Omega}_{n_{k}}\left(\theta_{n_{k},0}\right)-\Omega_{n_{k}}\left(\theta_{0}\right)\rightarrow_{p}0$}\\
ii) \textup{$\Omega_{n_{k}}\left(\hat{\theta}_{n_{k}}\right)-\Omega\left(\theta_{0}\right)\rightarrow_{p}0$
and $\hat{\Omega}_{n_{k}}\left(\hat{\theta}_{n_{k}}\right)-\Omega_{n_{k}}\left(\theta_{0}\right)\rightarrow_{p}0$.}
\end{assumption}
\phantom{.}

The full rank assumption for $\Omega\left(\theta\right)$ corresponds
to similar assumptions in \citet{Stock2000} and \citet{Kleibergen2005}.
The assumption has been relaxed by \citet{Andrews2019} for a modified
version of the AR statistic. The convergence of $\hat{\Omega}_{n_{k}}\left(\theta\right)\rightarrow_{p}\Omega\left(\theta\right)$
was shown for triangular near epoch dependent arrays defined on mixing
sequences by \citet{Jong2000}. The limits in Assumption \ref{assu:Omega_Conv}
may depend on the specific subsequence. Additional regularity conditions
may guarantee that convergence obtains for $n_{k}=n$, a scenario
that is a special case under our assumptions.

When $g_{t}\left(.\right)$ is of the form \eqref{eq:GMM_Moments}
and the data are stationary a necessary condition for Assumption \ref{assu:Omega_Conv}
is that for all $a\in\mathbb{R}^{d_{z}}$, $\left\Vert a\right\Vert =1$
and all $n\geq1$ the variables $z_{n,t}$ are not colinear in the
sense that 
\begin{equation}
\sup_{t}\Pr\left(z_{n,t}'a=0\right)<1.\label{eq:Prob_za}
\end{equation}
 To see this, assume that this condition does not hold. Then, for
some $a\in\mathbb{R}^{d_{z}}$ it follows that $z_{n,t}'a=0\text{ a.s.}$
for all $n$ and $t.$ But this implies that $g_{t}\left(\beta,\gamma\right)\left(I_{q}\otimes a\right)=0\text{ a.s.}$
for all $n$ and $t$ which in turn implies that Assumption \ref{assu:Omega_Conv}
cannot hold. In addition, \eqref{eq:Prob_za} implies that $\lambda_{min}\left(E\left[z_{n,t}z_{n,t}'\right]\right)>0$
and $\lambda_{min}\left(n^{-1}\sum_{t=1}^{n}E\left[z_{n,t}z_{n,t}'\right]\right)>0.$

Using the pseudo true parameter value $\theta_{n,0}$ we now define
the pseudo residual 
\begin{equation}
\varepsilon_{n,t}=g_{t}\left(\beta_{n,0},\gamma_{n,0}\right)\label{eq:epsilon_nt_h}
\end{equation}
where $\varepsilon_{n,t}$ is a $d\times1$ vector of random variables
and where by Assumption \ref{assu:H0} it follows $n^{-1/2}\sum_{t=1}^{n}E\left[\varepsilon_{n,t}\right]=0$.
Also define $d\times d_{\gamma}$ dimensional matrix of random functions
$g_{t}^{\gamma}\left(\beta,\gamma\right)=\partial g_{t}\left(\beta,\gamma\right)/\partial\gamma'$
and let 
\begin{equation}
\hat{g}_{n,\gamma}\left(\gamma\right)=n^{-1}\sum_{t=1}^{n}g_{t}^{\gamma}\left(\beta_{n,0},\gamma\right)\label{eq:g_gamma(gamma)_hat}
\end{equation}
 and 
\begin{equation}
\hat{g}_{n,\gamma}=\hat{g}_{n,\gamma}\left(\gamma_{n,0}\right).\label{eq:g_gamma_hat}
\end{equation}
The behavior of the derivative $g_{t}^{\gamma}\left(\beta,\gamma\right)$
determines whether $\gamma_{n,0}$ is identified. \citet{Stock2000}
consider moment functions that additively separate into a component
that only depends on $\gamma$ and assume that the derivative of the
moment function with respect to that component is full rank. Similarly,
\citet{Kleibergen2005} imposes a full column rank assumption on a
term equivalent to $\hat{g}_{n,\gamma}$ for the subset version of
his test. Instead, we allow for the derivatives $\hat{g}_{n,\gamma}$
to be of reduced column rank in the limit. We Impose the following
assumptions on $\hat{g}_{n,\gamma}.$
\begin{assumption}
\label{assu:g_gamma}For all converging subsequences $n_{k}$ assume
that there exists a sequence of non-stochastic matrices $g_{n,\gamma}=E\left[\hat{g}_{n,\gamma}\right]$
and some matrix $g_{\gamma}$ that depends on $n_{k}$ such that $g_{n_{k},\gamma}-g_{\gamma}\rightarrow0$
and $\hat{g}_{n_{k},\gamma}-g_{n_{k},\gamma}\rightarrow_{p}0$. In
addition, for any $\check{\gamma}_{n}$ such that $\check{\gamma}_{n_{k}}-\gamma_{n_{k},0}=O_{p}\left(n_{k}^{-1/2}\right)$
along converging subsequences $n_{k}$ it holds that \textup{$\hat{g}_{n_{k},\gamma}\left(\check{\gamma}_{n_{k}}\right)-g_{n_{k},\gamma}\rightarrow_{p}0.$}
\end{assumption}
Define the $d_{\gamma}\times d_{\gamma}$ non-stochastic matrix $\Gamma=g_{\gamma}'\Omega_{0}^{-1}g_{\gamma}$
with eigenvalues $\Delta_{1}\geq\Delta_{2}...\geq\Delta_{d_{\gamma}}\geq0$.
Assumption \ref{assu:g_gamma} does not restrict the column rank of
$g_{\gamma}$. Thus, we allow the rank $r$ of $\Gamma$ to take values
$0\leq r\leq d_{\gamma}$ and where in general $\Delta_{j}$ and $r$
depend on the subsequence $n_{k}.$

Define $\Pi{}_{n,\gamma}=n^{-1}\sum_{t=1}^{n}E\left[g_{t}^{\gamma}\left(\beta_{n,0},\gamma_{n,0}\right)'\right]$
where $\Pi_{n,\gamma}$ is a $d\times d_{\gamma}$ dimensional matrix.
Define $V_{n,t}$ as a $d\times d_{\gamma}$ dimensional matrix of
pseudo-residuals that can be written as
\begin{equation}
V_{n,t}=g_{t}^{\gamma}\left(\beta_{n,0},\gamma_{n,0}\right)-\Pi_{n,\gamma}\label{eq:Vnt}
\end{equation}
such that $\sum_{t=1}^{n}E\left[V_{n,t}\right]=0$ by construction.
As in GKMC12 no other restrictions are imposed on $\Pi_{n,\gamma}$
so that scenarios where $\Pi_{n,\gamma}$ is rank deficient, or even
$\Pi_{n,\gamma}=0$ are allowed. The properties of $\Pi_{n,\gamma}$
depend on the data distribution $P_{\chi}$ . Rank deficiency of $\Pi_{n,\gamma}$
arises when $\gamma$ is not identified, partially or weakly identified.

The analysis of the CUE also involves derivatives of $\hat{\Omega}_{n_{k}}\left(\theta\right)$
given by 
\begin{equation}
\frac{\partial Q_{n}\left(\theta\right)}{\partial\gamma'}=2n\hat{g}_{n}\left(\theta\right)'\hat{\Omega}_{n,0}^{-1}\left(\hat{g}_{n,\gamma}\left(\gamma\right)-\left(I_{d}\otimes\hat{g}_{n}\left(\theta\right)'\hat{\Omega}_{n,0}^{-1}\right)\hat{\Lambda}_{n}\left(\gamma\right)\right)\label{eq:dQ_dgamma}
\end{equation}
 where 
\begin{equation}
\hat{\Lambda}_{n}\left(\gamma\right)=\partial\vec\hat{\Omega}_{n}\left(\beta_{n,0},\gamma\right)/\partial\gamma'.\label{eq:Lambda_hat}
\end{equation}
 To guarantee convergence of $\hat{\Lambda}_{n}\left(\gamma\right)$
define the matrices 
\begin{equation}
\hat{\Lambda}_{z,n,0}=n^{-1}\sum_{t=1}^{n}\sum_{s=1}^{n}k\left(\frac{t-s}{a_{n}}\right)\left(I_{d}\otimes\varepsilon_{n,t}\right)\label{eq:Lambda_hat_z}
\end{equation}
and 
\begin{equation}
\hat{\Lambda}_{V,n,0}=n^{-1}\sum_{t=1}^{n}\sum_{s=1}^{n}k\left(\frac{t-s}{a_{n}}\right)\left[V_{n,t}\otimes\varepsilon_{n,t}\right]\label{eq:Lambda_hat_V}
\end{equation}
and impose the following restrictions on $\hat{\Lambda}_{z,n,0}$
and $\hat{\Lambda}_{V,n,0}$.
\begin{assumption}
\label{assu:Lambda}Let $\hat{\Lambda}_{z,n,0}$ and $\hat{\Lambda}_{V,n,0}$
be defined in \eqref{eq:Lambda_hat_z} and \eqref{eq:Lambda_hat_V},
and let $\Lambda_{z,n,0}=0$, and $\Lambda_{V,n,0}=n^{-1}\sum_{t=1}^{n}\sum_{s=1}^{n}E\left[V_{n,t}\otimes\varepsilon_{n,t}\right].$
For all converging subsequences $n_{k}$ as well as a sequence of
bandwidth parameters $a_{n}$ it follows that $\hat{\Lambda}_{z,n_{k},0}-\Lambda_{z,n_{k},0}\rightarrow_{p}0$
and $\hat{\Lambda}_{V,n_{k},0}-\Lambda_{V,n_{k},0}\rightarrow_{p}0$.
Let $\Lambda_{V,0}$ be matrices with bounded elements that may depend
on the subsquence $n_{k}.$ Assume that $\Lambda_{V,n_{k},0}\rightarrow\Lambda_{V,0}$.
\end{assumption}
\phantom{.}We now consider the process
\[
\omega_{n,t}=\left(\varepsilon_{n,t}',\vec\left(V_{n,t}\right)'\right)'
\]
where $\omega_{n,t}$ is a $d\left(1+d_{\gamma}\right)$ vector of
double arrays of random variables. Let $S_{n}=\frac{1}{\sqrt{n}}\sum_{t=1}^{n}\omega_{n,t}.$
Note that $E\left[S_{n}\right]=0$ by \eqref{eq:epsilon_nt_h}, \eqref{eq:Vnt}
and Assumption \ref{assu:H0}, such that, without loss of generality,
$\omega_{n,t}$ can be assumed to be mean zero. The next assumption
imposes that $\nu_{n,t}$ has a non-degenerate limiting distribution
along all converging subsequences.
\begin{assumption}
\label{assu:CLT}For all $n\geq1$ : $E\left[\text{\ensuremath{\left\Vert S_{n}\right\Vert _{2}^{2}}}\right]<\infty$,
and for $\Sigma_{n}=E\left[S_{n}S_{n}'\right]$, there is some $b_{\Sigma}$
such that $\lambda_{min}\left(\Sigma_{n}\right)>b_{\Sigma}>0$ for
all $n\geq1$. In addition, for all converging subsequences $n_{k}$
it holds that $\Sigma_{n_{k}}^{-1/2}S_{n_{k}}\rightarrow_{d}N\left(0,I\right)$
and $\Sigma_{n_{k}}\rightarrow\Sigma$ where $\Sigma$ is a positive
definite symmetric matrix with non-random elements that may depend
on the subsequence.
\end{assumption}
Assumptions \eqref{assu:Omega_Conv}, \eqref{assu:g_gamma}, \eqref{assu:Lambda}
and \eqref{assu:CLT} are high level assumptions. They could be formulated
more directly by imposing restrictions on the moments and the distribution
of $\chi_{n}$. Such assumptions then could be used to invoke laws
of large numbers and central limit theorems implying Assumptions \eqref{assu:Omega_Conv}-\eqref{assu:CLT}.
We do not follow this approach here because low level assumptions
tend to be more specific to a particular testing context and obscure
the fundamental requirements for our results. However, we consider
two specific examples in Section \ref{sec:Examples} were we discuss
a set of low level assumptions that can be imposed to guarantee that
Assumptions \eqref{assu:Omega_Conv}, \eqref{assu:g_gamma}, \eqref{assu:Lambda}
and \eqref{assu:CLT} hold.

Our main theoretical result is now stated. Recall the definition of
the AR statistic $\textrm{AR}_{\textrm{C}}\left(\beta_{n,0}\right)$
in \eqref{eq:CUAR}. The asymptotic size of the subset AR test is
defined as
\begin{equation}
\textrm{AsySz}_{\alpha}=\limsup_{n\rightarrow\infty}\sup_{\chi_{n}\in\mathscr{X}}P_{\chi_{n}}\left(\text{\ensuremath{\textrm{AR}_{\textrm{C}}\left(\beta_{n,0}\right)}}>c_{1-\alpha,\chi_{\left(d-d_{\gamma}\right)}^{2}}\right)\label{eq:AsySz}
\end{equation}
where $P_{\chi_{n}}$ is the induced probability measure of the process
$\chi_{n}$ on the underlying probability space and where each $\chi_{n}$
implies a corresponding sequences of parameters $\theta_{n,0}$ under
$H_{0}$. We prove the following result which extends Theorem 1 of
GKMC12 to the case of the AR statistic based on the CUE rather than
LIML and that allows for nonlinear moment conditions and non-stationary
or heterogenous processes with conditional heteroskedasticity.
\begin{thm}
\label{thm:AsySz}Assume $0<\alpha<1$ and that Assumptions \ref{assu:Asy_Tightness}-\ref{assu:CLT}
hold for all \textup{$\chi_{n}\in\mathscr{X}$ and all $n\geq1$}.
Then the asymptotic size of the $\textrm{AR}_{\textrm{C}}$ statistic
defined in (\ref{eq:AsySz}) satisfies 
\[
\text{\ensuremath{\textrm{AsySz}_{\alpha}}}=\alpha.
\]
\end{thm}
The challenge in proving Theorem \ref{thm:AsySz} under the maintained
assumptions of this paper is that the CUE $\hat{\gamma}_{n}$ of $\gamma$
is not necessarily unique and may not converge to a well defined limit.
Then, conventional approximation arguments for the first order conditions
of the CUE fail. The proof of Theorem \ref{thm:AsySz} is based on
a novel approach of regularizing the first order condition. We then
perturb the pseudo true value $\gamma_{n,0}$ toward a value $\tilde{\gamma}_{n,0}$
in the direction of the regularized first order conditions. This is
done in a way that keeps $\tilde{\gamma}_{n,0}$ close enough to $\gamma_{n,0}$
for the central limit theorem of Assumption \eqref{assu:CLT} to apply
while at the same time approximately satisfying the first order conditions.
An immediate consequence of the definition of $\textrm{AR}_{\textrm{C}}\left(\beta_{n,0}\right)$
is that all sequences $\gamma_{n}$ produce upper bounds because $\textrm{argmin}_{\gamma}Q\left(\beta_{n,0},\gamma\right)\leq Q\left(\beta_{n,0},\gamma_{n}\right)$
for any sequence $\gamma_{n}$. While $\gamma_{n,0}$ leads to one
such upper bound, the bound based on $\gamma_{n,0}$ is too large
in the sense that $Q\left(\beta_{n,0},\gamma_{n,0}\right)$ has a
limiting $\chi_{d}^{2}$ rather than a $\chi_{d-d_{\gamma}}^{2}$
distribution. This occurs because $\gamma_{n,0}$ does not satisfy
the first order conditions well enough even in an asymptotic sense.
The perturbed point $\tilde{\gamma}_{n,0}$ is constructed in such
a way that it approximately solves the CUE moment conditions evaluated
at the infeasible sequence $\theta_{n,0}$ under $H_{0}:\beta=\beta_{n,0}$
and that as a result $Q\left(\beta_{n,0},\tilde{\gamma}_{n,0}\right)$
has the desired $\chi_{d-d_{\gamma}}^{2}$ distribution in the limit
irrespective of whether $\gamma$ is identified or not. We stress
that $\tilde{\gamma}_{n,0}$ is a construct exclusively used in the
proofs and not needed to compute the test statistic $\textrm{AR}_{\textrm{C}}\left(\beta_{n,0}\right)$
itself.

\citet{Donald2000} show that the CUE moment conditions remove the
highest order bias term of a conventional GMM estimator by showing
that the CUE moment conditions are centered at zero. Their insight
implies an asymptotic orthogonality condition between the limiting
process of the moment function $\hat{g}_{n}$ and the residual of
the projection of $\hat{g}_{n,\gamma}$ onto $\text{\ensuremath{\hat{g}_{n,0}}.}$
The same insight underlies \citet{Kleibergen2005} but only for the
case where $\gamma$ is strongly identified. To show uniform validity,
the asymptotic orthogonality needs to be established for all converging
subsequences irrespective of whether $\gamma$ is identified or not.
This then allows to characterize the upper bound of the $\textrm{AR}_{\textrm{C}}\left(\beta_{n,0}\right)$
statistic in terms of a $\chi_{d-d_{\gamma}}^{2}$ limiting distribution
along all converging subsequences. The limiting $\chi_{d-d_{\gamma}}^{2}$
distribution is obtained by representing the AR-statistic asymptotically
as the sum of squares of the residuals of a projection of the moment
conditions $\hat{g}_{n}$ onto the column space spanned by the residualized
$\hat{g}_{n,\gamma}$. The rank of this projection residual is $d-d_{\gamma}$
irrespective of whether $\gamma$ is identified or not.

We now explain the construction of the sequence $\tilde{\gamma}_{n,0}$
and the proof strategy behind Theorem \ref{thm:AsySz} in more detail.
Use \eqref{eq:gn_hat} to define $\hat{g}_{n}\left(\gamma\right)=\hat{g}_{n}\left(\beta_{n,0},\gamma\right)$,
$\hat{g}_{n,0}=\hat{g}_{n}\left(\gamma_{n,0}\right)$ and $\tilde{g}_{n,0}=\hat{g}_{n}\left(\tilde{\gamma}_{n,0}\right)$.
Recall the definition of $\hat{g}_{n,\gamma}$ in \eqref{eq:g_gamma_hat}
and use \eqref{eq:g_gamma(gamma)_hat} to define $\tilde{g}_{n,\gamma}=\hat{g}_{n,\gamma}\left(\tilde{\gamma}_{n,0}\right)$.
Also recall the definition of $\hat{\Lambda}_{n}\left(\gamma\right)$
in \eqref{eq:Lambda_hat} where $\hat{\Lambda}_{n}\left(\gamma\right)$
has dimension $d^{2}\times d_{\gamma}$. Now construct the matrix
\begin{equation}
\hat{A}_{n}=\hat{\Omega}_{n,0}^{-1/2}\left(\hat{g}_{n,\gamma}\left(\gamma_{n,0}\right)-\left(I_{d}\otimes\hat{g}_{n,0}'\hat{\Omega}_{n,0}^{-1}\right)\hat{\Lambda}_{n}\left(\gamma_{n,0}\right)\right)\label{eq:Def_A_hat}
\end{equation}
and consider the empirical moment conditions
\begin{equation}
0=\hat{A}_{n}'\hat{\Omega}_{n,0}^{-1/2}\hat{g}_{n}\left(\gamma\right).\label{eq:Pseudo_CUE_FOC}
\end{equation}
To gain some intuition for the moment condition \eqref{eq:Pseudo_CUE_FOC}
we focus on the case where $g_{t}\left(.\right)$ is of the form \eqref{eq:GMM_Moments}
with $d_{q}=1$ such that \eqref{eq:Pseudo_CUE_FOC} can be written
as 
\[
0=\left(\hat{g}_{n,\gamma}-\hat{g}_{n,0}'\hat{\Omega}_{n,0}^{-1}\hat{\Lambda}_{n,0}\right)'\hat{\Omega}_{n,0}^{-1}\hat{g}_{n}\left(\gamma\right).
\]
As argued in \citet{Donald2000}, $\hat{g}_{n,0}'\hat{\Omega}_{n,0}^{-1}\hat{\Lambda}_{n,0}$
is the projection of $\hat{g}_{n,\gamma}$ onto $\hat{g}_{n,0}$ in
the case of independently distributed data. The term $\hat{g}_{n,\gamma}-\hat{g}_{n,0}'\hat{\Omega}_{n,0}^{-1}\hat{\Lambda}_{n,0}$
then is the projection residual, which by construction is orthogonal
to $\hat{g}_{n,0}$. As shown by \citet{Donald2000}, the orthogonality
holds exactly in the iid setting and in expectation in a time series
framework. The use of long run variance-covariance matrices in the
time series case ensures that this interpretation remains valid in
the limit in the more general setting of this paper. By constructing
our perturbed moment vector $\tilde{g}_{n,0}$ to be close to $\hat{g}_{n,0}$
and approximately orthogonal to $\hat{A}_{n}$ we guarantee that the
moment vector $\hat{g}_{n,0}$ is stochastically independent of the
column space spanned by $\hat{A}_{n}$ at least in the limit.

The construction of $\tilde{\gamma}_{n,0}$ depends on the particular
subsequence $n_{k}$ such that $\tilde{\gamma}_{n,0}$ is only defined
for $n_{k}$ and may differ for different sequences $n_{k}$ and $n_{k'}$.
We start with the mean value expansion of $\hat{g}_{n}\left(\gamma\right)$
around $\gamma_{n,0}$, 
\begin{equation}
\hat{g}_{n}\left(\gamma\right)=\hat{g}_{n,0}+\hat{g}_{n,\gamma}\left(\check{\gamma}_{n,0}\right)\left(\gamma-\gamma_{n,0}\right)\label{eq:gamma_mve}
\end{equation}
with $\left\Vert \check{\gamma}_{n,0}-\gamma_{n,0}\right\Vert \leq\left\Vert \gamma-\gamma_{n,0}\right\Vert $
which upon substitution into the moment conditions \eqref{eq:Pseudo_CUE_FOC}
leads to 
\begin{equation}
0=\hat{A}_{n}'\hat{\Omega}_{n,0}^{-1/2}\left(\hat{g}_{n,0}-\hat{g}_{n,\gamma}\left(\check{\gamma}_{n,0}\right)\left(\gamma_{n,0}-\gamma\right)\right).\label{eq:CUE_FOC_Apprx}
\end{equation}

If $\hat{A}_{n}$ and $\hat{g}_{n,\gamma}$ were full column rank
at least for large enough samples, then one could solve \eqref{eq:CUE_FOC_Apprx}
for $\gamma.$ Under the sequences considered in this paper, both
matrices may have reduced ranks for finite $n$ and in the limit.
Solutions based on the Moore-Penrose inverse have delicate convergence
properties, see for example \citet{Wedin1973}. In fact, such solutions
often do not converge because the operator norm of the MP inverse
becomes unbounded as eigenvalues of its argument tend towards zero,
see \citeauthor{Stewart1977} (1977, Theorem 3.1). To stabilize these
solutions we adopt a regularization scheme called truncated singular
value decomposition (TSVD), see the numerical analysis literature
\citeauthor{Hanson1971} (1971), \citeauthor{Varah1973} (1973) and
\citet{Hansen1987}.

Let the singular value decomposition (SVD), which in this case coincides
with the spectral representation of $\hat{\Gamma}_{n}=\hat{g}_{n,\gamma}'\hat{\Omega}_{n,0}^{-1}\hat{g}_{n,\gamma}$
be equal to $\hat{\Gamma}_{n}=\hat{R}_{n}\hat{\Delta}_{n}\hat{R}'_{n}$
and where $\hat{\Delta}_{n}$ is a diagonal matrix of the eigenvalues
$\hat{\Delta}_{1,n}\geq....\geq\hat{\Delta}_{d_{\gamma},n}\geq0$
of $\hat{\Gamma}_{n}$. By Assumptions \ref{assu:Omega_Conv} and
\ref{assu:g_gamma} and for a converging subsequence $n_{k}$, $\hat{\Gamma}_{n_{k}}\rightarrow_{p}\Gamma$
such that by Theorem \ref{thm:Weyl}, $\hat{\Delta}_{j,n_{k}}\rightarrow_{p}\Delta_{j}\geq0$.
Assume that $\Delta_{j}>0$ for $j\leq r$ and some $0\leq r\leq d_{\gamma}$.
Fix $\varepsilon$ arbitrary with $\Delta_{r}>\varepsilon>0$. For
each converging subsequence $n_{k}$ define $\mathring{\Delta}_{n_{k}}$
as the matrix with diagonal elements $\text{\ensuremath{\mathring{\Delta}_{n_{k},j}}}$
such that for all $l\in\left\{ 1,...,d_{\gamma}\right\} ,$ 
\[
\text{\ensuremath{\mathring{\Delta}_{n_{k},l}}}=\begin{cases}
\hat{\Delta}_{n_{k},l} & \text{\text{if \ensuremath{\hat{\Delta}_{n_{k},l}>\varepsilon}}}\\
0 & \text{otherwise}
\end{cases}.
\]
The TSVD of $\hat{\Gamma}_{n_{k}}$ is now given by 
\begin{equation}
\mathring{\Gamma}_{n_{k}}=\hat{R}_{n_{k}}\mathring{\Delta}_{n_{k}}\hat{R}'_{n_{k}}.\label{eq:Gamma_dot}
\end{equation}

We construct an infeasible solution $\tilde{\gamma}_{n_{k},0}$ to
\eqref{eq:CUE_FOC_Apprx} with two properties. The solution $\tilde{\gamma}_{n,0}$
is a perturbation of $\gamma_{n,0}$ small enough to converge to $\gamma_{n,0}$
but such that it also satisfies \eqref{eq:CUE_FOC_Apprx} with sufficient
accuracy in large samples. The purpose of the construction is to evaluate
the CUE criterion function at a point that allows for a limiting distribution
and where the moment conditions of the CUE, and thus orthogonality
between $\hat{A}_{n_{k}}$ and $\tilde{g}_{n_{k},0}$ hold approximately.
The parameter $\tilde{\gamma}_{n_{k},0}$ is defined as a perturbation
of the parameters $\gamma_{n_{k},0}$ in the direction of the first
order conditions for the CUE estimator by setting
\begin{equation}
\tilde{\gamma}_{n_{k},0}=\gamma_{n_{k},0}-\mathring{\Gamma}_{n_{k}}^{+}\hat{A}_{n_{k}}'\hat{\Omega}_{n_{k},0}^{-1/2}\hat{g}_{n_{k},0}.\label{eq:gamma_til_TSVD}
\end{equation}
Note that the MP-inverse of $\mathring{\Gamma}_{n_{k}}$ denoted by
$\mathring{\Gamma}_{n_{k}}^{+}$ is continuous and therefore converges
along converging subsequences. Continuity of $\mathring{\Gamma}_{n_{k}}^{+}$
holds because the eigenvalues of $\mathring{\Gamma}_{n_{k}}$ are
bounded away from zero by $\varepsilon$ due to regularization. This
is in contrast to the MP-inverse of $\hat{\Gamma}_{n_{k}}$ which
may not be continuous and thus not converge along converging subsequences.
The following properties of $\tilde{\gamma}_{n_{k},0}$ can now be
established.
\begin{lem}
\label{lem:gamma-conv}Let $\tilde{\gamma}_{n,0}$ be defined in (\ref{eq:gamma_til_TSVD}).
Let Assumptions \ref{assu:Asy_Tightness}, \ref{assu:H0}, \ref{assu:Omega_Conv},
and \ref{assu:g_gamma} hold. Then,

(i) for all converging subsequences $n_{k}$ 
\[
\tilde{\gamma}_{n_{k},0}-\gamma_{n_{k},0}=O_{p}\left(n_{k}^{-1/2}\right).
\]

(ii) for all converging subsequences $n_{k}$ it follows that the
first order condition in \eqref{eq:Pseudo_CUE_FOC} evaluated at $\tilde{\gamma}_{n,0}$
satisfies
\begin{eqnarray*}
\hat{A}_{n_{k}}'\hat{\Omega}_{n_{k},0}^{-1/2}\hat{g}_{n_{k}}\left(\tilde{\gamma}_{n_{k},0}\right) & = & o_{p}\left(n_{k}^{-1/2}\right).
\end{eqnarray*}

(iii) for the projection $P_{\hat{A}_{n}}$ onto the column space
of $\hat{A}_{n}$ it follows that along converging subsequences $n_{k}$
\[
P_{\hat{A}_{n_{k}}}-\hat{A}_{n_{k}}\mathring{\Gamma}_{n_{k}}^{+}\hat{A}_{n_{k}}'=o_{p}\left(1\right).
\]
\end{lem}
A proof of Lemma \ref{lem:gamma-conv} is given in Appendix \ref{subsec:Gamma_tilde}.
Let $\tilde{\Omega}_{n,0}=\hat{\Omega}_{n}\left(\beta_{n,0},\tilde{\gamma}_{n,0}\right)$.
Using Lemma \ref{lem:gamma-conv} and denoting by $P_{\hat{A}_{n}}$
the projection onto the column space spanned by $\hat{A}_{n}$ and
letting $M_{\hat{A}_{n}}=I-P_{\hat{A}_{n}}$, the proof of Theorem
\ref{thm:AsySz} then proceeds to show in Lemma \ref{lem:LemA1} that
\[
\textrm{AR}_{\textrm{C}}\left(\beta_{n,0}\right)\leq n\tilde{g}_{n,0}'\tilde{\Omega}_{n,0}^{-1/2}M_{\hat{A}_{n}}\tilde{\Omega}_{n,0}^{-1/2}\tilde{g}_{n,0}+\varpi_{n},
\]
where along converging subsequences $n_{k}$ 
\[
\varpi_{n_{k}}=n_{k}\tilde{g}_{n_{k}}'\tilde{\Omega}_{n_{k}}^{-1/2}P_{\hat{A}_{n_{k}}}\tilde{\Omega}_{n_{k}}^{-1/2}\tilde{g}_{n_{k}}=o_{p}\left(1\right).
\]
This result holds because $\tilde{g}_{n,0}=\hat{g}_{n}\left(\tilde{\gamma}_{n,0}\right)$
approximately satisfies the orthogonality condition in Lemma \ref{lem:gamma-conv}(ii).
With the help of Lemma $\ref{lem:gamma-conv}$ and \eqref{eq:gamma_til_TSVD}
we show in Lemma \ref{lem:LemA1} that along converging subsequences
\[
n_{k}\tilde{g}_{n_{k},0}'\tilde{\Omega}_{n_{k},0}^{-1/2}M_{\hat{A}_{n_{k}}}\tilde{\Omega}_{n_{k},0}^{-1/2}\tilde{g}_{n_{k},0}=n_{k}\hat{g}_{n_{k},0}'\hat{\Omega}_{n_{k},0}^{-1/2}M_{\hat{A}_{n_{k}}}\hat{\Omega}_{n_{k},0}^{-1/2}\hat{g}_{n_{k},0}+o_{p}\left(1\right).
\]
By Assumption \ref{assu:CLT} $\hat{g}_{n_{k},0}$ converges to a
Gaussian limit variable $\omega_{\varepsilon}$. We show in Lemma
\ref{lem:LemA2} that along converging subsequences $P_{\hat{A}_{n_{k}}}\rightarrow_{d}P_{A}$
where $P_{A}$ has constant rank $d_{\gamma}$ independent of the
subsequential limit and where $P_{A}$ is a possibly random matrix
with elements that are stochastically independent of $\omega_{\varepsilon}.$
The last result is essentially due to Lemma \ref{lem:gamma-conv}(ii)
which is at the heart of establishing asymptotic uncorrelatedness
between the elements of $\hat{A}_{n_{k}}$ and $\hat{g}_{n_{k},0}$
along all converging subsequences. Combining these results, we establish
in Lemma \ref{lem:LemA3} that $\textrm{AR}_{\textrm{C}}\left(\beta_{n,0}\right)$
is bounded above by a random variable that is asymptotically $\chi_{d-d_{\gamma}}^{2}$
along all converging subsequences.

\section{Examples\label{sec:Examples}}

Our general framework in Section \ref{sec:Tests} and our theoretical
results in Section \ref{sec:Theory} do not assume that the data are
generated by a parametric model. Nevertheless, parametric models are
often the basis for empirical work and serve as useful benchmarks
that allow more precise interpretations of the parameters $\theta.$
In this section we consider two specific examples of such parametric
models and discuss how they fit into our general framework.

\subsection{Linear Instrumental Variables}

Consider as an example a simple univariate static model where $g_{t}\left(.\right)$
is of the form \ref{eq:GMM_Moments} with $d_{q}=1$ and $q\left(.\right)$
is linear. The data generating process is given by
\begin{equation}
y_{n,t}=\beta x_{n,t}+w_{n,t}'\gamma_{n}+\varepsilon_{n,t}\label{eq:Gen_y}
\end{equation}
and 
\begin{align}
x_{n,t} & =\pi_{n,x}'z_{n,t}+V_{n,x,t}\label{eq:Gen_x}\\
w_{n,t} & =\pi_{n,w}'z_{n,t}+V_{n,w,t}\label{eq:Gen_w}
\end{align}
where $z_{n,t},\varepsilon_{n,t}$ and $V_{n,t}=\left(V_{x,t}',V_{w,t}'\right)'$
are the basic inputs with $z_{n,t}$ observed and $\varepsilon_{n,t},V_{n,t}$
unobserved. We assume that $\left(z_{n,t},\varepsilon_{n,t},V_{n,t}'\right)$
are independent across $t$ and where $w_{n,t}$ has $d_{w}$ elements
and $z_{n,t}$ has $d_{z}>d_{w}$ elements and $x_{n,t}$ is univariate.
The parameter vector $\pi_{x,n}$ and the $d_{z}\times d_{w}$ matrix
$\pi_{w,n}$ control instrument strength. The model is the same as
in GKMC12 but without the assumption of homoskedasticity of $\varepsilon_{n,t}$
and $V_{n,t}$, nor do we assume that the random variables are identically
distributed over $t.$ At the expense of a slightly more restrictive
example, one could assume that $\left(z_{n,t},\varepsilon_{n,t},V_{n,t}'\right)$
and $\gamma_{n}$ do not depend on $n.$ This scenario would still
contain weak instrument limits as considered by Staiger and Stock
(1997) by allowing $\text{\ensuremath{\pi_{x,n}}}$ and $\text{\ensuremath{\pi_{w,n}}}$
to be in a $1/\sqrt{n}$ neighborhood local to zero, but would also
include non-identified scenarios where $\text{\ensuremath{\pi_{x,n}}}=0$
and/or $\pi_{w,n}=0.$ In the more general setting considered here
and in GKMC12 one can allow for $\cov\left(\varepsilon_{n,t},V_{n,t}\right)$
to depend on $n$, and thus add additional dimensions along which
inference may be challenging.

If $E\left[z_{n,t}\left(\varepsilon_{n,t},V_{n,t}'\right)\right]=0$
then the moment conditions $n^{-1}\sum_{t=1}^{n}E\left[\left(y_{n,t}-X'_{n,t}\theta\right)z_{n,t}\right]=0$
have at least one solution $\theta_{n}=\left(\beta,\gamma_{n}'\right)'$.
If $n^{-1}\sum_{t=1}^{n}E\left[X_{n,t}z_{n,t}'\right]$ is of reduced
column rank then the model is not identified and the average moment
condition may have multiple solutions.

The CUE for this model is formed from moment functions 
\begin{equation}
g_{t}\left(\beta,\gamma\right)=\left(\begin{array}{c}
y_{n,t}-x_{n,t}\beta-w'_{n,t}\gamma\end{array}\right)z_{n,t}\label{eq:lin-g}
\end{equation}
and can be evaluated at $H_{0}:\beta=\beta_{0}$ by setting $\hat{g}_{n}\left(\beta_{0},\gamma\right)=0$.
The weight matrix $\hat{\Omega}_{n}$ is exploiting independence and
is given by 
\begin{equation}
\hat{\Omega}_{n}\left(\beta,\gamma\right)=n^{-1}\sum_{t=1}^{n}g_{t}\left(\beta,\gamma\right)g_{t}\left(\beta,\gamma\right)'.\label{eq:lin-Omega}
\end{equation}
The criterion function $Q_{n}\left(\beta,\gamma\right)$ and the test
statistic $\textrm{AR}_{\textrm{C}}$ then are formulated in the same
way as in Section \ref{sec:Tests} using the expressions for $\hat{g}_{n}$
and $\hat{\Omega}_{n}$ given in (\ref{eq:lin-g}) and (\ref{eq:lin-Omega}).
The derivatives $E\left[g_{n,\gamma}\right]=n^{-1}\sum_{t=1}^{n}E\left[z_{n,t}w_{n,t}'\right]=n^{-1}\sum_{t=1}^{n}E\left[z_{n,t}z_{n,t}'\right]\pi_{n,w}$
are reduced column rank if $\pi_{n,w}$ is of reduced column rank.

Assumption \ref{assu:Asy_Tightness} is satisfied in the following
way. Assume $\eta_{n,t}=\left(z_{n,t}',\varepsilon_{n,t},V_{n,x,t},V_{n,w,t}'\right)'$
is a double array of independent random variables, where in addition
$z_{n,t}$ is independent of $\varepsilon_{n,t}$, $V_{n,x,t}$ and
$V_{n,w,t}$. Assume that $y_{n,t}$, $x_{n,t}$ and $w_{n,t}$ are
generated by \eqref{eq:Gen_y}, \eqref{eq:Gen_x} and \eqref{eq:Gen_w}
respectively for parameters $\theta,$ $\pi_{n,x}$ and $\pi_{n,w}$
in some compact parameter space $\Theta.$ When the DGP is linear
as in \eqref{eq:Gen_y}-\eqref{eq:Gen_w} then Assumptions \eqref{assu:Omega_Conv},
\eqref{assu:g_gamma} and \eqref{assu:Lambda} can be replaced with
the following assumption.
\begin{assumption}
\label{assu:Mom_W}For all $n\geq1$ and $t\geq1:$\textup{ $E\left[\left\Vert \eta_{n,t}\right\Vert _{2}^{2}\right]<\infty$}.
Let $W_{n,\eta\eta}=n^{-1}\sum_{t=1}^{n}\text{\ensuremath{E}\ensuremath{\left[\eta_{n,t}\eta_{n,t}'\right]}}$,
$\hat{W}_{n,\eta\eta}=n^{-1}\sum_{t=1}^{n}\text{\ensuremath{\eta_{n,t}\eta_{n,t}}}'$.
For all converging subsequences $n_{k}$ assume that $W_{n_{k},\eta\eta}-W_{\eta\eta}=o\left(1\right)$
where $W_{\eta\eta}$ is a symmetric full rank matrix of non-random
constants. In addition, $\hat{W}_{n_{k},\eta\eta}-W_{n_{k},\eta\eta}=o_{p}\left(1\right)$.
\end{assumption}
Assumption \eqref{assu:Mom_W} holds under the following sufficient
low level assumptions. Let $\eta_{n,t,j}$ be the $j$-th element
of $\eta_{n,t}$. Assume $\sup_{t}E\left[\left|\eta_{n,t,j}\right|^{2+\delta}\right]<\infty$
for some $0<\delta<1$ and all $j$ and all $n\geq1$. This implies
that the elements $\chi_{n,t,j}$ of $\chi_{n,t}$ also have bounded
$2+\delta$ moments. Also assume that there exist subsequences $n_{k}$
such that $n_{k}^{-1}\sum_{t=1}^{n_{k}}E\left[\chi_{n_{k},t,j}\chi_{n_{k},t,l}\right]$
converges to some constant $c$ that may depend on the subsequence
$n_{k}$. If $\chi_{n,t}$ does not depend on $n$ Theorem 1 in \citeauthor{Chow1997}
(1997, p.124) implies that Assumption \ref{assu:Mom_W} holds for
all such converging subsequences. A strong law for double arrays of
independent heterogenous $\chi_{n,t}$ is given in \citet{Teicher1985}.
Assumption \ref{assu:Omega_Conv} holds if for $\varepsilon_{n,t}\left(\theta\right)=y_{n,t}-x_{n,t}\beta-w'_{n,t}\gamma$
it follows that $\inf_{\theta}\det\left(E\left[\varepsilon_{n,t}\left(\theta\right)^{2}z_{n,t}z_{n,t}^{'}\right]\right)>0.$
Sufficient conditions are that $\sup_{\theta}\Pr\left(\varepsilon_{n,t}\left(\theta\right)=0\right)<1$
which holds if $y_{n,t}$ is not co-linear with $x_{n,t}$ and $w_{n,t}$
and $\sup_{a}\Pr\left(z_{n,t}'a=0\right)<1$. Assumption \eqref{assu:CLT}
holds if for any $\ell\in\mathbb{R}^{\left(2+d_{w}\right)d_{z}}$
and $\nu_{n,t}=\left(\varepsilon_{n,t},V_{n,x,t},V_{n,w,t}'\right)'\otimes z_{n,t}$
it follows that $E\left[\ell'\nu_{n,t}\right]=0$, $\sigma_{n,t}^{2}=E\left[\left(\ell'\nu_{n,t}\right)^{2}\right]>0$,
$c_{n}^{2}=\sum_{t=1}^{n}\sigma_{n,t}^{2}$ and $\lim{}_{k\rightarrow\infty}\sum_{t=1}^{n_{k}}E\left[\left(\ell'\nu_{n_{k},t}/c_{n_{k}}\right)^{2}1\left\{ \left|\ell'\nu_{n_{k},t}\right|>\varepsilon c_{n_{k}}\right\} \right]\rightarrow0$
for every converging subsequence $n_{k}$ by Theorem 9.6.1 of \citet{Dudley1989}.

\subsection{Local Projections}

We consider time series data with a non-parametric data generating
process and dependence structure. The parameter of interest is the
local projection of \citet{Jorda2005} for a univariate outcome variable
$y_{n,t+h}$ at horizons $h\in\left\{ 0,...,H\right\} $ onto a policy
or intervention variable $x_{n,t}$ and additional controls $w_{n,t}$
under heterogeneity and nonstationarity. To simplify the discussion
we assume that $x_{n,t}$ is univariate. We assume the researcher
is specifying a linear model under $H_{0}$ which is formulated as
a restriction on the conditional mean function $E\left[y_{n,t+h}|x_{n,t},w_{n,t}\right]=x_{n,t}\beta_{t}\left(h\right)+w'_{n,t}\gamma_{t}\left(h\right)$
and where $\beta_{t}\left(h\right)$ for $h\in\left\{ 0,...,H\right\} $
are the parameters of interest. As in Section \ref{sec:Tests} we
only assume that a linear moment restriction holds under $H_{0}$
but we do not assume that $y_{n,t}$ is necessarily generated by a
linear model. The conditioning variables $w_{n,t}$ summarize information
$\mathcal{I}_{t-1}$ up to time $t-1$, which may also contain instrumental
variables necessary to identify structural parameters or to give the
impulse coefficient a causal interpretation.

When $x_{n,t}$ depends on unobserved components that are potentially
correlated with unobserved factors affecting the outcome variable,
identification is achieved with an instrumental variable $z_{n,t}$.
In addition, in settings without a recursive structure some of the
variables in $w_{n,t}$ may also be endogenous. Note that in the case
of selection on observables, $z_{n,t}$ may be identical to $x_{n,t}$
and $w_{n,t}$.

Let $d=H+1,$ $Y_{n,t}=\left(y_{n,t},....,y_{n,t+H}\right)$, $\beta_{n}=\left(\beta_{n}\left(0\right),....,\beta_{n}\left(H\right)\right)'$
and $\gamma_{n}=\left(\gamma_{n}\left(0\right)',....,\gamma_{n}\left(H\right)'\right)'$.
Let $X_{n,t}=\left(x_{n,t},w_{n,t}'\right)'$, 
\begin{equation}
q\left(Y_{n,t},X_{n,t},\beta,\gamma\right)=\left(Y_{n,t}-\left(I_{d}\otimes x_{n,t}\right)\beta-\left(I_{d}\otimes w'_{n,t}\right)\gamma\right)\label{eq:LP_q}
\end{equation}
and $g_{t}\left(\beta,\gamma\right)=q\left(Y_{n,t},X_{n,t},\beta,\gamma\right)\otimes z_{n,t}.$
If the underlying process generating the observed data is weakly stationary
$\beta_{t}\left(h\right)$ does not depend on $t$. In the absence
of stationarity, we define $\beta_{n,0},\gamma_{n,0}$ as the solution
to 
\begin{equation}
n^{-1}\sum_{t=1}^{n}E\left[g_{t}\left(\beta_{n,0},\gamma_{n,0}\right)\right]=0\label{eq:GMM}
\end{equation}
for any fixed $n\geq1$. It is then possible to define pseudo residuals
$\varepsilon_{n,t}:=Y_{t}-\left(I_{d}\otimes x_{n,t}\right)\beta_{n,0}-\left(I_{d}\otimes w'_{n,t}\right)\gamma_{n,0}$
which inherit the property $\sum_{t=1}^{n}E\left[\varepsilon_{n,t}\left(h\right)z_{n,t}\right]=0$
from \eqref{eq:GMM} under $H_{0}$. With $\hat{g}_{n}\left(\beta,\gamma\right)$
as defined in \eqref{eq:gn_hat} and $\hat{\Omega}_{n}\left(\beta,\gamma\right)$
as defined in \eqref{eq:Omega_hat} the statistic $\textrm{A\ensuremath{\textrm{R}_{C}}}\left(\beta_{n,0}\right)$
is defined in \eqref{eq:CUAR}.

Assumption \eqref{assu:Asy_Tightness} is satisfied for many weakly
dependent processes. Assumption \ref{assu:Omega_Conv} can be satisfied
by imposing the conditions of \citet{Jong2000} which include L$_{2}$
near-epoch dependent triangluar arrays. Assumptions \ref{assu:g_gamma}
and \ref{assu:Lambda} can be replaced by Assumption \ref{assu:Mom_W}
due to linearity of \ref{eq:LP_q}. Assumption \ref{assu:Mom_W} holds
for $L_{1}$-mixingale arrays and $L_{1}$-near epoch dependent arrays
by results for weak laws of large numbers in \citet{Andrews1988}.
De Jong (1996)\citet{Jong1996} proves strong laws for $L_{q}$-mixingale
arrays for $q\geq1.$ The multivariate central limit theorem in Assumption
\ref{assu:CLT} can be established with the help of the Cramer-Wold
theorem and by verifying regularity conditions for $\lambda'\Sigma_{n}^{-1/2}\nu_{n,t}$
for all $\lambda$ with $\left\Vert \lambda\right\Vert =1$. The assumption
that $E\left[\text{\ensuremath{\left\Vert S_{n}\right\Vert ^{2}}}\right]<\infty$
is not sufficient in itself to guarantee a CLT without additional
restrictions on dependence. Low level conditions for univariate double
array processes can be found in \citet{McLeish1977} who imposes a
weak dependence condition in the form of a mixingale assumption. McLeish's
theory allows for non-stationarity, while trending behavior, including
unit roots are ruled out by the assumptions on dependence. \citet{Wooldridge1988}
prove a related result under near epoch dependence, which implies
a mixingale property (see Proposition 2.9 of \citeauthor{Wooldridge1988},
1988). They also prove a multivariate version, directly leading to
the condition in Assumption \ref{assu:CLT}. For non-stationary triangular
arrays instead of double arrays, \citet{Jong1997} proves a central
limit theorem directly for heterogenous $L_{2}$-mixingale arrays.
The result is obtained by showing that a blocking scheme can be used
to obtain a martingale approximation for which a martingale CLT can
be applied. \citet{Jong1997} and \citet{Davidson2000} establish
a functional CLT that is based on \citet{Jong1997}. \citet{Rio1997}
establishes a CLT extending Lindeberg's CLT for iid random variables
to double arrays of strongly mixing random variables. An alternative
approach is explored by \citet{Neumann2013} and Merlevede, Peligrad
and Utev \citeyearpar{Merlevede2019} who impose a maximal correlation
measure to control dependence. When $\nu_{n,t}$ is a martingale difference
array, \citet{McLeish1974}, see also \citet{HallHeyde1980}, provides
a central limit theorem.

\subsection{Implementation}

Practitioners are often more interested in confidence intervals for
a single parameter than in confidence sets for the entire parameter
vector. One can obtain such confidence intervals from our confidence
set via the projection method, see e.g., \citet{Dufour1997} and \citet{Dufour2005}.

When $\beta$ is a scalar ($d_{\beta}=1$), which is often the case
in empirical applications of LP-IV, a $100(1-\alpha)$\% confidence
interval for $\beta$ is given by $[\underline{\beta}_{1-\alpha},\overline{\beta}_{1-\alpha}]$,
where $\underline{\beta}_{1-\alpha}=\min\left\{ \beta|AR_{c}(\beta)\leq c_{1-\alpha}\right\} $,
$\overline{\beta}_{1-\alpha}=\max\left\{ \beta|AR_{c}(\beta)\leq c_{1-\alpha}\right\} $,
and $c_{1-\alpha}$ is the $1-\alpha$ quantile of the $\chi^{2}$
distribution with $d-d_{\gamma}$ degrees of freedom. If the confidence
set is convex, this confidence interval has exact $(1-\alpha)$ asymptotic
coverage probability. Otherwise, it is conservative.

When $d_{\beta}>1$, let $\tilde{c}_{1-\alpha}$ be the $1-\alpha$
quantile of the $\chi^{2}$ distribution with $d-d_{\beta}-d_{\gamma}+1$
degrees of freedom. A $100(1-\alpha)$\% confidence interval for $\beta_{1}$
(without of loss of generality) is given by $[\underline{\beta}_{1,1-\alpha},\overline{\beta}_{1,1-\alpha}]$,
where for some $\beta_{2:d_{\beta}}$ 
\begin{eqnarray}
\underline{\beta}_{1,1-\alpha} & = & \min_{\beta_{1}}\left\{ \beta_{1}|AR_{c}(\beta_{1},\beta_{2:d_{\beta}})\leq\tilde{c}_{1-\alpha}\right\} ,\\
\overline{\beta}_{1,1-\alpha} & = & \max_{\beta_{1}}\left\{ \beta_{1}|AR_{c}(\beta_{1},\beta_{2:d_{\beta}})\leq\tilde{c}_{1-\alpha}\right\} ,
\end{eqnarray}
and $\beta_{2:d_{\beta}}$ denotes the $(d_{\beta}-1)\times1$ subvector
of $\beta$ excluding $\beta_{1}$. In practice, the condition, $AR_{c}(\beta_{1},\beta_{2:d_{\beta}})\leq\tilde{c}_{1-\alpha}$,
can be checked by $\min_{\beta_{2:d_{\beta}},\gamma}Q_{n}(\beta,\beta_{2:d_{\beta}},\gamma)\leq\tilde{c}_{-1-\alpha}$
for a given value of $\beta_{1}$, where the minimum on the left-hand
side is achieved by numerical optimization methods. For example, to
obtain $\underline{\beta}_{1,1-\alpha}$ {[}$\overline{\beta}_{1,1-\alpha}${]},
one minimizes $\beta_{1}$ {[}resp. -$\beta_{1}${]} subject to $\min_{\beta_{2:d_{\beta}},\gamma}Q_{n}([\beta_{1}\;\beta_{2:d_{\beta}}],\gamma)\leq\tilde{c}_{1-\alpha}$.
In MATLAB, one can use \textsf{fminunc} to evaluate the constraint
function and use \textsf{fmincon} to minimize $\beta_{1}$ {[}$-\beta_{1}${]}.

When $d_{\beta}>1$, these confidence intervals are conservative in
that the asymptotic coverage probability is greater than or equal
to $1-\alpha$.

\section{Monte Carlo Experiments\label{sec:Monte-Carlo-Experiments}}

The data generating process is based on the bivariate VAR model of
inflation ($\pi_{t}$) and output gap ($x_{t}$) in equations (4)
and (7) of \citet{Kleibergen2009}: 
\begin{eqnarray}
x_{t} & = & \rho_{1}x_{t-1}+\rho_{2}x_{t-2}+\nu_{t},\\
\pi_{t} & = & (\alpha_{0}\rho_{1}+\alpha_{1})x_{t-1}+\alpha_{0}\rho_{2}x_{t-2}+\eta_{t},
\end{eqnarray}
where $\rho_{1}=\rho(1-\rho_{2})=0.9(1-\rho_{2})$, $\alpha_{0}=\lambda/[1-\gamma_{f}(\rho_{1}+\gamma_{f}\rho_{2})]$,
$\alpha_{1}=\lambda\gamma_{f}\rho_{2}/[1-\gamma_{f}(\rho_{1}+\gamma_{f}\rho_{2})]$,
$\lambda=\rho_{\eta\nu}\frac{\sigma_{\eta}}{\sigma_{\nu}}\left(1-\gamma_{f}\left(\rho_{1}+\gamma_{f}\rho_{2}\right)\right)$.
We introduce conditional heteroskedasticity to allow the conditional
variance of the moment conditions to depend on the instruments: $\eta_{t}=\exp(0.5h_{1t})\epsilon_{1t}/\kappa$,
$\nu_{t}=\exp(0.5h_{2t})\epsilon_{2t}/\kappa$, $\epsilon_{1t}$ and
$\epsilon_{2t}$ are iid $N(0,1)$ with correlation $\rho_{\eta\nu}$,
, $\kappa$ is set so that $\textrm{Var}(\eta_{t})=\sigma_{\eta}^{2}=1$
and $\textrm{Var}(\nu_{t})=\sigma_{\nu}^{2}=1$, and $h_{1t}$ and
$h_{2t}$ are two independent AR(1) processes: 
\begin{eqnarray}
h_{1t} & = & 0.9h_{1,t-1}+\xi_{1,t},\\
h_{2t} & = & 0.9h_{2,t-1}+\xi_{2,t},
\end{eqnarray}
and $\xi_{1t}$ and $\xi_{2t}$ are independent iid $N(0,0.2)$ random
variables that are independent of $\epsilon_{1t}$ and $\epsilon_{2t}$.

We estimate the Phillips curve, 
\begin{equation}
\pi_{t}\;=\;\lambda x_{t}+\gamma_{f}E_{t}(\pi_{t+1})+u_{t},
\end{equation}
using the moment conditions, 
\begin{equation}
E\left[z_{t}(\pi_{t}-\lambda x_{t}-\gamma_{f}\pi_{t+1})\right]\;=\;0,
\end{equation}
We consider two choices of instruments: (i) $z_{t}=[\pi_{t-1}\;x_{t-1}\;\pi_{t-2}\;x_{t-2}'\;\pi_{t-3}\;x_{t-3}]'$
as in \citet{Kleibergen2009} and (ii) $z_{t}=[x_{t}\;x_{t-1}]$ which
represents a just identified specification.

It follows from equation 6 and other equations in \citet{Kleibergen2009}
that the population reduced-form equations for the RHS endogenous
variables can be written as 
\begin{equation}
\left[\begin{array}{c}
\pi_{t+1}\\
x_{t}
\end{array}\right]\;=\;\left[\begin{array}{cc}
\alpha_{0}((\rho_{1}+\rho_{2}\gamma_{f})\rho_{1}+\rho_{2}) & \alpha_{0}(\rho_{1}+\rho_{2}\gamma_{f})\rho_{2}\\
\rho_{1} & \rho_{2}
\end{array}\right]\left[\begin{array}{c}
x_{t-1}\\
x_{t-2}
\end{array}\right]+\left[\begin{array}{c}
(\alpha_{1}+\alpha_{0}\rho_{1})\nu_{t}+\eta_{t+1}\\
\nu_{t}
\end{array}\right].\label{eq:reduced-form equation}
\end{equation}
There are at least two cases of identification failures. First, when
$\rho_{2}=0$, the projection coefficient matrix simplifies to 
\begin{equation}
\left[\begin{array}{cc}
\rho^{2}\rho_{\eta\nu} & 0\\
\rho & 0
\end{array}\right].
\end{equation}
Thus, $\gamma_{f}$ and $\lambda$ are only partially identified (i.e.,
their linear combination is identified). Second, when $\gamma_{\eta\nu}=0$,
the matrix simplifies to 
\begin{equation}
\left[\begin{array}{cc}
0 & 0\\
\rho(1-\rho_{2}) & \rho_{2}
\end{array}\right].
\end{equation}
In this case, $\gamma_{f}$ is identified but $\lambda$ is not.

We consider $\rho_{2}\in\{-0.05,-0.65,-0.99\}$ and $\textrm{Corr}(\eta_{t},\nu_{t})\in\{0.2,0.99\}$.
The number of Monte Carlo simulations is set to 5,000. For HAC estimation,
the lag truncation parameter is set to zero, as the moment conditions
are serially uncorrelated.

Tables \ref{tab:Table1} and \ref{tab:Table2} report the finite-sample
size of the $t$ test, the subset KLM test and the subset $AR_{c}$
test when testing $\gamma_{f}=0.5$ at the 5\% significance level.
Table 1 shows that, when the three lags of $\pi_{t}$ and $x_{t}$
are used as instruments, the conventional $t$ test suffers from size
distortions especially when the instruments are weak and the degree
of endogeneity is high ($\rho_{2}=-0.05$ and $\rho_{\epsilon\eta}=0.99$)
However, when the parameters are just identified, the size distortions
of the conventional $t$ test are much smaller, as shown in Table
2. In either case, the $AR_{c}$ and $KLM$ tests do not suffer from
size distortions as expected.

Figures \ref{fig:Fig1}(a)-(f) and Figures \ref{fig:Fig2}(a)-(f)
show the power curve of these three tests when $z_{t}=[\pi_{t-1}\;x_{t-1}\;\pi_{t-2}\;x_{t-2}\;\pi_{t-3}\;x_{t-3}]'$
and $z_{t}=[x_{t-1}\;x_{t-2}]'$, respectively. The sample size is
set to 1,000. When the instruments are weak ($\rho_{2}=-0.05$ or
$\rho_{\epsilon\eta}=0.0$), the $\textrm{AR}_{c}$ and KLM tests
have correct size, although they do not have any power, as shown in
Figures \ref{fig:Fig1}(a) and (d) and Figures \ref{fig:Fig2}(a)
and (d). \footnote{Interestingly, when $\rho_{\epsilon\eta}=0.0$, the $AR_{c}$ and
KLM tests do not have power when testing $\gamma_{f}=0.5$.} As the instruments become stronger, the tests become more powerful,
as demonstrated in Figures \ref{fig:Fig1}(b),(c) and (d),(f) and
Figures \ref{fig:Fig2} (b),(c) and (d),(f). When the parameters are
overidentified and are very strongly identified, the KLM test is more
powerful than the $\textrm{AR}_{c}$ test, seeFigures \ref{fig:Fig1}(c)
and (f). However, when the parameter is mildly strongly identified,
the $\textrm{AR}_{c}$ test is more powerful than the KLM test when
the true parameter value is close to the null value, see Figures Figures
\ref{fig:Fig1}(b) and (d). When the parameters are just identified,
their power curves are virtually identical, as seen in Figures Figures
\ref{fig:Fig2}(a) through (f). These results are consistent with
\citet{Moreira2009} who shows that in the linear Gaussian case the
just identified $\textrm{AR}_{c}$ based tests as well as score tests
are uniformly most powerful for certain testing problems.

\section{Conclusions}

We obtain a theoretical upper bound for the Anderson Rubin subset
statistic based on the continuous updating estimator for general parametric
non-linear moment restrictions with heteroskedasticity, dependence
and heterogeneity. The upper bound is shown to be valid for all converging
subsequences of data distributions irrespective of whether nuisance
parameters are identified or not. The upper bound is sharp in the
sense that the $\textrm{AR}_{c}$ attains it when nuisance parameters
are identified. We establish that the upper bound has a limiting $\chi_{d-d_{\gamma}}^{2}$
distribution along all converging subsequences. These results imply
that the $\textrm{AR}_{c}$ statistic has uniformly correct size and
can be used to construct confidence intervals that are uniformly valid.

\newpage{}

\begin{table}
\caption{\label{tab:Table1}Finite-sample size (5\% nominal size): $z_{t}=[\pi_{t-1}\;x_{t-1}\;\pi_{t-2}\;x_{t-2}\;\pi_{t-3}\;x_{t-3}]^{\prime}$}

\vspace*{1em}
\begin{tabular}{ll|lll|lll|lll}
\hline 
 &  & \multicolumn{3}{c|}{$\rho_{\epsilon\eta}=0.0$} & \multicolumn{3}{c|}{$\rho_{\epsilon\eta}=0.2$} & \multicolumn{3}{c}{$\rho_{\epsilon\eta}=0.99$}\tabularnewline
\hline 
$T$ & $\rho_{2}$ & $t$-test & $AR_{c}$ & KLM & $t$-test & $AR_{c}$ & KLM & $t$-test & $AR_{c}$ & KLM\tabularnewline
\hline 
100 & 0.0 & 0.148 & 0.041 & 0.022 & 0.158 & 0.043 & 0.023 & 0.970 & 0.033 & 0.075\tabularnewline
100 & -0.05 & 0.139 & 0.041 & 0.020 & 0.155 & 0.043 & 0.017 & 0.970 & 0.034 & 0.075\tabularnewline
100 & -0.65 & 0.099 & 0.034 & 0.016 & 0.102 & 0.033 & 0.017 & 0.172 & 0.030 & 0.026\tabularnewline
\hline 
500 & 0.0 & 0.124 & 0.058 & 0.017 & 0.131 & 0.053 & 0.018 & 0.978 & 0.046 & 0.043\tabularnewline
500 & -0.05 & 0.123 & 0.057 & 0.016 & 0.140 & 0.056 & 0.022 & 0.971 & 0.049 & 0.047\tabularnewline
500 & -0.65 & 0.086 & 0.038 & 0.014 & 0.066 & 0.041 & 0.012 & 0.070 & 0.050 & 0.042\tabularnewline
\hline 
\end{tabular}
\end{table}
\begin{table}
\caption{\label{tab:Table2}Finite-sample size (5\% nominal size): $z_{t}=[x_{t-1}\;x_{t-2}]^{\prime}$}

\vspace*{1em}
\begin{tabular}{ll|lll|lll|lll}
\hline 
 &  & \multicolumn{3}{c|}{$\rho_{\epsilon\eta}=0.0$} & \multicolumn{3}{c|}{$\rho_{\epsilon\eta}=0.2$} & \multicolumn{3}{c}{$\rho_{\epsilon\eta}=0.99$}\tabularnewline
\hline 
$T$ & $\rho_{2}$ & $t$-test & $AR_{c}$ & KLM & $t$-test & $AR_{c}$ & KLM & $t$-test & $AR_{c}$ & KLM\tabularnewline
\hline 
100 & 0.0 & 0.014 & 0.044 & 0.044 & 0.016 & 0.060 & 0.060 & 0.450 & 0.044 & 0.044\tabularnewline
100 & -0.05 & 0.012 & 0.043 & 0.043 & 0.014 & 0.054 & 0.054 & 0.329 & 0.050 & 0.050\tabularnewline
100 & -0.65 & 0.001 & 0.010 & 0.010 & 0.005 & 0.013 & 0.013 & 0.056 & 0.044 & 0.044\tabularnewline
\hline 
500 & 0.0 & 0.010 & 0.047 & 0.047 & 0.013 & 0.067 & 0.067 & 0.453 & 0.051 & 0.051\tabularnewline
500 & -0.05 & 0.010 & 0.046 & 0.046 & 0.011 & 0.063 & 0.063 & 0.185 & 0.054 & 0.054\tabularnewline
500 & -0.65 & 0.000 & 0.007 & 0.007 & 0.012 & 0.011 & 0.011 & 0.050 & 0.050 & 0.050\tabularnewline
\hline 
\end{tabular}
\end{table}
\begin{figure}
\centering %
\begin{tabular}{cc}
(a) $\rho_{2}=-0.05$ & (d) $\rho_{\eta\nu}=0.0$\tabularnewline
\includegraphics[width=65mm]{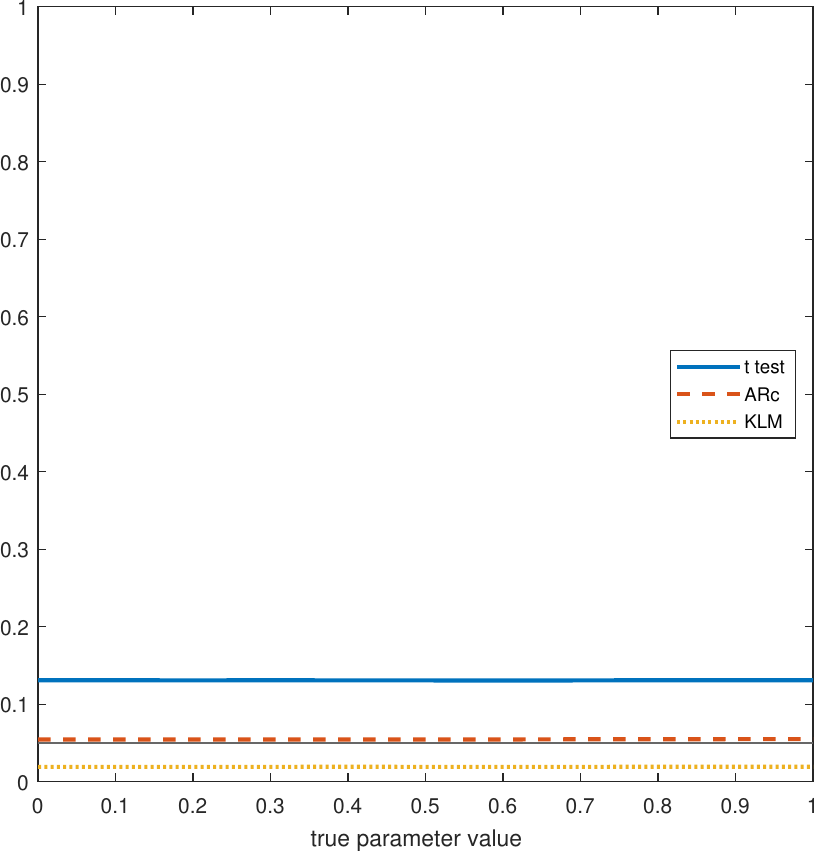} & \includegraphics[width=65mm]{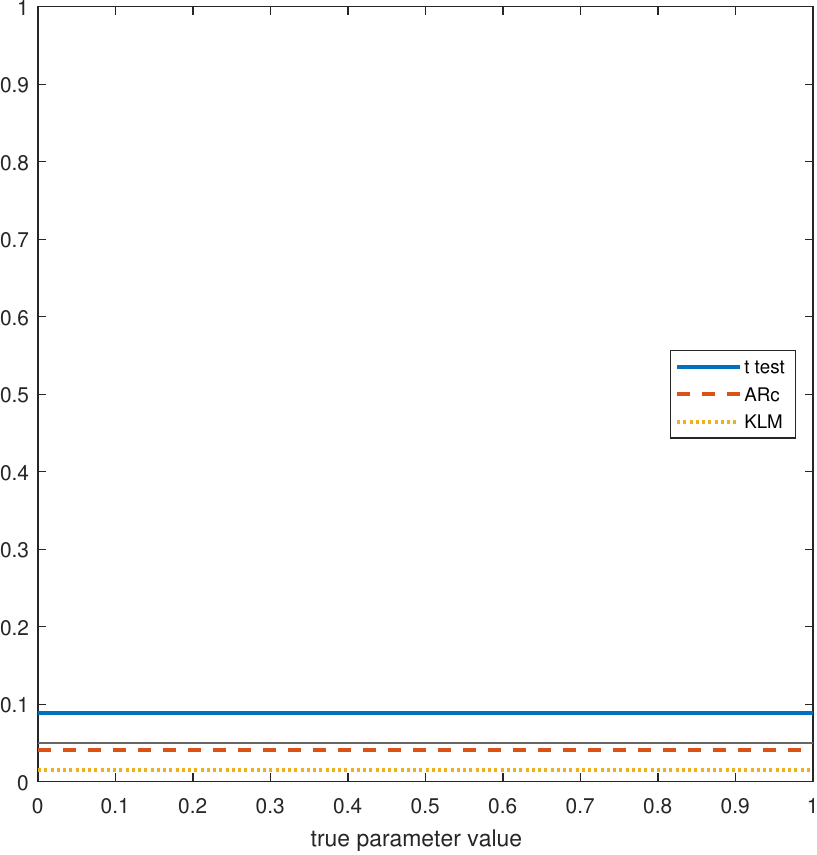}\tabularnewline
(b) $\rho_{2}=-0.65$ & (e) $\rho_{\eta\nu}=0.2$\tabularnewline
\includegraphics[width=65mm]{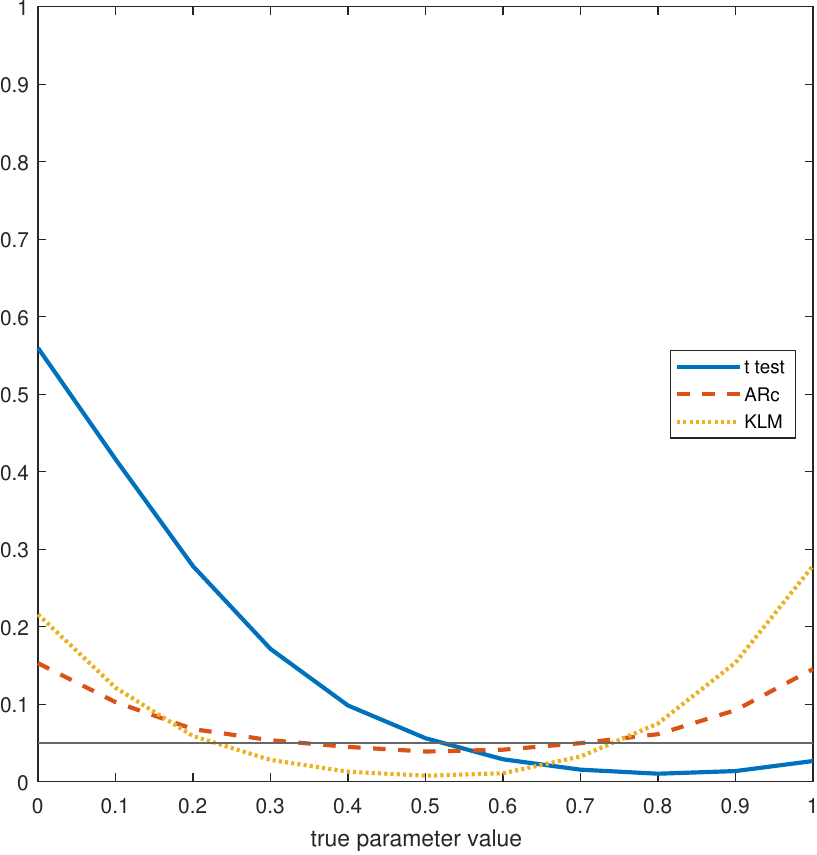} & \includegraphics[width=65mm]{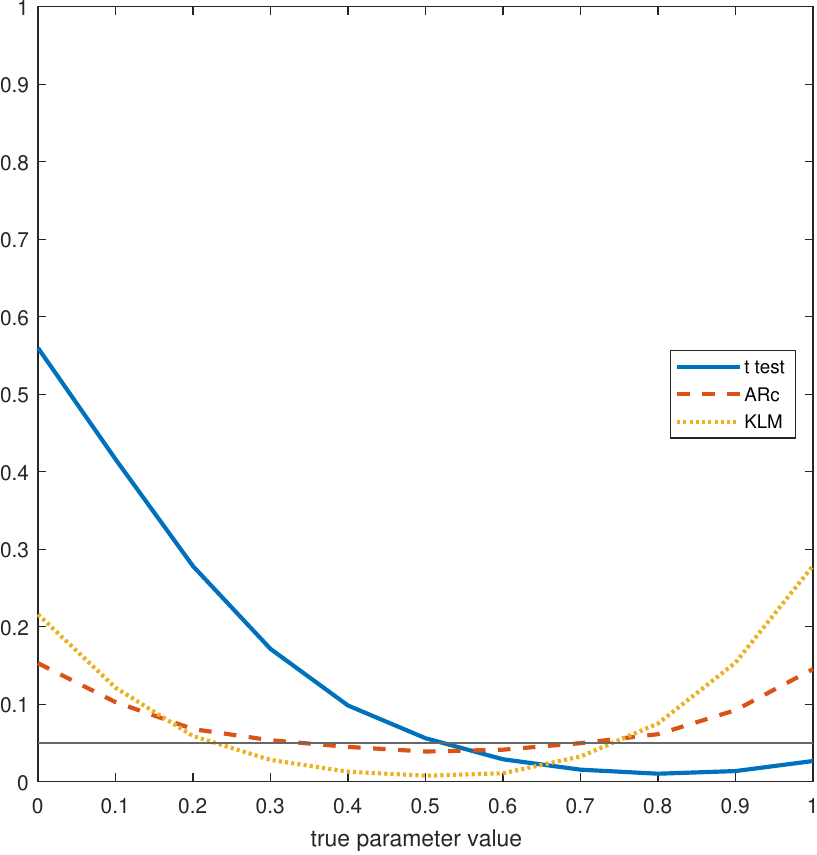}\tabularnewline
(c) $\rho_{2}=-0.99$ & (f) $\rho_{\eta\nu}=0.99$\tabularnewline
\includegraphics[width=65mm]{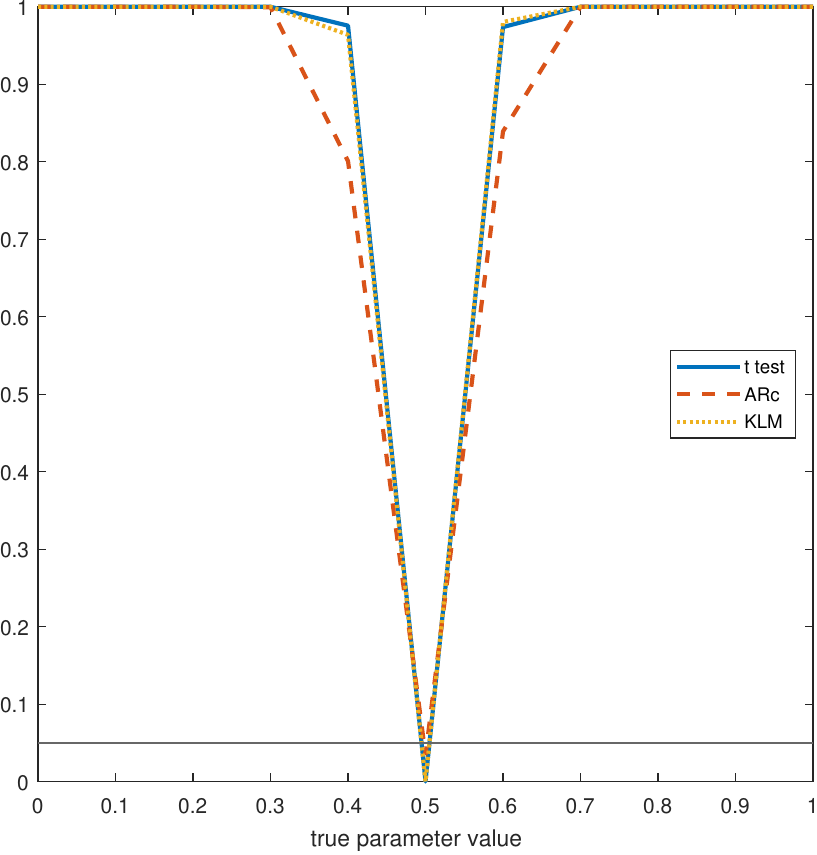} & \includegraphics[width=65mm]{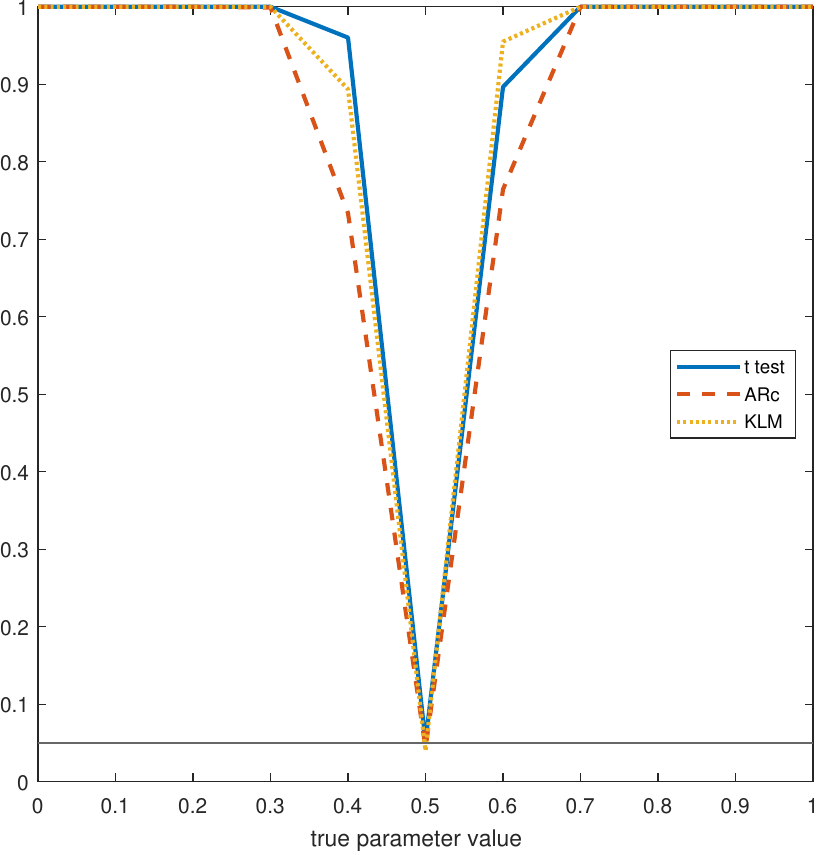}\tabularnewline
\end{tabular}\caption{\label{fig:Fig1}(a)-(f): $z_{t}=[\pi_{t-1}\;x_{t-1}\;\pi_{t-2}\;x_{t-2}\;\pi_{t-3}\;x_{t-3}]'$;
(a)-(c): $\rho_{\eta\nu}=0.2$, (d)-(f): $\rho_{2}=-0.65$}
\end{figure}
\begin{figure}
\centering %
\begin{tabular}{cc}
(a) $\rho_{2}=-0.05$ & (d) $\rho_{\eta\nu}=0.0$\tabularnewline
\includegraphics[width=65mm]{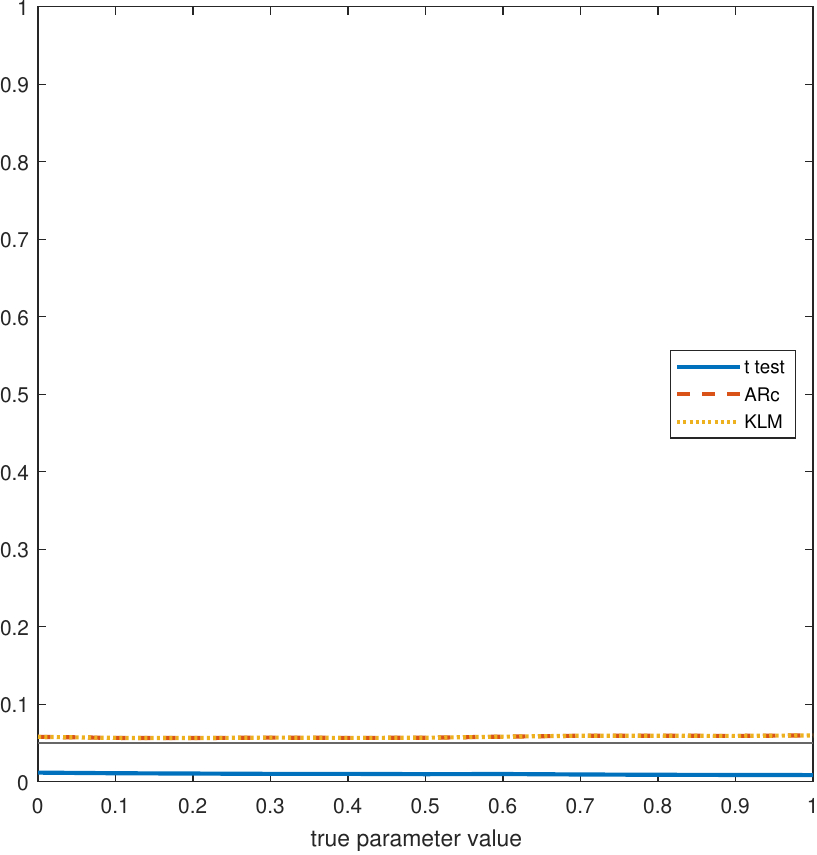} & \includegraphics[width=65mm]{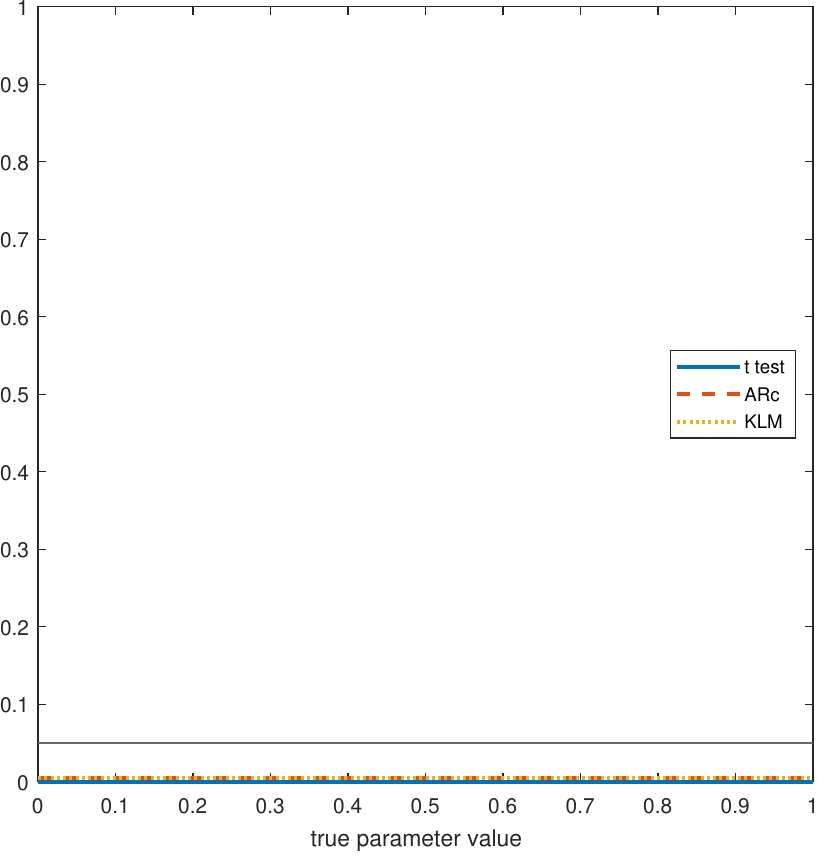}\tabularnewline
(b) $\rho_{2}=-0.65$ & (e) $\rho_{\eta\nu}=0.2$\tabularnewline
\includegraphics[width=65mm]{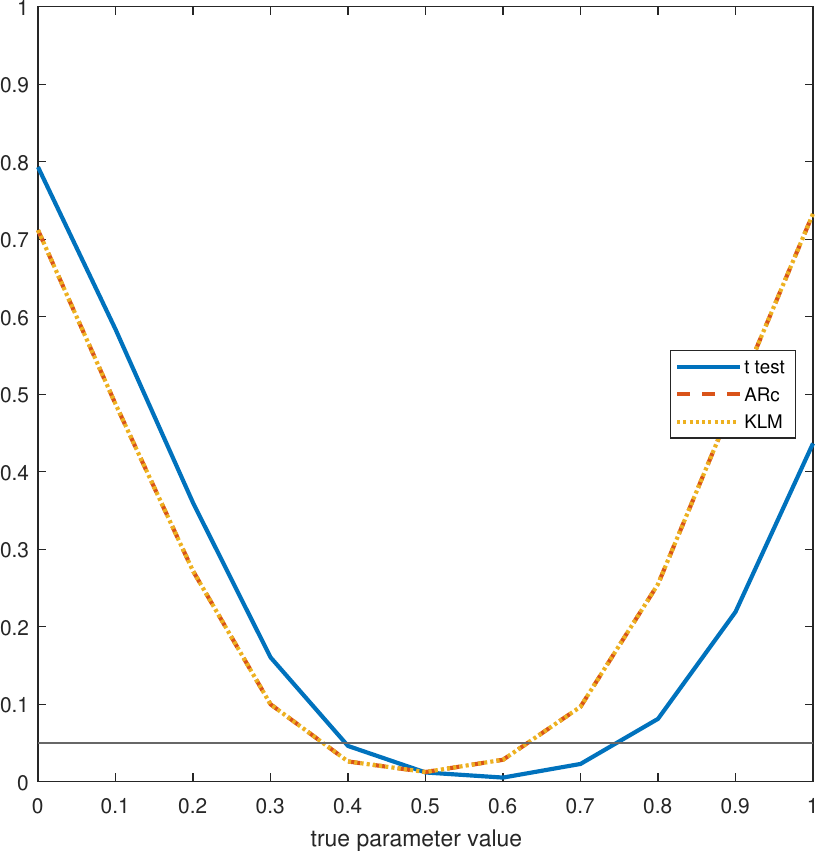} & \includegraphics[width=65mm]{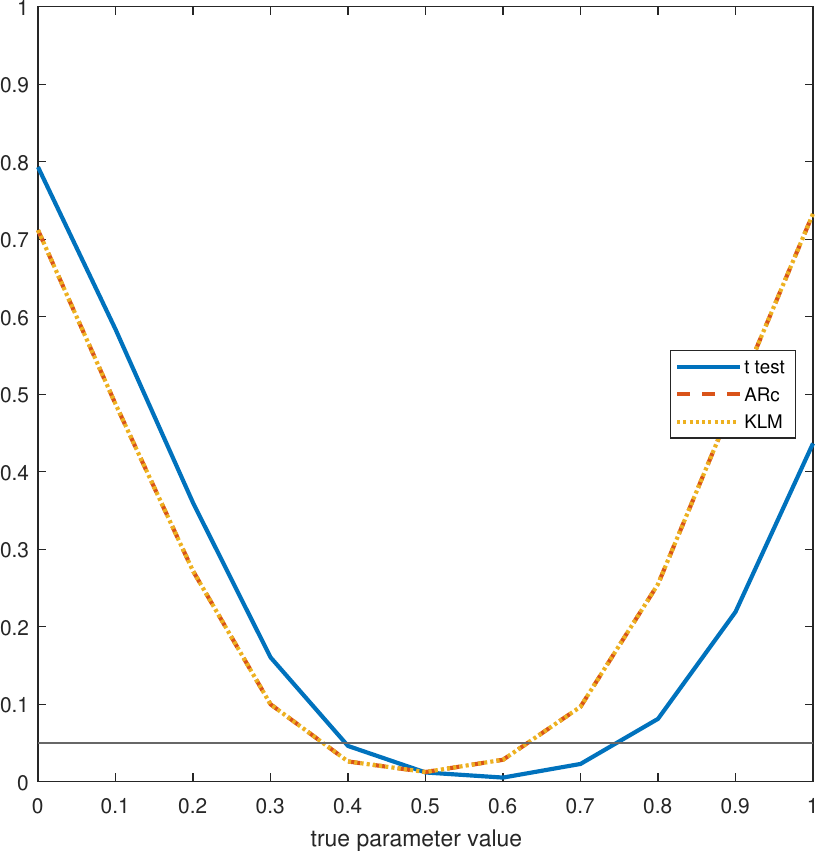}\tabularnewline
(c) $\rho_{2}=-0.99$ & (f) $\rho_{\eta\nu}=0.99$\tabularnewline
\includegraphics[width=65mm]{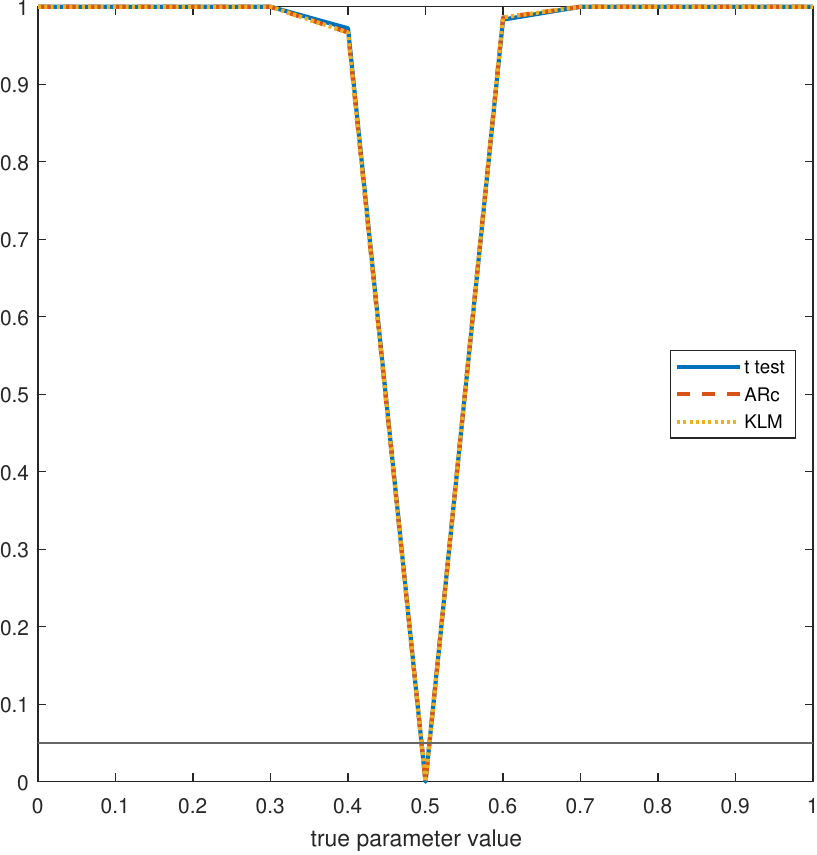} & \includegraphics[width=65mm]{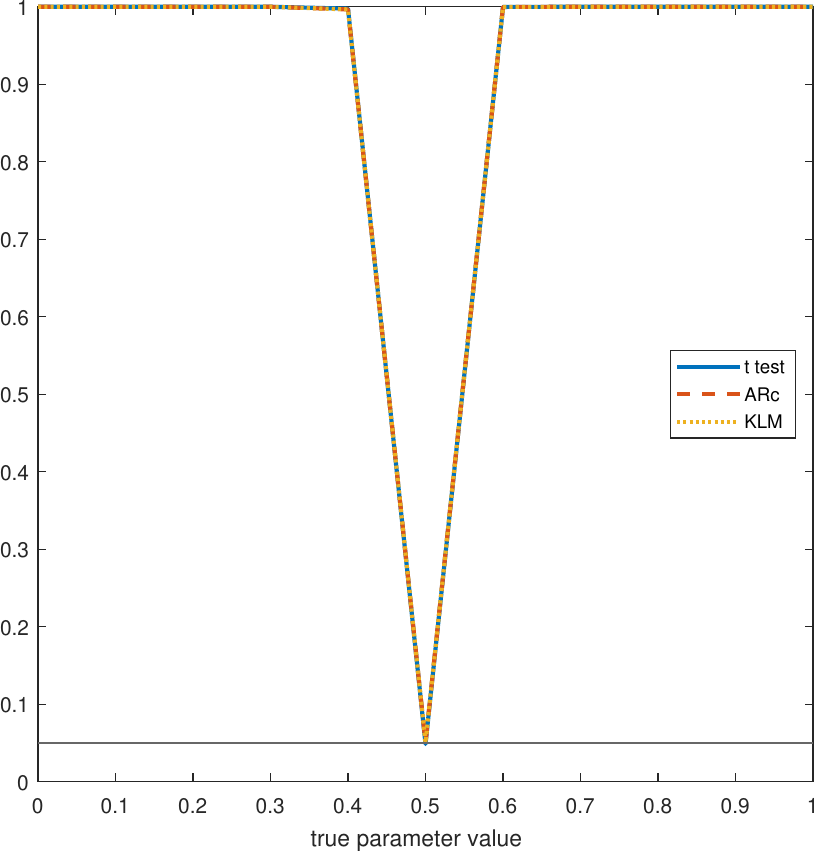}\tabularnewline
\end{tabular}\caption{\label{fig:Fig2}(a)-(f): $z_{t}=[x_{t-1}\;x_{t-2}]'$; (a)-(c): $\rho_{\eta\nu}=0.2$,
(d)-(f): $\rho_{2}=-0.65$}
\end{figure}

\newpage{}

\bibliographystyle{elsarticle-harv}
\addcontentsline{toc}{section}{\refname}\bibliography{sig_band_bibtex}

\begin{thebibliography}{56}
\expandafter\ifx\csname natexlab\endcsname\relax\def\natexlab#1{#1}\fi
\providecommand{\url}[1]{\texttt{#1}}
\providecommand{\href}[2]{#2}
\providecommand{\path}[1]{#1}
\providecommand{\DOIprefix}{doi:}
\providecommand{\ArXivprefix}{arXiv:}
\providecommand{\URLprefix}{URL: }
\providecommand{\Pubmedprefix}{pmid:}
\providecommand{\doi}[1]{\href{http://dx.doi.org/#1}{\path{#1}}}
\providecommand{\Pubmed}[1]{\href{pmid:#1}{\path{#1}}}
\providecommand{\bibinfo}[2]{#2}
\ifx\xfnm\relax \def\xfnm[#1]{\unskip,\space#1}\fi
\bibitem[{Anderson and Rubin(1949)}]{Anderson1949}
\bibinfo{author}{Anderson, T.}, \bibinfo{author}{Rubin, H.},
  \bibinfo{year}{1949}.
\newblock \bibinfo{title}{Estimation of the parameters of a single equation in
  a complete system of stochastic equations}.
\newblock \bibinfo{journal}{The Annals of Mathematical Statistics}
  \bibinfo{volume}{20}, \bibinfo{pages}{46--63}.
\newblock \DOIprefix\doi{10.1214/aoms/1177730090}.
\bibitem[{Andrews(1988)}]{Andrews1988}
\bibinfo{author}{Andrews, D.W.K.}, \bibinfo{year}{1988}.
\newblock \bibinfo{title}{Laws of large numbers for dependent non-identically
  distributed random variables}.
\newblock \bibinfo{journal}{Econometric Theory} \bibinfo{volume}{4},
  \bibinfo{pages}{458--467}.
\newblock \DOIprefix\doi{10.1017/s0266466600013396}.
\bibitem[{Andrews(1991)}]{Andrews1991}
\bibinfo{author}{Andrews, D.W.K.}, \bibinfo{year}{1991}.
\newblock \bibinfo{title}{Heteroskedasticity and autocorrelation consistent
  covariance matrix estimation}.
\newblock \bibinfo{journal}{Econometrica} \bibinfo{volume}{59},
  \bibinfo{pages}{817}.
\newblock \DOIprefix\doi{10.2307/2938229}.
\bibitem[{Andrews(2017)}]{Andrews2017a}
\bibinfo{author}{Andrews, D.W.K.}, \bibinfo{year}{2017}.
\newblock \bibinfo{title}{Identification--robust subvector inference}.
\newblock \bibinfo{journal}{Cowles Foundation Discussion Paper No 2105}
  \DOIprefix\doi{10.2139/ssrn.3032675}.
\bibitem[{Andrews and Cheng(2012)}]{Andrews2012a}
\bibinfo{author}{Andrews, D.W.K.}, \bibinfo{author}{Cheng, X.},
  \bibinfo{year}{2012}.
\newblock \bibinfo{title}{Estimation and inference with weak semi-strong and
  strong identification}.
\newblock \bibinfo{journal}{Econometrica} \bibinfo{volume}{80},
  \bibinfo{pages}{2153--2211}.
\newblock \DOIprefix\doi{10.3982/ECTA9456}.
\bibitem[{Andrews et~al.(2020)Andrews, Cheng and Guggenberger}]{Andrews2020}
\bibinfo{author}{Andrews, D.W.K.}, \bibinfo{author}{Cheng, X.},
  \bibinfo{author}{Guggenberger, P.}, \bibinfo{year}{2020}.
\newblock \bibinfo{title}{Generic results for establishing the asymptotic size
  of confidence sets and tests}.
\newblock \bibinfo{journal}{Journal of Econometrics} \bibinfo{volume}{218},
  \bibinfo{pages}{496--531}.
\newblock \DOIprefix\doi{10.1016/j.jeconom.2020.04.027}.
\bibitem[{Andrews and Guggenberger(2019)}]{Andrews2019}
\bibinfo{author}{Andrews, D.W.K.}, \bibinfo{author}{Guggenberger, P.},
  \bibinfo{year}{2019}.
\newblock \bibinfo{title}{Identification-- and singularity--robust inference
  for moment condition models}.
\newblock \bibinfo{journal}{Quantitative Economics} \bibinfo{volume}{10},
  \bibinfo{pages}{1703--1746}.
\newblock \DOIprefix\doi{10.3982/QE1219}.
\bibitem[{Andrews et~al.(2019)Andrews, Marmer and Yu}]{Andrews2019a}
\bibinfo{author}{Andrews, D.W.K.}, \bibinfo{author}{Marmer, V.},
  \bibinfo{author}{Yu, Z.}, \bibinfo{year}{2019}.
\newblock \bibinfo{title}{On optimal inference in the linear iv model}.
\newblock \bibinfo{journal}{Quantitative Economics} \bibinfo{volume}{10},
  \bibinfo{pages}{457--485}.
\newblock \DOIprefix\doi{10.3982/qe1082}.
\bibitem[{Andrews and Monahan(1992)}]{Andrews1992}
\bibinfo{author}{Andrews, D.W.K.}, \bibinfo{author}{Monahan, J.C.},
  \bibinfo{year}{1992}.
\newblock \bibinfo{title}{An improved heteroskedasticity and autocorrelation
  consistent covariance matrix estimator}.
\newblock \bibinfo{journal}{Econometrica} \bibinfo{volume}{60},
  \bibinfo{pages}{953--966}.
\newblock \DOIprefix\doi{10.2307/2951574}.
\bibitem[{Andrews et~al.(2006)Andrews, Moreira and Stock}]{Andrews2006}
\bibinfo{author}{Andrews, D.W.K.}, \bibinfo{author}{Moreira, M.},
  \bibinfo{author}{Stock, J.}, \bibinfo{year}{2006}.
\newblock \bibinfo{title}{Optimal invariant similar tests for instrumental
  variables regression}.
\newblock \bibinfo{journal}{Econometrica} \bibinfo{volume}{74},
  \bibinfo{pages}{715--752}.
\newblock \DOIprefix\doi{10.3386/t0299}.
\bibitem[{Andrews and Mikusheva(2016)}]{Andrews2016}
\bibinfo{author}{Andrews, I.}, \bibinfo{author}{Mikusheva, A.},
  \bibinfo{year}{2016}.
\newblock \bibinfo{title}{A geometric approach to nonlinear econometric
  models}.
\newblock \bibinfo{journal}{Econometrica} \bibinfo{volume}{84},
  \bibinfo{pages}{1249--1264}.
\newblock \DOIprefix\doi{10.3982/ecta12030}.
\bibitem[{Bhatia(1996)}]{Bhatia1996}
\bibinfo{author}{Bhatia, R.}, \bibinfo{year}{1996}.
\newblock \bibinfo{title}{Matrix Analysis}.
\newblock \bibinfo{publisher}{Springer Verlag}.
\bibitem[{Chaudhuri and Zivot(2011)}]{Chaudhuri2011}
\bibinfo{author}{Chaudhuri, S.}, \bibinfo{author}{Zivot, E.},
  \bibinfo{year}{2011}.
\newblock \bibinfo{title}{A new method of projection-based inference in gmm
  with weakly identified nuisance parameters}.
\newblock \bibinfo{journal}{Journal of Econometrics} \bibinfo{volume}{164},
  \bibinfo{pages}{239--251}.
\newblock \DOIprefix\doi{10.1016/j.jeconom.2011.05.012}.
\bibitem[{Chow and Teicher(1997)}]{Chow1997}
\bibinfo{author}{Chow, Y.S.}, \bibinfo{author}{Teicher, H.},
  \bibinfo{year}{1997}.
\newblock \bibinfo{title}{Probability Theory, Third Edition}.
\newblock \bibinfo{publisher}{Springer Verlag}.
\newblock \DOIprefix\doi{10.1007/978-1-4612-1950-7}.
\bibitem[{Davidson and de~Jong(2000)}]{Davidson2000}
\bibinfo{author}{Davidson, J.}, \bibinfo{author}{de~Jong, R.M.},
  \bibinfo{year}{2000}.
\newblock \bibinfo{title}{The functional central limit theorem and weak
  convergence to stochastic integrals. {II}: {Fractionally} integrated
  processes}.
\newblock \bibinfo{journal}{Econometric Theory} \bibinfo{volume}{16},
  \bibinfo{pages}{643--666}.
\newblock \DOIprefix\doi{10.1017/S0266466600165028}.
\bibitem[{Donald and Newey(2000)}]{Donald2000}
\bibinfo{author}{Donald, S.G.}, \bibinfo{author}{Newey, W.K.},
  \bibinfo{year}{2000}.
\newblock \bibinfo{title}{A jackknife interpretation of the continuous updating
  estimator}.
\newblock \bibinfo{journal}{Economics Letters} \bibinfo{volume}{67},
  \bibinfo{pages}{239--243}.
\newblock \DOIprefix\doi{10.1016/s0165-1765(99)00281-5}.
\bibitem[{Dudley(1989)}]{Dudley1989}
\bibinfo{author}{Dudley, R.M.}, \bibinfo{year}{1989}.
\newblock \bibinfo{title}{Real Analysis and Probability}.
\newblock \bibinfo{publisher}{Chapman \& Hall/CRC}.
\newblock \DOIprefix\doi{10.1017/cbo9780511755347}.
\bibitem[{Dufour(1997)}]{Dufour1997}
\bibinfo{author}{Dufour, J.M.}, \bibinfo{year}{1997}.
\newblock \bibinfo{title}{Some impossibility theorems in econometrics with
  applications tostructural and dynamic models}.
\newblock \bibinfo{journal}{Econometrica} \bibinfo{volume}{65},
  \bibinfo{pages}{1365--1387}.
\bibitem[{Dufour and Taamouti(2005)}]{Dufour2005}
\bibinfo{author}{Dufour, J.M.}, \bibinfo{author}{Taamouti, M.},
  \bibinfo{year}{2005}.
\newblock \bibinfo{title}{Projection-based statistical inference in linear
  structural models with possibly weak instruments}.
\newblock \bibinfo{journal}{Econometrica} \bibinfo{volume}{73},
  \bibinfo{pages}{1351--1365}.
\newblock \DOIprefix\doi{10.1111/j.1468-0262.2005.00618.x}.
\bibitem[{Guggenberger et~al.(2019)Guggenberger, Kleibergen and
  Mavroeidis}]{Guggenberger2019}
\bibinfo{author}{Guggenberger, P.}, \bibinfo{author}{Kleibergen, F.},
  \bibinfo{author}{Mavroeidis, S.}, \bibinfo{year}{2019}.
\newblock \bibinfo{title}{A more powerful subvector anderson rubin test in
  linear instrumental variables regression}.
\newblock \bibinfo{journal}{Quantitative Economics} \bibinfo{volume}{10},
  \bibinfo{pages}{487--526}.
\newblock \DOIprefix\doi{10.3982/qe1116}.
\bibitem[{Guggenberger et~al.(2024)Guggenberger, Kleibergen and
  Mavroeidis}]{Guggenberger2024}
\bibinfo{author}{Guggenberger, P.}, \bibinfo{author}{Kleibergen, F.},
  \bibinfo{author}{Mavroeidis, S.}, \bibinfo{year}{2024}.
\newblock \bibinfo{title}{A powerful subvector anderson–rubin test in linear
  instrumental variables regression with conditional heteroskedasticity}.
\newblock \bibinfo{journal}{Econometric Theory} \bibinfo{volume}{40},
  \bibinfo{pages}{957--1002}.
\newblock \DOIprefix\doi{10.1017/s0266466622000627}.
\bibitem[{Guggenberger et~al.(2012)Guggenberger, Kleibergen, Mavroeidis and
  Chen}]{GuggenbergerKleibergenMavroeidisChen2012}
\bibinfo{author}{Guggenberger, P.}, \bibinfo{author}{Kleibergen, F.},
  \bibinfo{author}{Mavroeidis, S.}, \bibinfo{author}{Chen, L.},
  \bibinfo{year}{2012}.
\newblock \bibinfo{title}{On the asymptotic sizes of subset anderson-rubin and
  lagrange multiplier tests in linear instrumental variables regression}.
\newblock \bibinfo{journal}{Econometrica} \bibinfo{volume}{80},
  \bibinfo{pages}{2649--2666}.
\newblock \DOIprefix\doi{10.3982/qe1116}.
\bibitem[{Hall and Heyde(1980)}]{HallHeyde1980}
\bibinfo{author}{Hall, P.}, \bibinfo{author}{Heyde, C.}, \bibinfo{year}{1980}.
\newblock \bibinfo{title}{Martingale Limit Theory and its Application}.
\newblock \bibinfo{publisher}{Academic Press}.
\bibitem[{Hansen(1982)}]{Hansen1982}
\bibinfo{author}{Hansen, L.}, \bibinfo{year}{1982}.
\newblock \bibinfo{title}{Large sample properties of generalized method of
  moments estimators}.
\newblock \bibinfo{journal}{Econometrica} \bibinfo{volume}{50},
  \bibinfo{pages}{1029}.
\newblock \DOIprefix\doi{10.2307/1912775}.
\bibitem[{Hansen et~al.(1996)Hansen, Heaton and Yaron}]{HansenHeatonYaron1996}
\bibinfo{author}{Hansen, L.P.}, \bibinfo{author}{Heaton, J.},
  \bibinfo{author}{Yaron, A.}, \bibinfo{year}{1996}.
\newblock \bibinfo{title}{Finite--sample properties of some alternative gmm
  estimators}.
\newblock \bibinfo{journal}{Journal of Business and Economic Statistics}
  \bibinfo{volume}{14}, \bibinfo{pages}{262--280}.
\bibitem[{Hansen(1987)}]{Hansen1987}
\bibinfo{author}{Hansen, P.C.}, \bibinfo{year}{1987}.
\newblock \bibinfo{title}{The truncated svd as a method for regularization}.
\newblock \bibinfo{journal}{BIT} \bibinfo{volume}{27},
  \bibinfo{pages}{534--553}.
\bibitem[{Hanson(1971)}]{Hanson1971}
\bibinfo{author}{Hanson, R.}, \bibinfo{year}{1971}.
\newblock \bibinfo{title}{A numerical method for solving fredholm integral
  equations of the first kind using singular values}.
\newblock \bibinfo{journal}{SIAM Journal on Numerical Analysis}
  \bibinfo{volume}{8}, \bibinfo{pages}{616--622}.
\newblock \DOIprefix\doi{10.1137/0708058}.
\bibitem[{Horn and Johnson(1985)}]{Horn1985}
\bibinfo{author}{Horn, R.A.}, \bibinfo{author}{Johnson, C.R.},
  \bibinfo{year}{1985}.
\newblock \bibinfo{title}{Matrix Analysis}.
\newblock \bibinfo{publisher}{Cambridge University Press}.
\newblock \DOIprefix\doi{10.1017/cbo9780511810817}.
\bibitem[{de~Jong(1996)}]{Jong1996}
\bibinfo{author}{de~Jong, R.M.}, \bibinfo{year}{1996}.
\newblock \bibinfo{title}{A strong law of large numbers for triangular
  mixingale arrays}.
\newblock \bibinfo{journal}{Statistics \& Probability Letters}
  \bibinfo{volume}{27}, \bibinfo{pages}{1--9}.
\newblock \DOIprefix\doi{10.1016/0167-7152(95)00036-4}.
\bibitem[{de~Jong(1997)}]{Jong1997}
\bibinfo{author}{de~Jong, R.M.}, \bibinfo{year}{1997}.
\newblock \bibinfo{title}{Central limit theorems for dependent heterogeneous
  random variables}.
\newblock \bibinfo{journal}{Econometric Theory} \bibinfo{volume}{13},
  \bibinfo{pages}{353--367}.
\newblock \DOIprefix\doi{10.1017/s0266466600005843}.
\bibitem[{de~Jong and Davidson(2000)}]{Jong2000}
\bibinfo{author}{de~Jong, R.M.}, \bibinfo{author}{Davidson, J.},
  \bibinfo{year}{2000}.
\newblock \bibinfo{title}{Consistency of kernel estimators of heteroscedastic
  and autocorrelated covariance matrices}.
\newblock \bibinfo{journal}{Econometrica} \bibinfo{volume}{68},
  \bibinfo{pages}{407--423}.
\newblock \DOIprefix\doi{10.1111/1468-0262.00115}.
\bibitem[{Jordà(2005)}]{Jorda2005}
\bibinfo{author}{Jordà, O.}, \bibinfo{year}{2005}.
\newblock \bibinfo{title}{Estimation and inference of impulse responses by
  local projections}.
\newblock \bibinfo{journal}{American Economic Review} \bibinfo{volume}{95},
  \bibinfo{pages}{161--182}.
\newblock \DOIprefix\doi{10.1257/0002828053828518}.
\bibitem[{Kleibergen(2002)}]{Kleibergen2002}
\bibinfo{author}{Kleibergen, F.}, \bibinfo{year}{2002}.
\newblock \bibinfo{title}{Pivotal statistics for testing structural parameters
  in instrumental variables regression}.
\newblock \bibinfo{journal}{Econometrica} \bibinfo{volume}{70},
  \bibinfo{pages}{1781--1803}.
\newblock \DOIprefix\doi{10.1111/1468-0262.00353}.
\bibitem[{Kleibergen(2005)}]{Kleibergen2005}
\bibinfo{author}{Kleibergen, F.}, \bibinfo{year}{2005}.
\newblock \bibinfo{title}{Testing parameters in gmm without assuming that they
  are identified}.
\newblock \bibinfo{journal}{Econometrica} \bibinfo{volume}{73},
  \bibinfo{pages}{1103--1123}.
\newblock \DOIprefix\doi{10.1111/j.1468-0262.2005.00610.x}.
\bibitem[{Kleibergen and Mavroeidis(2009)}]{Kleibergen2009}
\bibinfo{author}{Kleibergen, F.}, \bibinfo{author}{Mavroeidis, S.},
  \bibinfo{year}{2009}.
\newblock \bibinfo{title}{Weak instrument robust tests in gmm and the new
  keynesian phillips curve}.
\newblock \bibinfo{journal}{Journal of Business \& Economic Statistics}
  \bibinfo{volume}{27}, \bibinfo{pages}{293--311}.
\newblock \DOIprefix\doi{10.1198/jbes.2009.08280}.
\bibitem[{Magnus and Neudecker(1988)}]{Magnus1988}
\bibinfo{author}{Magnus, J.R.}, \bibinfo{author}{Neudecker, H.},
  \bibinfo{year}{1988}.
\newblock \bibinfo{title}{Matrix Differential Calculus with Applications in
  Statistics and Econometrics.}
\newblock \bibinfo{publisher}{John Wiley \& Sons Ltd}.
\newblock \DOIprefix\doi{10.2307/2531754}.
\bibitem[{McLeish(1974)}]{McLeish1974}
\bibinfo{author}{McLeish, D.L.}, \bibinfo{year}{1974}.
\newblock \bibinfo{title}{Dependent central limit theorems and invariance
  principles}.
\newblock \bibinfo{journal}{The Annals of Probability} \bibinfo{volume}{2},
  \bibinfo{pages}{620--628}.
\newblock \DOIprefix\doi{10.1214/aop/1176996608}.
\bibitem[{McLeish(1977)}]{McLeish1977}
\bibinfo{author}{McLeish, D.L.}, \bibinfo{year}{1977}.
\newblock \bibinfo{title}{On the invariance principle for nonstationary
  mixingales}.
\newblock \bibinfo{journal}{The Annals of Probability} \bibinfo{volume}{5},
  \bibinfo{pages}{616--621}.
\newblock \DOIprefix\doi{10.1214/aop/1176995772}.
\bibitem[{Merlevède et~al.(2019)Merlevède, Peligrad and Utev}]{Merlevede2019}
\bibinfo{author}{Merlevède, F.}, \bibinfo{author}{Peligrad, M.},
  \bibinfo{author}{Utev, S.}, \bibinfo{year}{2019}.
\newblock \bibinfo{title}{Functional clt for martingale-like nonstationary
  dependent structures}.
\newblock \bibinfo{journal}{Bernoulli} \bibinfo{volume}{25},
  \bibinfo{pages}{3203--3233}.
\newblock \DOIprefix\doi{10.3150/18-bej1088}.
\bibitem[{Moreira(2003)}]{Moreira2003}
\bibinfo{author}{Moreira, M.}, \bibinfo{year}{2003}.
\newblock \bibinfo{title}{A conditional likelihood ratio test for structural
  models}.
\newblock \bibinfo{journal}{Econometrica} \bibinfo{volume}{71},
  \bibinfo{pages}{1027--1048}.
\newblock \DOIprefix\doi{10.1111/1468-0262.00438}.
\bibitem[{Moreira(2009)}]{Moreira2009}
\bibinfo{author}{Moreira, M.}, \bibinfo{year}{2009}.
\newblock \bibinfo{title}{Tests with correct size when instruments can be
  arbitrarily weak}.
\newblock \bibinfo{journal}{Journal of Econometrics} \bibinfo{volume}{152},
  \bibinfo{pages}{131--140}.
\newblock \DOIprefix\doi{10.1016/j.jeconom.2009.01.012}.
\bibitem[{Neumann(2013)}]{Neumann2013}
\bibinfo{author}{Neumann, M.H.}, \bibinfo{year}{2013}.
\newblock \bibinfo{title}{A central limit theorem for triangular arrays of
  weakly dependent random variables, with applications in statistics}.
\newblock \bibinfo{journal}{ESAIM: Probability and Statistics}
  \bibinfo{volume}{17}, \bibinfo{pages}{120--134}.
\newblock \DOIprefix\doi{10.1051/ps/2011144}.
\bibitem[{Newey and Smith(2004)}]{Newey2004}
\bibinfo{author}{Newey, W.}, \bibinfo{author}{Smith, R.}, \bibinfo{year}{2004}.
\newblock \bibinfo{title}{Higher order properties of gmm and generalized
  empirical likelihood estimators}.
\newblock \bibinfo{journal}{Econometrica} \bibinfo{volume}{72},
  \bibinfo{pages}{219--255}.
\newblock \DOIprefix\doi{10.1111/j.1468-0262.2004.00482.x}.
\bibitem[{Newey and West(1987)}]{Newey1987}
\bibinfo{author}{Newey, W.}, \bibinfo{author}{West, K.}, \bibinfo{year}{1987}.
\newblock \bibinfo{title}{A simple, positive semi-definite, heteroskedasticity
  and autocorrelation consistent covariance matrix}.
\newblock \bibinfo{journal}{Econometrica} \bibinfo{volume}{55},
  \bibinfo{pages}{703}.
\newblock \DOIprefix\doi{10.2307/1913610}.
\bibitem[{Neyman(1959)}]{Neyman1959}
\bibinfo{author}{Neyman, J.}, \bibinfo{year}{1959}.
\newblock \bibinfo{title}{Statistical problems in science. the symmetric test
  of a composite hypothesis}.
\newblock \bibinfo{journal}{Journal of the American Statistical Association}
  \bibinfo{volume}{64}, \bibinfo{pages}{1154--1171}.
\newblock \DOIprefix\doi{10.1080/01621459.1969.10501047}.
\bibitem[{Phillips(1989)}]{Phillips1989}
\bibinfo{author}{Phillips, P.C.}, \bibinfo{year}{1989}.
\newblock \bibinfo{title}{Partially identified econometric models}.
\newblock \bibinfo{journal}{Econometric Theory} \bibinfo{volume}{5},
  \bibinfo{pages}{181--240}.
\bibitem[{Rio(1997)}]{Rio1997}
\bibinfo{author}{Rio, E.}, \bibinfo{year}{1997}.
\newblock \bibinfo{title}{About the lindeberg method for strongly mixing
  sequences}.
\newblock \bibinfo{journal}{ESAIM: Probability and Statistics}
  \bibinfo{volume}{1}, \bibinfo{pages}{35--61}.
\newblock \DOIprefix\doi{10.1051/ps:1997102}.
\bibitem[{Sargan(1958)}]{Sargan1958}
\bibinfo{author}{Sargan, J.D.}, \bibinfo{year}{1958}.
\newblock \bibinfo{title}{The estimation of economic relationships using
  instrumental variables}.
\newblock \bibinfo{journal}{Econometrica} \bibinfo{volume}{26},
  \bibinfo{pages}{393--415}.
\bibitem[{Staiger and Stock(1997)}]{Staiger1997}
\bibinfo{author}{Staiger, D.}, \bibinfo{author}{Stock, J.},
  \bibinfo{year}{1997}.
\newblock \bibinfo{title}{Instrumental variables regression with weak
  instruments}.
\newblock \bibinfo{journal}{Econometrica} \bibinfo{volume}{65},
  \bibinfo{pages}{557--586}.
\newblock \DOIprefix\doi{10.3386/t0151}.
\bibitem[{Stewart(1977)}]{Stewart1977}
\bibinfo{author}{Stewart, G.}, \bibinfo{year}{1977}.
\newblock \bibinfo{title}{On the perturbation of pseudo-inverses, projections
  and linear least squares problems}.
\newblock \bibinfo{journal}{SIAM Review} \bibinfo{volume}{19},
  \bibinfo{pages}{634--662}.
\newblock \DOIprefix\doi{10.1137/1019104}.
\bibitem[{Stock and Wright(2000)}]{Stock2000}
\bibinfo{author}{Stock, J.}, \bibinfo{author}{Wright, J.},
  \bibinfo{year}{2000}.
\newblock \bibinfo{title}{Gmm with weak identification}.
\newblock \bibinfo{journal}{Econometrica} \bibinfo{volume}{68},
  \bibinfo{pages}{1055--1096}.
\newblock \DOIprefix\doi{10.1111/1468-0262.00151}.
\bibitem[{Teicher(1985)}]{Teicher1985}
\bibinfo{author}{Teicher, H.}, \bibinfo{year}{1985}.
\newblock \bibinfo{title}{Almost certain convergence in double arrays}.
\newblock \bibinfo{journal}{Z. Wahrscheinlichkeitstheorie verw. Gebiete}
  \bibinfo{volume}{69}, \bibinfo{pages}{331--345}.
\newblock \DOIprefix\doi{10.1007/bf00532738}.
\bibitem[{van~der Vaart and Wellner(1996)}]{vanderVaartWellner1996}
\bibinfo{author}{van~der Vaart, A.W.}, \bibinfo{author}{Wellner, J.A.},
  \bibinfo{year}{1996}.
\newblock \bibinfo{title}{Weak Convergence and Empirical Processes}.
\newblock \bibinfo{publisher}{Spinger Verlag}.
\bibitem[{Varah(1973)}]{Varah1973}
\bibinfo{author}{Varah, J.}, \bibinfo{year}{1973}.
\newblock \bibinfo{title}{On the numerical solution of ill-conditioned linear
  systems with applications to ill-posed problems}.
\newblock \bibinfo{journal}{SIAM Journal on Numerical Analysis}
  \bibinfo{volume}{10}, \bibinfo{pages}{257--267}.
\newblock \DOIprefix\doi{10.1137/0710025}.
\bibitem[{Wedin(1973)}]{Wedin1973}
\bibinfo{author}{Wedin, P.}, \bibinfo{year}{1973}.
\newblock \bibinfo{title}{Perturbation theory for pseudo-inverses}.
\newblock \bibinfo{journal}{BIT} \bibinfo{volume}{13},
  \bibinfo{pages}{217--232}.
\newblock \DOIprefix\doi{10.1007/bf01933494}.
\bibitem[{Wooldridge and White(1988)}]{Wooldridge1988}
\bibinfo{author}{Wooldridge, J.M.}, \bibinfo{author}{White, H.},
  \bibinfo{year}{1988}.
\newblock \bibinfo{title}{Some invariance principles and central limit theorems
  for dependent heterogeneous processes}.
\newblock \bibinfo{journal}{Econometric Theory} \bibinfo{volume}{4},
  \bibinfo{pages}{210--230}.
\newblock \DOIprefix\doi{10.1017/s0266466600012032}.

\end{thebibliography}

\newpage{}

\appendix

\section{Proofs}

\subsection{Preliminaries}

The proof of Theorem \ref{thm:AsySz} is based on Lemmas \ref{lem:gamma-conv},
\ref{lem:LemA1}, \ref{lem:LemA2} and \ref{lem:LemA3}. We prove
Lemma \ref{lem:gamma-conv} in Section \ref{subsec:Gamma_tilde} and
we state and prove Lemmas \ref{lem:LemA1}, \ref{lem:LemA2} and \ref{lem:LemA3}
in Section \ref{subsec:Proofs_Main} which also contains the proof
of Theorem \ref{thm:AsySz}. We collect some results on perturbation
bounds in Section \ref{subsec:Perturbation-Bounds}.

\subsection{\label{subsec:Perturbation-Bounds}Perturbation Bounds}

We state the following result for ease of reference that links the
converge in operator norm of matrices to the convergence of the corresponding
eigenvalues. For a $d_{\gamma}\times d_{\gamma}$ Hermitian matrices
$\mathscr{A}_{n}$ and $\mathscr{A}$ such that $\left\Vert \mathscr{A}_{n}-A\right\Vert \rightarrow0$
let $\lambda_{1}\left(\mathscr{A}\right)\geq...\geq\lambda_{d_{\gamma}}\left(\mathscr{A}\right)$
and $\lambda_{1}\left(\mathscr{A}_{n}\right)\geq...\geq\lambda_{d_{\gamma}}\left(\mathscr{A}_{n}\right)$
be the eigenvalues of $\mathscr{A}$ and $\mathscr{A}_{n}$ in decreasing
order.
\begin{thm}
\label{thm:Weyl}(Weyl's perturbation theorem) For $\mathscr{A},$$\mathscr{A}_{n}$
and $\lambda_{j}\left(.\right)$ defined above it follows that 
\begin{equation}
\text{max}_{j}\left|\lambda_{j}\left(\mathscr{A}\right)-\lambda_{j}\left(\mathscr{A}_{n}\right)\right|\leq\left\Vert \mathscr{A}-\mathscr{A}_{n}\right\Vert \rightarrow0.\label{eq:Weyl-Perturb}
\end{equation}
\end{thm}
For a proof of the theorem see \citeauthor{Bhatia1996} (1996, Corollary
III.2.6).

A related result links the convergence of eigenvectors to the converge
in operator norm of matrices. The convergence of the eigenvectors
is in terms of the supspaces they span. The following result is slightly
adapted from \citeauthor{Bhatia1996}(1996, Theorem VII.3.2).\nocite{Bhatia1996}
\begin{thm}
\label{thm:Eigenvector}Let $\mathscr{L}_{1}$ and $\mathscr{L}_{2}$
be two subsets of $\mathbb{R}$ such that $\textrm{dist}\left(\mathscr{L}_{1},\mathscr{L}_{2}\right)=\delta>0$
for some $\delta.$ Let $\mathscr{A}_{n}$ and $\mathscr{A}$ be Hermitian
matrices with eigenvectors $R_{n},$ and $R$ respectively. Let $\hat{R}_{1,n}=P_{\mathscr{A}_{n}}\left(\mathscr{L}_{1}\right)$
be the projection matrix projecting onto the subspace spanned by the
eigenvectors $R_{n}$ of $\mathscr{A}_{n}$ associated with eigenvalues
who are elements of the set $\mathscr{L}_{1}.$ Similarly, define
$\hat{R_{2}}=P_{\mathscr{A}}\left(\mathscr{L}_{2}\right)$. Then,
\begin{equation}
\left\Vert \hat{R}_{1,n}\hat{R_{2}}\right\Vert \leq\frac{\pi}{2\delta}\left\Vert \mathscr{A}_{n}-\mathscr{A}\right\Vert .\label{eq:eig_ortho}
\end{equation}
\end{thm}
For a proof of the theorem see \citeauthor{Bhatia1996} (1996, p.212)\nocite{Bhatia1996}.

\subsection{\label{subsec:Gamma_tilde}Proof of Lemma \ref{lem:gamma-conv}}
\begin{proof}[Proof of Lemma \ref{lem:gamma-conv}]
 First introduce some notation. Recall the definition of $\hat{\Omega}_{n}\left(\theta\right)$
in \eqref{eq:Omega_hat} and $\Omega_{n}\left(\theta\right)$ in Assumption
\ref{assu:Omega_Conv} which implies $\det\left(\Omega_{n}\left(\theta\right)\right)\geq K_{\Omega}>0$
uniformly in $\theta$ and define $\hat{\Omega}_{n,0}=\hat{\Omega}_{n}\left(\theta_{n,0}\right).$
For $\theta_{1}=\left(\beta_{1},\gamma_{1}\right)$ and $\theta_{2}=\left(\beta_{2},\gamma_{2}\right)$
define 
\[
\hat{\Omega}_{n}\left(\theta_{1},\theta_{2}\right)=n^{-1}\sum_{t=1}^{n}\sum_{s=1}^{n}k\left(\frac{t-s}{a_{n}}\right)g_{t}\left(\beta_{1},\gamma_{1}\right)g_{s}\left(\beta_{2},\gamma_{2}\right)'.
\]
Next consider $\hat{g}_{n}\left(\beta,\gamma\right)=n^{-1}\sum_{t=1}^{n}g_{t}\left(\beta,\gamma\right)$
and define 
\begin{equation}
\hat{g}_{n}\left(\gamma\right)=\hat{g}_{n}\left(\beta_{n,0},\gamma\right),\label{eq:ghat_gamma}
\end{equation}
 
\begin{equation}
\hat{g}_{n,0}=\hat{g}_{n}\left(\beta_{n,0},\gamma_{n,0}\right),\label{eq:ghat_zero}
\end{equation}
 
\begin{equation}
g_{t}^{\gamma}\left(\gamma\right)=\partial g_{t}\left(\beta_{n,0},\gamma\right)/\partial\gamma',\label{eq:Def_g_gam}
\end{equation}
where $g_{t}^{\gamma}\left(\gamma\right)$ is a $d\times d_{\gamma}$
matrix of functions. Use the short hand notation $\hat{g}_{n,\gamma}=\hat{g}_{n,\gamma}\left(\gamma_{n,0}\right)$
for $\hat{g}_{n,\gamma}\left(\gamma\right)$ defined in \eqref{eq:g_gamma(gamma)_hat}.
In the same way, it follows from $\varepsilon_{n,t}=g_{t}\left(\beta_{n,0},\gamma_{n,0}\right)$
in (\ref{eq:epsilon_nt_h}) that 
\begin{align*}
\hat{\Omega}_{n}\left(\theta_{n,0},\left(\beta_{n,0},\gamma\right)\right) & =n^{-1}\sum_{t=1}^{n}\sum_{s=1}^{n}k\left(\frac{t-s}{a_{n}}\right)g_{t}\left(\beta_{n,0},\gamma_{n,0}\right)g_{s}\left(\beta_{n,0},\gamma\right)'\\
 & =n^{-1}\sum_{t=1}^{n}\sum_{s=1}^{n}k\left(\frac{t-s}{a_{n}}\right)\varepsilon_{n,t}g_{s}\left(\beta_{n,0},\gamma\right)'.
\end{align*}

To show (i) consider an arbitrary converging subsequence $n_{k}$
with limit $\Gamma=g_{\gamma}'\Omega_{0}^{-1}g_{\gamma}$ and eigenvalues
$\Delta_{j}$ in Assumption \ref{assu:g_gamma}. Fix $\varepsilon$
such that $\Delta_{r}>\varepsilon>0,$ and recall that $\left\Vert .\right\Vert $
is the operator norm such that 
\[
\left\Vert \mathring{\Gamma}_{n}^{+}\right\Vert \leq\frac{1}{\min_{j\text{ s.t.\ensuremath{\hat{\lambda}_{n,j}}\ensuremath{>\ensuremath{\varepsilon}} }}\hat{\lambda}_{n,j}}\leq\varepsilon^{-1}<\infty
\]
where $\mathring{\Gamma}_{n}$ is defined in \eqref{eq:Gamma_dot}.
This shows that $\left\Vert \mathring{\Gamma}_{n_{k}}^{+}\right\Vert =O_{p}\left(1\right).$
The result in (i) then follows from $\hat{A}_{n_{k}}=O_{p}\left(1\right)$
which holds by \eqref{eq:Def_A_hat} and Assumptions \ref{assu:g_gamma},
\ref{assu:Lambda} and \ref{assu:CLT}, $\hat{\Omega}_{n_{k},0}^{-1/2}=O_{p}\left(1\right)$
which holds by Assumption \ref{assu:Omega_Conv} and $\hat{g}_{n_{k},0}=O_{p}\left(n_{k}^{-1/2}\right)$
which holds by \eqref{eq:epsilon_nt_h} and Assumption \ref{assu:CLT}.

For the second result in (ii) consider any of the converging subsequences
$n_{k}.$ Note that by Assumptions \ref{assu:Asy_Tightness} and \ref{assu:g_gamma}
it follows that $\hat{g}_{n_{k},\gamma}\rightarrow_{p}g_{\gamma}$
for some $d\times d_{\gamma}$ matrix $g_{\gamma}$ that may depend
on $n_{k}$ and that 
\[
\hat{\Omega}_{n_{k},0}^{-1/2}\hat{g}_{n_{k},\gamma}=\Omega_{0}^{-1/2}g_{\gamma}+o_{p}\left(1\right).
\]
Then, define 
\begin{equation}
\bar{A}=\Omega_{0}^{-1/2}g_{\gamma}\label{eq:A_bar}
\end{equation}
 and 
\begin{equation}
\Gamma=g_{\gamma}'\Omega_{0}^{-1}g_{\gamma}=\bar{A}'\bar{A}\label{eq:Gamma}
\end{equation}
 and denote $\rank\left(\Gamma\right)=r$. First consider the case
where $r\geq1.$ By Theorem \ref{thm:Weyl} and Assumptions \ref{assu:Asy_Tightness}
and \ref{assu:g_gamma} it follows that $\hat{\Delta}_{n_{k},r}\rightarrow_{p}\Delta_{r}>0$.
Now set $\varepsilon=\Delta_{r}/2>0.$ When $r\geq1$ \citeauthor{Hansen1987}
(1987, Theorem 3.2) applies because $\Delta_{r}-\Delta_{r+1}=\Delta_{r}>\varepsilon$
and $\left\Vert \Gamma-\hat{\Gamma}_{n_{k}}\right\Vert <\varepsilon$
with probability approaching one by Assumptions \ref{assu:Omega_Conv}
and \ref{assu:g_gamma}. Then, for $n_{k}$ large enough \citeauthor{Hansen1987}
(1987, Theorem 3.2)\nocite{Hansen1987}\footnote{Note that \citeauthor{Hansen1987} (1987, p.534) uses $\left\Vert .\right\Vert _{2}$
for the spectral norm, which is the same as the operator norm $\left\Vert .\right\Vert $
defined in this paper, see \citeauthor{Horn1985} (1985, Section 5.6,
p.295).} provides an upper bound for the difference between the regularized
and limiting MP inverse of $\Gamma$ where it follows that 
\begin{equation}
\left\Vert \Gamma^{+}-\mathring{\Gamma}_{n_{k}}^{+}\right\Vert \leq\frac{3}{\Delta_{r}^{2}\left(1-\eta_{r}\right)}\left\Vert \Gamma-\hat{\Gamma}_{n_{k}}\right\Vert =o_{p}\left(1\right)\label{eq:genInvApp}
\end{equation}
 and where 
\[
\eta_{r}=\frac{\left\Vert \Gamma-\hat{\Gamma}_{n_{k}}\right\Vert }{\Delta_{r}}=o_{p}\left(1\right).
\]
For the case when $r=0$ we have $\Gamma^{+}=0$ and $\left\Vert \mathring{\Gamma}_{n_{k}}^{+}\right\Vert =o_{p}\left(1\right)$
because $\hat{\Delta}_{n,j}\rightarrow_{p}\Delta_{j}=0$ for all $j=1,...,d$.
Now define
\[
\bar{A}_{n}=\hat{\Omega}_{n,0}^{-1/2}\hat{g}_{n,\gamma}
\]
such that 
\begin{equation}
\bar{A}_{n}'\bar{A}_{n}=\hat{g}_{n,\gamma}'\hat{\Omega}_{n,0}^{-1}\hat{g}_{n,\gamma}=\hat{\Gamma}_{n}\label{eq:Gamma_A_rep}
\end{equation}
and 
\begin{equation}
\hat{\Omega}_{n_{k},0}^{-1/2}\hat{g}_{n_{k},0}=O_{p}\left(n_{k}^{-1/2}\right).\label{eq:Lem-gc-0}
\end{equation}
Then, for $\hat{A}_{n}$ defined in \eqref{eq:Def_A_hat}, it follows
\begin{equation}
\hat{A}_{n_{k}}=\bar{A}_{n_{k}}+O_{p}\left(n_{k}^{-1/2}\right)\label{eq:A_approx}
\end{equation}
Next, note that by evaluating \eqref{eq:Pseudo_CUE_FOC} at $\tilde{\gamma}_{n,0}$
defined in \eqref{eq:gamma_til_TSVD}, setting $\tilde{g}_{n_{k},0}=\hat{g}_{n}\left(\tilde{\gamma}_{n,0}\right)$
and using \eqref{eq:gamma_mve} and \eqref{eq:gamma_til_TSVD} leads
to 
\begin{eqnarray}
\hat{A}_{n_{k}}'\hat{\Omega}_{n_{k},0}^{-1/2}\tilde{g}_{n_{k},0} & = & \hat{A}_{n_{k}}'\hat{\Omega}_{n_{k},0}^{-1/2}\left(\hat{g}_{n_{k},0}+\hat{g}_{n_{k},\gamma}\left(\check{\gamma}_{n_{k},0}\right)\left(\tilde{\gamma}_{n_{k},0}-\gamma_{n_{k},0}\right)\right)\label{eq:FOC_mve}\\
 & = & \left(\hat{A}_{n_{k}}'-\hat{A}_{n_{k}}'\bar{A}_{n_{k}}\mathring{\Gamma}_{n_{k}}^{+}\hat{A}_{n_{k}}'\right)\hat{\Omega}_{n_{k},0}^{-1/2}\hat{g}_{n_{k},0}+o_{p}\left(n_{k}^{-1/2}\right)\nonumber 
\end{eqnarray}
where $\bar{A}$ is defined in \eqref{eq:A_bar}. Next consider the
RHS of \eqref{eq:FOC_mve} where by \eqref{eq:Gamma_A_rep} and \eqref{eq:A_approx}
\begin{equation}
\left(\hat{A}_{n_{k}}'-\hat{A}_{n_{k}}'\bar{A}_{n_{k}}\mathring{\Gamma}_{n}^{+}\hat{A}_{n_{k}}'\right)\hat{\Omega}_{n_{k},0}^{-1/2}\hat{g}_{n_{k},0}=\left(\bar{A}_{n_{k}}'-\hat{\Gamma}_{n_{k}}\mathring{\Gamma}_{n_{k}}^{+}\bar{A}_{n_{k}}'\right)\hat{\Omega}_{n_{k},0}^{-1/2}\hat{g}_{n_{k},0}+o_{p}\left(n_{k}^{-1/2}\right).\label{eq:Lem-gc-1}
\end{equation}
The first term on the RHS of \eqref{eq:Lem-gc-1} can be approximated
as 
\begin{equation}
\left(\bar{A}_{n_{k}}'-\hat{\Gamma}_{n_{k}}\mathring{\Gamma}_{n_{k}}^{+}\bar{A}_{n_{k}}'\right)=\left(\bar{A}'-\Gamma\mathring{\Gamma}_{n_{k}}^{+}\bar{A}'\right)+o_{p}\left(1\right)\label{eq:Lem-gc-2}
\end{equation}
where $\Gamma$ is defined in \eqref{eq:Gamma}. Then, \eqref{eq:genInvApp}
implies that 
\begin{equation}
\bar{A}\mathring{\Gamma}_{n_{k}}^{+}\bar{A}'=P_{\bar{A}}+o_{p}\left(1\right)\label{eq:Lem-gc-3}
\end{equation}
where $P_{\bar{A}}=\bar{A}\left(\bar{A}'\bar{A}\right)^{+}\bar{A}'$
is the projection onto the column space spanned by $\bar{A}$. Note
$\bar{A}'\bar{A}=\Gamma$ and that $\mathring{\Gamma}_{n_{k}}-\Gamma=o_{p}\left(1\right)$
along converging subsequences, such that by \eqref{eq:genInvApp}
it follows that $\mathring{\Gamma}_{n_{k}}^{+}-\Gamma^{+}=o_{p}\left(1\right).$
Plugging back \eqref{eq:Lem-gc-3} in \eqref{eq:Lem-gc-2}, and using
properties of projections, shows that 
\begin{equation}
\bar{A}'-\bar{A}'\bar{A}\mathring{\Gamma}_{n_{k}}^{+}\bar{A}'=\bar{A}'-\bar{A}'P_{\bar{A}}+o_{p}\left(1\right)=o_{p}\left(1\right).\label{eq:Lem-gc-4}
\end{equation}
Now, combine \eqref{eq:FOC_mve}, \eqref{eq:Lem-gc-1}, \eqref{eq:Lem-gc-2}
and \eqref{eq:Lem-gc-4} and use \eqref{eq:Lem-gc-0} to establish
the claim in (ii).

For the last claim in (iii), note that for $P_{\hat{A}_{n}}=\hat{A}_{n}\left(\hat{A}_{n}'\hat{A}_{n}\right)^{+}\hat{A}_{n}'$
it follows that
\begin{eqnarray}
\hat{A}_{n_{k}}'\left(P_{\hat{A}_{n_{k}}}-\hat{A}_{n_{k}}\mathring{\Gamma}_{n_{k}}^{+}\hat{A}_{n_{k}}'\right) & = & \hat{A}_{n_{k}}'-\hat{A}_{n_{k}}'\hat{A}_{n_{k}}\mathring{\Gamma}_{n_{k}}^{+}\hat{A}_{n_{k}}'\nonumber \\
 & = & \bar{A}'-\bar{A}'\bar{A}\mathring{\Gamma}_{n_{k}}^{+}\bar{A}'+o_{p}\left(1\right)=o_{p}\left(1\right)\label{eq:Lem-gc-5}
\end{eqnarray}
where the first equality follows from the fact that $P_{\hat{A}_{n_{k}}}$
is a projection, the second equality follows from \eqref{eq:A_approx}
and the third equality follows from \eqref{eq:Lem-gc-4}. Similarly,
let $\hat{A}_{n_{k}}^{\bot}$ be the matrix with columns orthogonal
to the column space of $\hat{A}_{n_{k}}.$ Then, it follows immediately
that 
\begin{equation}
\left(\hat{A}_{n_{k}}^{\bot}\right)'\left(P_{\hat{A}_{n_{k}}}-\hat{A}_{n_{k}}\mathring{\Gamma}_{n_{k}}^{+}\hat{A}_{n_{k}}'\right)=0.\label{eq:Lem-gc-6}
\end{equation}
Combining \eqref{eq:Lem-gc-5} and \eqref{eq:Lem-gc-6} and using
the fact that projections are uniquely defined by the space they project
on, establishes the claim.
\end{proof}

\subsection{Proof of Theorem \ref{thm:AsySz}\label{subsec:Proofs_Main}}
\begin{lem}
\label{lem:LemA1}Let $\textrm{AR}_{\textrm{C}}$ be defined in \eqref{eq:CUAR}.
Let $\theta_{n,0}=\left(\beta_{n,0},\gamma_{n,0}\right)$ be as in
Assumption \ref{assu:H0} and for any converging subsequence $n_{k}$
let \textup{$\tilde{\gamma}_{n,0}$} be as defined in \eqref{eq:gamma_til_TSVD}.
Recall the definitions of $\tilde{\Omega}_{n,0}=\hat{\Omega}_{n}\left(\beta_{n,0},\tilde{\gamma}_{n,0}\right)$,
$\widetilde{g}_{n,0}=\hat{g}_{n}\left(\beta_{n,0},\tilde{\gamma}_{n,0}\right),$
$\tilde{g}_{n,\gamma}=\partial\hat{g}_{n}\left(\beta_{n,0},\tilde{\gamma}_{n,0}\right)/\partial\gamma'$,
and let $\hat{\Lambda}_{n,0}=\hat{\Lambda}_{n}\left(\gamma_{n,0}\right)$
be as defined in \eqref{eq:Lambda_hat}, $\hat{A}_{n}$ as in \eqref{eq:Def_A_hat}
and $P_{\hat{A}_{n}}$ the projection onto the column space spanned
by $\hat{A}_{n}$. Let $M_{\hat{A}_{n}}=I-P_{\hat{A}_{n}}.$ Then,
along all converging subsequences $n_{k},$\\
(i) 
\begin{equation}
\textrm{AR}_{\textrm{C}}\left(\beta_{n_{k},0}\right)\leq n_{k}\tilde{g}_{n_{k},0}'\tilde{\Omega}_{n_{k},0}^{-1/2}M_{\hat{A}_{n_{k}}}\tilde{\Omega}_{n_{k},0}^{-1/2}\tilde{g}_{n_{k},0}+\varpi_{n_{k}},\label{eq:up-bound}
\end{equation}
(ii) 
\begin{equation}
\varpi_{n_{k}}=n_{k}\tilde{g}_{n_{k}}'\tilde{\Omega}_{n_{k}}^{-1/2}P_{\hat{A}_{n_{k}}}\tilde{\Omega}_{n_{k}}^{-1/2}\tilde{g}_{n_{k}}=o_{p}\left(1\right),\label{eq:Lemma2_Res1}
\end{equation}
(iii)
\begin{equation}
n_{k}\tilde{g}_{n_{k},0}'\tilde{\Omega}_{n_{k},0}^{-1/2}M_{\hat{A}_{n_{k}}}\tilde{\Omega}_{n_{k},0}^{-1/2}\tilde{g}_{n_{k},0}=n_{k}\hat{g}_{n_{k},0}'\hat{\Omega}_{n_{k},0}^{-1/2}M_{\hat{A}_{n_{k}}}\hat{\Omega}_{n_{k},0}^{-1/2}\hat{g}_{n_{k},0}+o_{p}\left(1\right).\label{eq:Lemma2_Res2}
\end{equation}
\end{lem}
\begin{proof}
By Assumption \ref{assu:Omega_Conv} $\hat{\Omega}_{n}\left(\beta,\gamma\right)$
is full rank for all $\beta,\gamma$ w.p.a.1. Thus, w.p.a.1, $\hat{\Omega}_{n}^{-1}\left(\beta,\gamma\right)$
exists. Recall the definitions $\hat{g}_{n,0}=\hat{g}_{n}\left(\beta_{n,0},\gamma_{n,0}\right),$
$\hat{\Omega}_{n,0}=\hat{\Omega}_{n}\left(\beta_{n,0},\gamma_{n,0}\right)$
and $\hat{g}_{n,\gamma}=\partial\hat{g}_{n}\left(\beta_{n,0},\gamma_{n,0}\right)/\partial\gamma'$.

We start with the proof of (i). We have the following inequality for
all converging subsequences $n_{k}$:
\begin{align}
\text{\ensuremath{\textrm{min}_{\gamma}\hat{g}_{n_{k}}\left(\beta_{n_{k},0},\gamma\right)}'\ensuremath{\hat{\Omega}_{n_{k}}^{-1}\left(\beta_{n_{k},0},\gamma\right)\hat{g}_{n_{k}}\left(\beta_{n_{k},0},\gamma\right)}} & \leq\tilde{g}_{n_{k},0}'\tilde{\Omega}_{n_{k},0}^{-1}\tilde{g}_{n_{k},0}\label{eq:Lem1-D1}
\end{align}
which holds because $\tilde{\gamma}_{n,0}$ does not necessarily minimize
the expression on the LHS. Then, 
\[
\tilde{g}_{n_{k},0}'\tilde{\Omega}_{n_{k},0}^{-1}\tilde{g}_{n_{k},0}=\tilde{g}_{n_{k},0}'\tilde{\Omega}_{n_{k},0}^{-1/2}M_{\hat{A}_{n_{k}}}\tilde{\Omega}_{n_{k},0}^{-1/2}\tilde{g}_{n_{k},0}+\varpi_{n_{k}}
\]
since $M_{\hat{A}_{n}}+P_{\hat{A}_{n}}=I.$ This establishes part
(i).

We now turn to part (iii) first. We mimic standard calculations for
GMM, see \citet{Hansen1982}. Recall the mean value expansion of $\tilde{g}_{n,0}$
in \ref{eq:gamma_mve} which gives 
\begin{align}
\tilde{g}_{n,0} & =\hat{g}_{n,0}+\hat{g}_{n,\gamma}\left(\check{\gamma}_{n,0}\right)\left(\tilde{\gamma}{}_{n,0}-\gamma_{n,0}\right).\label{eq:Lemma1-D2}
\end{align}
Then premultiply (\ref{eq:Lemma1-D2}) by $\hat{\Omega}_{n,0}^{-1/2}$
such that 
\begin{eqnarray}
\hat{\Omega}_{n,0}^{-1/2}\tilde{g}_{n,0} & = & \hat{\Omega}_{n,0}^{-1/2}\hat{g}_{n,0}+\hat{\Omega}_{n,0}^{-1/2}\hat{g}_{n,\gamma}\left(\tilde{\gamma}{}_{n,0}-\gamma_{n,0}\right)\label{eq:Lemma1-D3-1}\\
 &  & +\hat{\Omega}_{n,0}^{-1/2}\left(\hat{g}_{n,\gamma}-\hat{g}_{n,\gamma}\left(\check{\gamma}_{n,0}\right)\right)\left(\tilde{\gamma}{}_{n,0}-\gamma_{n,0}\right).\nonumber 
\end{eqnarray}
By Lemma \ref{lem:gamma-conv} and Assumptions \ref{assu:H0} and
\ref{assu:g_gamma}, it follows that $\hat{g}_{n_{k},\gamma}-\hat{g}_{n_{k},\gamma}\left(\check{\gamma}_{n_{k},0}\right)=o_{p}\left(1\right)$
such that 
\begin{equation}
\hat{\Omega}_{n_{k},0}^{-1/2}\left(\hat{g}_{n_{k},\gamma}-\hat{g}_{n_{k},\gamma}\left(\check{\gamma}_{n_{k},0}\right)\right)\left(\tilde{\gamma}{}_{n_{k},0}-\gamma_{n_{k},0}\right)=o_{p}\left(n_{k}^{-1/2}\right).\label{eq:Lemma1_D4}
\end{equation}
By substituting \eqref{eq:Def_A_hat} in \eqref{eq:Lemma1-D3-1} together
with \eqref{eq:Lemma1_D4} it follows that 
\begin{eqnarray}
\hat{\Omega}_{n_{k},0}^{-1/2}\tilde{g}_{n_{k},0} & = & \hat{\Omega}_{n_{k},0}^{-1/2}\hat{g}_{n_{k},0}+\hat{A}_{n}\left(\tilde{\gamma}_{n_{k},0}-\gamma_{n_{k},0}\right)\label{eq:Lemma1-D3}\\
 &  & +\hat{\Omega}_{n_{k},0}^{-1/2}\left(I_{d}\otimes\hat{g}_{n_{k},0}'\hat{\Omega}_{n_{k},0}^{-1}\right)\hat{\Lambda}_{n,0}\left(\tilde{\gamma}{}_{n_{k},0}-\gamma_{n_{k},0}\right)+o_{p}\left(n_{k}^{-1/2}\right).\nonumber 
\end{eqnarray}
The last term on the RHS is $o_{p}\left(n_{k}^{-1/2}\right)$ along
all converging subsequences $n_{k}$ by Lemma \ref{lem:gamma-conv}(i)
and Assumptions \ref{assu:Omega_Conv} and \ref{assu:CLT}. Now substitute
\eqref{eq:gamma_til_TSVD} in \eqref{eq:Lemma1-D3} such that 
\begin{eqnarray}
\hat{\Omega}_{n_{k},0}^{-1/2}\tilde{g}_{n_{k},0} & = & \hat{\Omega}_{n_{k},0}^{-1/2}\hat{g}_{n_{k},0}+\hat{A}_{n_{k}}\left(\tilde{\gamma}{}_{n_{k},0}-\gamma_{n_{k},0}\right)+o_{p}\left(n_{k}^{-1/2}\right)\label{eq:Lemma1-D4}\\
 & = & \hat{\Omega}_{n_{k},0}^{-1/2}\hat{g}_{n_{k},0}-\hat{A}_{n_{k}}\mathring{\Gamma}_{n_{k}}^{+}\hat{A}_{n_{k}}'\hat{\Omega}_{n_{k},0}^{-1/2}\hat{g}_{n_{k},0}+o_{p}\left(n_{k}^{-1/2}\right)\nonumber \\
 & = & \left(I-P_{\hat{A}_{n_{k}}}\right)\hat{\Omega}_{n_{k},0}^{-1/2}\hat{g}_{n_{k},0}+\left(P_{\hat{A}_{n_{k}}}-\hat{A}_{n_{k}}\mathring{\Gamma}_{n_{k}}^{+}\hat{A}_{n_{k}}'\right)\hat{\Omega}_{n_{k},0}^{-1/2}\hat{g}_{n_{k},0}+o_{p}\left(n_{k}^{-1/2}\right)\nonumber \\
 & = & \left(I-P_{\hat{A}_{n_{k}}}\right)\hat{\Omega}_{n_{k},0}^{-1/2}\hat{g}_{n_{k},0}+o_{p}\left(n_{k}^{-1/2}\right)\nonumber 
\end{eqnarray}
where the second equality uses \eqref{eq:gamma_til_TSVD}. The last
equality follows from Lemma \ref{lem:gamma-conv}(iii) and $\hat{g}_{n_{k},0}=O_{p}\left(n_{k}^{-1/2}\right)$.
Next note that by Assumption \ref{assu:g_gamma} $\hat{g}_{n_{k},\gamma}\left(\check{\gamma}_{n_{k},0}\right)=O_{p}\left(1\right)$.
By Lemma \ref{lem:gamma-conv}(i) and \eqref{eq:gamma_mve} it then
follows that $\tilde{g}_{n_{k},0}-\hat{g}_{n_{k},0}=O_{p}\left(n_{k}^{-1/2}\right)$.
By Assumption \ref{assu:Omega_Conv} and Lemma \ref{lem:gamma-conv}(i)
it follows that $\hat{\Omega}_{n_{k},0}^{-1/2}-\tilde{\Omega}_{n_{k},0}^{-1/2}=o_{p}\left(1\right)$.
Then,
\begin{eqnarray}
\hat{\Omega}_{n_{k},0}^{-1/2}\tilde{g}_{n_{k},0} & = & \tilde{\Omega}_{n_{k},0}^{-1/2}\tilde{g}_{n_{k},0}+\left(\hat{\Omega}_{n_{k},0}^{-1/2}-\tilde{\Omega}_{n_{k},0}^{-1/2}\right)\tilde{g}_{n_{k},0}\nonumber \\
 & = & \tilde{\Omega}_{n_{k},0}^{-1/2}\tilde{g}_{n_{k},0}+\left(\hat{\Omega}_{n_{k},0}^{-1/2}-\tilde{\Omega}_{n_{k},0}^{-1/2}\right)\left(\hat{g}_{n,0}+\hat{g}_{n,\gamma}\left(\check{\gamma}_{n,0}\right)\left(\tilde{\gamma}{}_{n,0}-\gamma_{n,0}\right)\right)\label{eq:Lemma1-D4a}\\
 & = & \tilde{\Omega}_{n_{k},0}^{-1/2}\tilde{g}_{n_{k},0}+o_{p}\left(n_{k}^{-1/2}\right).\nonumber 
\end{eqnarray}
By substituting the LHS in \eqref{eq:Lemma1-D4} for the RHS in \eqref{eq:Lemma1-D4a}
one obtains 
\begin{equation}
\tilde{\Omega}_{n_{k},0}^{-1/2}\tilde{g}_{n_{k},0}=\left(I-P_{\hat{A}_{n}}\right)\hat{\Omega}_{n,0}^{-1/2}\hat{g}_{n,0}+o_{p}\left(n_{k}^{-1/2}\right).\label{eq:Lemma1_D5}
\end{equation}
This establishes part (iii) and \eqref{eq:Lemma2_Res2}.

The claim in (ii) and \eqref{eq:Lemma2_Res1} follows from Lemma \ref{lem:gamma-conv}(ii)
because 
\[
\hat{A}_{n_{k}}'\tilde{\Omega}_{n_{k}}^{-1/2}\tilde{g}_{n_{k},0}=o_{p}\left(n_{k}^{-1/2}\right)
\]
which implies that 
\[
P_{\hat{A}_{n_{k}}}'\tilde{\Omega}_{n_{k}}^{-1/2}\tilde{g}_{n_{k},0}=o_{p}\left(n_{k}^{-1/2}\right).
\]
Alternatively, one can argue that 
\[
P_{\hat{A}_{n_{k}}}'\tilde{\Omega}_{n_{k}}^{-1/2}\tilde{g}_{n_{k},0}=P_{\hat{A}_{n_{k}}}\left(I-P_{\hat{A}_{n_{k}}}\right)\hat{\Omega}_{n,0}^{-1/2}\hat{g}_{n,0}+o_{p}\left(n_{k}^{-1/2}\right)=o_{p}\left(n_{k}^{-1/2}\right)
\]
 where the first equality uses \eqref{eq:Lemma1_D5} and the second
equality uses the fact that $P_{\hat{A}_{n}}$ is a projection.
\end{proof}
We employ a device that was developed by GKMC12 to control for possible
rank deficiency of $\Pi_{n,w}$ along converging subsequences, see
in particular p.2662 of that paper. For a converging subsequence $n_{k}$
define
\begin{equation}
Q_{n_{k}}=n_{k}^{1/2}\Pi_{n_{k},\gamma}\label{eq:Qzz_n}
\end{equation}
 and let the singular value decomposition of $Q_{n_{k}}$ be (see
Bhatia, 1996, p. 6)
\begin{equation}
Q_{n_{k}}=R_{n_{k}}\Lambda_{n_{k}}T_{n_{k}}'\label{eq:Qzz_SVD}
\end{equation}
where $R_{n_{k}}$ is a $d\times d$ unitary matrix, $T_{n_{k}}$is
a $d_{\gamma}\times d_{\gamma}$ unitary matrix and $\Lambda_{n_{k}}$
is a $d\times d_{\gamma}$ matrix with diagonal elements $\lambda_{n_{k}j}$
equal to the positive square roots of the eigenvalues of $Q_{n_{k}}'Q_{n_{k}}$.
Consider $\tilde{Q}_{n_{k}}=\Pi_{n_{k},\gamma}=R_{n_{k}}\tilde{\Lambda}_{n_{k}}T_{n_{k}}'$
with $\tilde{\Lambda}_{n_{k}}=n_{k}^{-1/2}\Lambda_{n_{k}}$ and elements
$\text{\ensuremath{\tilde{\lambda}_{n_{k}j}}}=n_{k}^{-1/2}\lambda_{n_{k}j}.$
By continuity, along converging subsequences $n_{k}$ , it follows
that $\Pi_{n_{k},\gamma}$ converges to some matrix $\Pi_{\gamma}$
that may depend on the subsequence $n_{k}$. Then, $\tilde{Q}_{n_{k}}$
converges, to a matrix with possibly reduced rank $\tilde{Q}=\Pi_{\gamma}=R\tilde{\Lambda}T'$.

The singular value decomposition of $\tilde{Q}_{n_{k}}'\tilde{Q}_{n_{k}}$
is $T_{n_{k}}\tilde{\Lambda}_{n_{k}}'\tilde{\Lambda}_{n_{k}}T_{n_{k}}'$
with eigenvalues $\tilde{\lambda}_{n_{k}j}^{2}$. By continuity of
the eigenvalues, see (\ref{eq:Weyl-Perturb}), $\tilde{\lambda}_{n_{k}j}^{2}\rightarrow\tilde{\lambda}_{j}^{2}$
where all $\tilde{\lambda}_{j}^{2}\geq0$ are bounded and some may
be zero. By consequence, $\lambda_{n_{k}j}\rightarrow\text{\ensuremath{\lambda_{j}}}$
with some $\lambda_{j}=\infty.$ Collect all eigenvalues $\lambda_{j}$
in a corresponding matrix $\Lambda.$ Assume without loss of generality,
as in GKMC12, that for $j\leq p$ and $0\leq p\leq d$, the $j$-th
diagonal element of $\Lambda,$ $\lambda_{j}<\infty,$ while $\lambda_{j}=\infty$
for $j>p.$ Again as in GKMC12, define a full rank diagonal $d_{\gamma}\times d_{\gamma}$
matrix $L_{n}$ with diagonal element $\left[L_{n}\right]_{jj}=l_{n,j}=\lambda_{n,j}^{-1}$
where 
\begin{equation}
\left[L_{n}\right]_{jj}=\left\{ \begin{array}{cc}
\lambda_{n,j}^{-1} & \text{ if \ensuremath{j>p} }\\
1 & \text{if \ensuremath{j\leq p}}
\end{array}\right..\label{eq:L}
\end{equation}
The left and right singular vectors $R_{n_{k}}$ and $T_{n_{k}}$
of $\tilde{Q}_{n_{k}}$ converge to $R$ and $T$ in the following
sense. Let $\mathscr{A}_{n}=\tilde{Q}_{n}\tilde{Q}_{n}'$, $\mathscr{A}=\tilde{Q}\tilde{Q}',$
$\mathscr{B}_{n}=\tilde{Q}_{n}'\tilde{Q}_{n},$ and $\mathscr{B}=\tilde{Q}'\tilde{Q}.$
Note that $\mathscr{A}_{n},\mathscr{A},\mathscr{B}_{n}$ and $\mathscr{B}$
are Hermitian and therefore normal matrices with eigenvectors $R_{n},$$R$,
$T_{n}$ and $T$ respectively. Let $\hat{R}_{1,n}=P_{\mathscr{A}_{n}}\left(\mathscr{L}_{1}\right)$
be the projection matrix projecting onto the subspace spanned by the
eigenvectors $R_{n}$ of $\mathscr{A}_{n}$ associated with eigenvalues
who are elements of the set $\mathscr{L}_{1}.$ Similarly, define
$\hat{R_{2}}=P_{\mathscr{A}}\left(\mathscr{L}_{2}\right)$, $\hat{T}_{1,n}=P_{\mathscr{B}_{n}}\left(\mathscr{L}_{1}\right)$
and $\hat{T}_{2}=P_{\mathscr{B}}\left(\mathscr{L}_{2}\right).$

For $n_{k}$ large enough, distinct eigenvalues of $\tilde{Q}_{n_{k}}$and
$\tilde{Q}$ can be placed into separated subsets. Then it follows
from Theorem \ref{thm:Eigenvector} that eigenvectors $R_{n_{k}}$
associated with a distinct eigenvalue $\tilde{\lambda}_{n_{k}j}$
converge to the eigenvector in $R$ associated with eigenvalue $\tilde{\lambda}_{j}$.
This holds because in the limit $\text{\ensuremath{\hat{R}_{1,n_{k}}}}$
is orthogonal to all other eigenvectors of $R$ by (\ref{eq:eig_ortho}).
Another way to understand this result is to note that the composition
of projections projecting onto orthogonal subspaces is zero, for example
$P_{\mathscr{A}}\left(\mathscr{L}_{1}\right)P_{\mathscr{A}}\left(\mathscr{L}_{2}\right)=0$.
As $\mathscr{A}_{n}$ approaches $\mathscr{A}$ the same relationship
holds approximately for $\hat{R}_{1,n}$ and $\hat{R_{2}}$.

For eigenvectors related to eigenvalues with multiplicities, the result
shows that these eigenvectors span a space that is orthogonal to the
space spanned by all eigenvectors in $R$ not associated with that
eigenvalue. The limits of these eigenvectors with multiple eigenvalues
can be chosen to equal the corresponding eigenvectors in $R.$

The next Lemma characterizes the upper bound in \eqref{eq:up-bound}
and establishes that $n^{1/2}\hat{A}_{n}T_{n}L_{n}$, where $T_{n}$
and $L_{n}$ are defined in \eqref{eq:Qzz_SVD} and \eqref{eq:L},
converges to a full column rank matrix along converging subsequences
$n_{k}$.
\begin{lem}
\label{lem:LemA2}Let $\hat{A}_{n}$ be defined as in \eqref{eq:Def_A_hat}.
Then, for any converging subsequence $n_{k}$, there exist sequences
of full rank matrices $T_{n_{k}}$ and $L_{n_{k}}$of dimension $d_{\gamma}\times d_{\gamma}$
and a matrix $A$ with possibly random elements such that $n_{k}^{1/2}\hat{A}_{n_{k}}T_{n_{k}}L_{n_{k}}\rightarrow_{d}A$
and where $A$ has full column rank w.p.1.
\end{lem}
\begin{proof}
Note that statements made with subscript $n$ hold for all $n=1,...$while
statements made with subscript $n_{k}$ only hold along converging
subsequences $n_{k}$. Let $\omega_{n,V}=n^{-1/2}\sum_{t=1}^{n}V_{n,t}$
and $S_{n,V}=\vec\left(\omega_{n,V}\right)$. Now use the definition
of $V_{n,t}$ in \eqref{eq:Vnt} and write 
\[
n^{1/2}\vec\left(\hat{g}_{n,\gamma}\right)=n^{-1/2}\sum_{t=1}^{n}\vec\left(V_{n,t}+\Pi_{n,\gamma}\right)=\vec\left(n^{1/2}\Pi_{n,\gamma}+\omega_{n,V}\right).
\]
We then consider 
\begin{equation}
n_{k}^{1/2}\hat{A}_{n_{k}}T_{n_{k}}L_{n_{k}}=\hat{\Omega}_{n,0}^{-1/2}\left(\left(n^{1/2}\Pi_{n,\gamma}+\omega_{n,V}\right)-\left(I_{d}\otimes\hat{g}_{n,0}'\hat{\Omega}_{n,0}^{-1}\right)\hat{\Lambda}_{n}\left(\gamma_{n,0}\right)\right)T_{n_{k}}L_{n_{k}}.\label{eq:nATL}
\end{equation}
Also define $S_{n,\varepsilon}=n^{-1/2}\sum_{t=1}^{n}\varepsilon_{n,t}$
and partition the limiting process $S_{n}\rightarrow_{d}\omega$ as
$\omega=\left(\omega_{\varepsilon}',\vec\omega_{V}'\right)'$ where
$\omega_{\varepsilon}$ is a $d\times1$ vector of Gaussian random
variables, $\vec\omega_{V}$ is a $dd_{\gamma}$ dimensional vector
of Gaussian random variables and where $\omega\sim N\left(0,\Sigma\right)$
by Assumption \ref{assu:CLT}. Write 
\begin{equation}
\Sigma=\left[\begin{array}{cc}
\Sigma_{\varepsilon} & \Sigma_{\varepsilon V}\\
\Sigma_{V\varepsilon} & \Sigma_{V}
\end{array}\right]\label{eq:Sig_part}
\end{equation}
where the sub-matrices $\Sigma_{\varepsilon},$ $\Sigma_{V\varepsilon}=\Sigma_{\varepsilon V}'$
and $\Sigma_{V}$ are partitioned according to $\omega_{\varepsilon}$
and $\omega_{V}.$

Using the definitions of $T_{n_{k}}$ and $L_{n_{k}}$ in \eqref{eq:Qzz_SVD}
and \eqref{eq:L} and for $Q_{n_{k}}$ defined in \eqref{eq:Qzz_n}
it follows that the first part of the RHS in \eqref{eq:nATL} can
be written as
\begin{equation}
n_{k}^{1/2}\text{\ensuremath{\hat{g}_{n_{k},\gamma}}}T_{n_{k}}L_{n_{k}}=Q_{n_{k}}T_{n_{k}}L_{n_{k}}+\omega_{Vn_{k}}T_{n_{k}}L_{n_{k}}.\label{eq:g_gamma_decomp}
\end{equation}
Then, along converging subsequences it follows from Assumption \ref{assu:g_gamma},
Theorem \eqref{thm:Weyl}, \eqref{eq:Vnt} and \eqref{eq:eig_ortho}
that $\text{\ensuremath{Q_{n_{k}}}}T_{n_{k}}L_{n_{k}}=R_{n_{k}}\Lambda_{n_{k}}L_{n_{k}}\rightarrow R\bar{L}$
where $\bar{L}$ is a $d_{\gamma}\times d_{\gamma}$ matrix with diagonal
elements equal to $\lambda_{j}$ for $j\leq p$ and equal to one for
$j>p.$ By Assumption \ref{assu:CLT}, 
\begin{equation}
\omega_{Vn_{k}}T_{n_{k}}L_{n_{k}}\rightarrow_{d}\left(\omega_{V}T_{1},....,\omega_{V}T_{p},0...,0\right)=\bar{\omega}_{V}\label{eq:omega_bar}
\end{equation}
where $T_{j}$ is the $j$-th column of $T$ and $\vec\left(\omega_{V}\right)\sim N\left(0,\Sigma_{V}\right)$
such that $\vec\left(\omega_{V}T_{1},....,\omega_{V}T_{p}\right)\sim N\left(0,\Sigma_{V,pp}\right)$
where $\Sigma_{V,pp}$ is the upper corner $p\times p$ block of $\Sigma_{V}$
because $T_{j}$ are orthonormal vectors. It now follows that 
\begin{equation}
n_{k}^{1/2}\text{\ensuremath{\hat{g}_{n_{k},\gamma}}}T_{n_{k}}L_{n_{k}}\rightarrow_{d}R\bar{L}+\bar{\omega}_{V}\label{eq:g_gamma_conv}
\end{equation}
where the RHS is full column rank with probability 1 as in GKMC12.
Next consider the limit of $\hat{\Lambda}_{n,0}=\hat{\Lambda}_{n}\left(\gamma\right)$
defined in (\ref{eq:Lambda_hat}). Recall that 
\[
\hat{\Lambda}_{n,0}=n^{-1}\sum_{t=1}^{n}\sum_{s=1}^{n}k\left(\frac{t-s}{a_{n}}\right)\left(g_{s}^{\gamma}\left(\gamma_{n,0}\right)\otimes\varepsilon_{n,t}\right)
\]
where
\[
g_{t}^{\gamma}\left(\gamma_{n,0}\right)=V_{n,t}+\Pi_{n,\gamma}
\]
from \eqref{eq:Vnt}. Note that 
\[
\left(\Pi_{n,\gamma}\otimes\varepsilon_{n,t}\right)=\left(I_{d}\otimes\varepsilon_{n,t}\right)\Pi_{n,\gamma}
\]
such that $\hat{\Lambda}_{n,0}$ can be decomposed into 
\begin{equation}
\hat{\Lambda}_{n,0}=\hat{\Lambda}_{\gamma,n,0}\Pi_{n,\gamma}+\hat{\Lambda}_{V,n,0}\label{eq:Lambda_decomp}
\end{equation}
 with 
\begin{equation}
\hat{\Lambda}_{\gamma,n,0}=n^{-1}\sum_{t=1}^{n}\sum_{s=1}^{n}k\left(\frac{t-s}{a_{n}}\right)\left(I_{d}\otimes\varepsilon_{n,t}\right)\label{eq:Lam_z_Def}
\end{equation}
and 
\begin{equation}
\hat{\Lambda}_{V,n,0}=n^{-1}\sum_{t=1}^{n}\sum_{s=1}^{n}k\left(\frac{t-s}{a_{n}}\right)\left(V_{n,s}\otimes\varepsilon_{n,t}\right).\label{eq:Lam_V_Def}
\end{equation}
 Now consider 
\begin{align}
n_{k}^{1/2}\left(I_{d}\otimes\hat{g}_{n_{k},0}'\hat{\Omega}_{n_{k},0}^{-1}\right)\hat{\Lambda}_{\gamma,n_{k},0}\Pi_{n_{k},\gamma}T_{n_{k}}L_{n_{k}} & =\left[\left(I_{d}\otimes\hat{g}_{n_{k},0}'\hat{\Omega}_{n_{k},0}^{-1}\right)\hat{\Lambda}_{\gamma,n,0}\right]\text{\ensuremath{Q_{n_{k},zz}}}T_{n_{k}}L_{n_{k}}\label{eq:LamTL_limit}\\
 & =o_{p}\left(1\right).\nonumber 
\end{align}
where the $n_{k}^{1/2}$ term on the LHS of \eqref{eq:LamTL_limit}
is absorbed in $\ensuremath{Q_{n_{k},zz}}$ which was defined in \eqref{eq:Qzz_n}.
To see that the RHS of \eqref{eq:LamTL_limit} is $o_{p}\left(1\right)$,
note that $\hat{\Omega}_{n_{k},0}^{-1}\rightarrow_{p}\Omega_{0}^{-1}$
where $\Omega_{0}$ is full rank by Assumption \ref{assu:Omega_Conv},
and by Assumption \ref{assu:Lambda},
\begin{equation}
\hat{\Lambda}_{\gamma,n_{k},0}\rightarrow_{p}\Lambda_{\gamma,0}\label{eq:Lam_z_conv}
\end{equation}
for some matrix $\Lambda_{\gamma,0}$ with fixed and finite elements,
and $\hat{g}_{n_{k},0}'\hat{\Omega}_{n_{k},0}^{-1}=o_{p}\left(1\right)$
by Assumption \ref{assu:CLT}. In addition $\ensuremath{Q_{n_{k},zz}}T_{n_{k}}L_{n_{k}}=O_{p}\left(1\right)$
by a previous argument. This establishes that \eqref{eq:LamTL_limit}
is $o_{p}\left(1\right).$

For $\Sigma_{\varepsilon V}$ and $\Sigma_{\varepsilon}$ defined
in \eqref{eq:Sig_part}, note that 
\[
\Sigma_{\varepsilon V}=\lim_{k\rightarrow\infty}n_{k}^{-1}\sum_{s=1}^{n_{k}}\sum_{t=1}^{n_{k}}E\left[\varepsilon_{n_{k},t}\vec\left(V_{n_{k},s}\right)'\right]
\]
 such that 
\[
\vec\left(\Sigma_{\varepsilon V}\right)=\lim_{k\rightarrow\infty}n_{k}^{-1}\sum_{s=1}^{n_{k}}\sum_{t=1}^{n_{k}}E\left[\left(\vec\left(V_{n_{k},s}\right)\otimes\varepsilon_{n_{k},t}\right)\right].
\]
Similarly, note that 
\[
\vec\left(V_{n,s}\otimes\varepsilon_{n,t}\right)=\left(\vec\left(V_{n,s}\right)\otimes\varepsilon_{n,t}\right)
\]
and
\[
\Lambda_{V,0}=\lim_{k\rightarrow\infty}n_{k}^{-1}\sum_{s=1}^{n_{k}}\sum_{t=1}^{n_{k}}E\left[V_{n,s}\otimes\varepsilon_{n,t}\right].
\]
Write the elements of $\Lambda_{V,0}$ as 
\begin{equation}
\Lambda_{V,0}=\left(\begin{array}{ccc}
\Lambda_{V,0}^{11} & \cdots & \Lambda_{V,0}^{1d_{\gamma}}\\
\vdots & \ddots & \vdots\\
\Lambda_{V,0}^{d1} & \cdots & \Lambda_{V,0}^{dd_{\gamma}}
\end{array}\right)\label{eq:Lambda_part}
\end{equation}
where $\Lambda_{V,0}^{ij}$ are vectors of dimension $d.$ This implies
that the term $\hat{\Lambda}_{V,n_{k},0}$ in \eqref{eq:Lam_V_Def}
satisfies 
\begin{equation}
\vec\left(\hat{\Lambda}_{V,n_{k},0}\right)\rightarrow_{p}\vec\left(\Lambda_{V,0}\right)=\vec\left(\Sigma_{\varepsilon V}\right)\label{eq:Lam_V_conv}
\end{equation}
and where the convergence follows by Assumptions \ref{assu:Lambda}
and \ref{assu:CLT}. Also note that Assumptions \ref{assu:Omega_Conv}
and \ref{assu:CLT} lead to 
\begin{equation}
n_{k}^{1/2}\hat{g}_{n_{k},0}'\hat{\Omega}_{n_{k},0}^{-1}\rightarrow_{d}\omega_{\varepsilon}'\Sigma_{\varepsilon}^{-1}.\label{eq:gOmega_conv}
\end{equation}
Now consider 
\begin{align}
n_{k}^{1/2}\left(I_{d}\otimes\hat{g}_{n_{k},0}'\hat{\Omega}_{n_{k},0}^{-1}\right)\hat{\Lambda}_{V,n_{k},0} & \rightarrow_{d}\left(I_{d}\otimes\omega_{\varepsilon}'\Sigma_{\varepsilon}^{-1}\right)\Lambda_{V,0}.\label{eq:gOmgLam_conv}
\end{align}
To represent the limit \eqref{eq:gOmgLam_conv} let 
\[
\omega_{V\varepsilon}=\left(I_{d}\otimes\omega_{\varepsilon}'\Sigma_{\varepsilon}^{-1}\right)\Lambda_{V,0}.
\]
Then, using \eqref{eq:Lambda_part}
\[
\left(I_{d}\otimes\omega_{\varepsilon}'\Sigma_{\varepsilon}^{-1}\right)\Lambda_{V,0}=\left(\begin{array}{ccc}
\omega_{\varepsilon}'\Sigma_{\varepsilon}^{-1}\Lambda_{V,0}^{11} & \cdots & \omega_{\varepsilon}'\Sigma_{\varepsilon}^{-1}\Lambda_{V,0}^{1d_{\gamma}}\\
\vdots & \ddots & \vdots\\
\omega_{\varepsilon}'\Sigma_{\varepsilon}^{-1}\Lambda_{Vz,0}^{d1} & \cdots & \omega_{\varepsilon}'\Sigma_{\varepsilon}^{-1}\Lambda_{V,0}^{dd_{\gamma}}
\end{array}\right)
\]
where $\omega_{\varepsilon}'\Sigma_{\varepsilon}^{-1}\Lambda_{V,0}^{ij}$
are scalars such that 
\begin{equation}
E\left[\omega_{\varepsilon}\vec\left(\omega_{V\varepsilon}\right)'\right]=\left(E\left[\omega_{\varepsilon}\omega_{\varepsilon}'\Sigma_{\varepsilon}^{-1}\Lambda_{V,0}^{11}\right],...,E\left[\omega_{\varepsilon}\omega_{\varepsilon}'\Sigma_{\varepsilon}^{-1}\Lambda_{V,0}^{d1}\right],....,E\left[\omega_{\varepsilon}\omega_{\varepsilon}'\Sigma_{\varepsilon}^{-1}\Lambda_{V,0}^{dd_{\gamma}}\right]\right)=\Sigma_{\varepsilon V}\label{eq:E_veps_omegaVeps}
\end{equation}
Collect terms in \ref{eq:nATL} using\eqref{eq:g_gamma_decomp}, \eqref{eq:Lambda_decomp}
and \eqref{eq:LamTL_limit}, such that 
\begin{align*}
n_{k}^{1/2}\hat{A}_{n_{k}}T_{n_{k}}L_{n_{k}} & =\hat{\Omega}_{n_{k},0}^{-1/2}\left(\left(n_{k}^{1/2}\Pi_{n_{k},\gamma}+\omega_{n_{k},V}\right)-\left(I_{d}\otimes n_{k}^{1/2}\hat{g}_{n_{k},0}'\hat{\Omega}_{n_{k},0}^{-1}\right)\hat{\Lambda}_{V,n_{k},0}\right)T_{n_{k}}L_{n_{k}}+o_{p}\left(1\right).
\end{align*}
It now follows from \eqref{eq:omega_bar} and \eqref{eq:gOmgLam_conv}
that 
\begin{equation}
\left(\omega_{n_{k},V}-\left(I_{d}\otimes n_{n_{k}}^{1/2}\hat{g}_{n,0}'\hat{\Omega}_{n,0}^{-1}\right)\hat{\Lambda}_{V,n_{k},0}\right)T_{n_{k}}L_{n_{k}}\rightarrow_{d}\bar{\omega}_{V}-\overline{\omega_{V\varepsilon}}\label{eq:omega-conv}
\end{equation}
and where $\overline{\omega_{V\varepsilon}}=\left(\omega_{V\varepsilon}T_{1},...,\omega_{V\varepsilon}T_{p},0,...0\right)$.
The column vectors $T_{1},....,T_{p}$ are as in (\ref{eq:omega_bar}).
Introduce the notation $\omega_{V.\varepsilon}=\omega_{V}-\omega_{V\varepsilon}$
and $\overline{\omega_{V.\varepsilon}}=\bar{\omega}_{V}-\overline{\omega_{V\varepsilon}}$.
It follows from Assumption $\ref{assu:CLT}$ that $\omega_{\varepsilon}$
and $\vec\left(\omega_{V}\right)$ are jointly Gaussian, and therefore
that $\omega_{\varepsilon}$ and $\vec\left(\omega_{V.\varepsilon}\right)$
are joint Gaussian. It follows from \eqref{eq:E_veps_omegaVeps} that
\begin{equation}
E\left[\omega_{\varepsilon}\vec\left(\omega_{V.\varepsilon}\right)'\right]=E\left[\omega_{\varepsilon}\vec\left(\omega_{V}-\omega_{V\varepsilon}\right)'\right]=\Sigma_{\varepsilon V}-E\left[\omega_{\varepsilon}\vec\left(\omega_{V\varepsilon}\right)'\right]=0.\label{eq:omega_e_omega_V.e_corr}
\end{equation}
This establishes that $\omega_{V.\varepsilon}$ and $\omega_{\varepsilon}$
are independent. Combining \eqref{eq:g_gamma_conv}, \eqref{eq:Lam_z_conv},
\eqref{eq:Lam_V_conv}, \eqref{eq:gOmega_conv} and \eqref{eq:omega-conv}
then leads to 
\begin{equation}
n_{k}^{1/2}\hat{A}_{n_{k}}T_{n_{k}}L_{n_{k}}\rightarrow_{d}A=\Omega_{0}^{-1/2}\left(R\bar{L}+\overline{\omega_{V.\varepsilon}}\right)\label{eq:A_conv}
\end{equation}
where $A$ has full rank with probability 1. To see this, note that
$\bar{L}$ is a diagonal matrix with the first $p$ elements equal
to the eigenvalues $\lambda_{j}$ and the remaining elements equal
to one. Thus, the space spanned by $R\bar{L}$ consists of the last
$d_{\gamma}-p$ columns of $R$ and possibly some of the first $p.$
Since $R$ is unitary, this space has rank $d_{\gamma}-p+p^{*}$ where
$0\leq p^{*}\leq p$ is the number of non-zero eigenvalues $\lambda_{j}$
for $0\leq j\leq p.$ In addition, $\omega_{V.\varepsilon}$ is a
set of $p$ Gaussian vectors of dimension $d_{\gamma}$. Because $\Sigma_{V.\varepsilon}=E\left[\vec\omega_{V.\varepsilon}\vec\left(\omega_{V.\varepsilon}\right)'\right]$
is non-singular, $\omega_{V.\varepsilon}T_{1},...,\omega_{V.\varepsilon}T_{p}$
is a full column rank matrix w.p.1. Then, 
\begin{equation}
\left(R\bar{L}+\overline{\omega_{V.\varepsilon}}\right)=\left(\omega_{V.\varepsilon}T_{1}+R_{1}\lambda_{11},...,\omega_{V.\varepsilon}T_{p}+R_{p}\lambda_{pp},R_{p+1},...,R_{d_{\gamma}}\right)\label{eq:RL_omega_lim}
\end{equation}
which is full column rank w.p.1.
\end{proof}
\begin{lem}
\label{lem:LemA3}Let $\hat{A}_{n}$ be defined as in \ref{eq:Def_A_hat}
. Then, for any converging subsequence $n_{k}$ it follows that
\[
n_{k}\hat{g}_{n_{k},0}'\hat{\Omega}_{n_{k},0}^{-1/2}M_{\hat{A}_{n_{k}}}\hat{\Omega}_{n_{k},0}^{-1/2}\hat{g}_{n_{k},0}\rightarrow_{d}\chi_{d-d_{\gamma}}^{2}.
\]
\end{lem}
\begin{proof}
First note that by Lemma \ref{lem:LemA1}(iii) along converging subsequences

\begin{equation}
n_{k}\tilde{g}_{n_{k},0}'\tilde{\Omega}_{n_{k},0}^{-1/2}M_{\hat{A}_{n_{k}}}\tilde{\Omega}_{n_{k},0}^{-1/2}\tilde{g}_{n_{k},0}=n_{k}\hat{g}_{n_{k},0}'\hat{\Omega}_{n_{k},0}^{-1/2}M_{\hat{A}_{n_{k}}}\hat{\Omega}_{n_{k},0}^{-1/2}\hat{g}_{n_{k},0}+o_{p}\left(1\right).\label{eq:Jstat1}
\end{equation}
Note that by Assumption \ref{assu:CLT} and for $S_{n}$ partitioned
as $S_{n}=\left[S'_{n,\varepsilon},S'_{n,V}\right]'$ where $S_{n,\varepsilon}=n^{-1/2}\sum_{i=1}^{n}\varepsilon_{n,t}=\hat{g}_{n,0}$
and $S_{n,V}=n^{-1/2}\sum_{i=1}^{n}\vec\left(V_{n,t}\right)$ it follows
that $\Sigma_{n_{k}}^{-1/2}S_{n}\rightarrow_{d}\Sigma^{-1/2}\omega\sim N\left(0,I\right)$.
The limiting random variable $\omega$ can be partitioned conformingly
as $\omega=\left(\omega_{\varepsilon}',\vec\omega_{V}'\right)'$.
By Assumption \ref{assu:CLT} the matrix $\Sigma_{n_{k}}$ converges
to a limit, denoted by $\Sigma$ which can be partitioned conformingly
with $\omega_{\varepsilon}$ as in \eqref{eq:Sig_part} such that
$\Sigma_{\varepsilon}$ is equal to $\Omega_{0}$.

As in GKMC12, we note that $P_{\hat{A}_{n_{k}}}$ is invariant to
the scaling of $\hat{A}_{n_{k}}$ by $n_{k}^{1/2}$ as well as to
the rotation of the column space of $\hat{A}_{n_{k}}$ by a full rank
matrix $T_{n_{k}}L_{n_{k}}.$ By (\ref{eq:A_conv}) it follows that
$P_{\hat{A}_{n_{k}}}\rightarrow_{d}P_{A}$ by the continuous mapping
theorem. To see this note that $P_{A}$ is continuous in $A$ as long
as $A$ has full column rank, \citeauthor{Stewart1977} (1977, Theorem
2.2). Since $A$ is full column rank with probability one, \citeauthor{vanderVaartWellner1996}
(1996, Theorem 1.3.6) can be applied. This argument shows that $M_{\hat{A}_{n_{k}}}=I-P_{\hat{A}_{n_{k}}}\rightarrow_{d}I-P_{A}=M_{A}.$
In addition, $\omega_{\varepsilon}$ is independent of $A$ by \eqref{eq:omega-conv},
\eqref{eq:omega_e_omega_V.e_corr}, \eqref{eq:A_conv} and \eqref{eq:RL_omega_lim}.
It now follows from \eqref{eq:Jstat1}, Assumption \ref{assu:CLT}
and the continuous mapping theorem that 
\[
n_{k}\hat{g}_{n_{k},0}'\hat{\Omega}_{n_{k},0}^{-1/2}M_{\hat{A}_{n_{k}}}\hat{\Omega}_{n_{k},0}^{-1/2}\hat{g}_{n_{k},0}\rightarrow_{d}\omega_{\varepsilon}'\Sigma_{\varepsilon}^{-1/2}M_{A}\Sigma_{\varepsilon}^{-1/2}\omega_{\varepsilon}
\]
where, conditionally on $A,$ $\Sigma_{\varepsilon}^{-1/2}\omega_{\varepsilon}\sim N\left(0,I_{d}\right).$
Since this distribution does not depend on $A,$ it follows by the
same arguments as in GKMC12 and the fact that $M_{A}$ is projection
matrix of rank $d-d_{\gamma}$ that the distribution is $\omega_{\varepsilon}'\Sigma_{\varepsilon}^{-1/2}M_{A}\Sigma_{\varepsilon}^{-1/2}\omega_{\varepsilon}\sim\chi_{d-d_{\gamma}}^{2}.$
\end{proof}
We now present the proof of the main theorem. 
\begin{proof}[Proof of Theorem \ref{thm:AsySz}]
The proof of the Theorem follows arguments in GKMC12 and \citet{Andrews2020}
with the necessary adjustments. Let $\chi_{n}^{*}$ be a double array
that solves 
\begin{equation}
\sup_{\chi_{n}\in\mathscr{X}}P\left(\text{\ensuremath{\textrm{AR}_{\textrm{C}}\left(\beta_{n,0}\right)}}>c_{1-\alpha,\chi_{d-d_{\gamma}}^{2}}\right)
\end{equation}
and denote by $P_{\chi_{n}^{*}}$ the measure induced by $\chi_{n}^{*}$.
Then, 
\begin{eqnarray}
\textrm{AsySz}_{\alpha} & = & \limsup_{n\rightarrow\infty}P_{\chi_{n}^{*}}\left(\text{\ensuremath{\textrm{AR}_{\textrm{C}}\left(\beta_{n,0}\right)}}>c_{1-\alpha,\chi_{d-d_{\gamma}}^{2}}\right).\label{eq:Proof_Thm1_D1}
\end{eqnarray}
Now let $n_{k}$ be a converging subsequence such that 
\begin{equation}
\limsup_{n\rightarrow\infty}P_{\chi_{n}^{*}}\left(\text{\ensuremath{\textrm{AR}_{\textrm{C}}\left(\beta_{n,0}\right)}}>c_{1-\alpha,\chi_{d-d_{\gamma}}^{2}}\right)=\lim_{k\rightarrow\infty}P_{\chi_{n_{k}}^{*}}\left(\text{\ensuremath{\textrm{AR}_{\textrm{C}}\left(\beta_{n_{k},0}\right)}}>c_{1-\alpha,\chi_{d-d_{\gamma}}^{2}}\right).\label{eq:Proof_Thm1_D2}
\end{equation}
Such a subsequence $n_{k}$ exists by the properties of '$\limsup$'
and Assumption \ref{assu:Asy_Tightness}. By Lemma \ref{lem:LemA1}(i)
and \eqref{eq:up-bound} the following inequality holds for all converging
subsequences $n_{k}$
\begin{equation}
\textrm{AR}_{\textrm{C}}\left(\beta_{n_{k},0}\right)\leq n_{k}\tilde{g}_{n_{k},0}'\tilde{\Omega}_{n_{k},0}^{-1/2}M_{\hat{A}_{n_{k}}}\tilde{\Omega}_{n_{k},0}^{-1/2}\tilde{g}_{n_{k},0}+\varpi_{n_{k}}.\label{eq:Proof_Thm1_D3}
\end{equation}
Combining \eqref{eq:Proof_Thm1_D1}, \eqref{eq:Proof_Thm1_D2} and
\eqref{eq:Proof_Thm1_D3} gives 
\begin{equation}
\textrm{AsySz}_{\alpha}\leq\lim_{k\rightarrow\infty}P_{\chi_{n_{k}}^{*}}\left(n_{k}\tilde{g}_{n_{k},0}'\tilde{\Omega}_{n_{k},0}^{-1/2}M_{\hat{A}_{n_{k}}}\tilde{\Omega}_{n_{k},0}^{-1/2}\tilde{g}_{n_{k},0}+\varpi_{n_{k}}>c_{1-\alpha,\chi_{d-d_{\gamma}}^{2}}\right).\label{eq:Proof_Thm1_D4}
\end{equation}
By Lemmas \ref{lem:LemA1}, \ref{lem:LemA2} and \ref{lem:LemA3}
and for any converging subsequence $n_{k}$, it follows that 
\begin{equation}
n\tilde{g}_{n_{k},0}'\tilde{\Omega}_{n_{k},0}^{-1/2}M_{\hat{A}_{n_{k}}}\tilde{\Omega}_{n_{k},0}^{-1/2}\tilde{g}_{n_{k},0}+\varpi_{n_{k}}\rightarrow_{d}\chi_{d-d_{\gamma}}^{2}.\label{eq:AR_conv}
\end{equation}
Then, \eqref{eq:Proof_Thm1_D4} and \eqref{eq:AR_conv} imply that
$\textrm{AsySz}_{\alpha}\leq\alpha.$ We show strict equality by following
the argument in GKMC12. When the CUE estimator $\hat{\gamma}_{n}$
is identified it follows that $\hat{\gamma}_{n_{k}}-\gamma_{n_{k},0}=O_{p}\left(n_{k}^{-1/2}\right).$
Then, the arguments in Lemmas \ref{lem:LemA1}, \ref{lem:LemA2} and
\ref{lem:LemA3} imply that 
\[
\textrm{AR}_{\textrm{C}}\left(\beta_{n_{k},0}\right)=n_{k}\hat{g}_{n_{k},0}'\hat{\Omega}_{n_{k},0}^{-1/2}M_{\hat{A}_{n_{k}}}\Omega_{n_{k},0}^{-1/2}\hat{g}_{n_{k},0}+o_{p}\left(1\right)
\]
and 
\[
n\hat{g}_{n_{k},0}'\hat{\Omega}_{n_{k},0}^{-1/2}M_{\hat{A}_{n_{k}}}\Omega_{n_{k},0}^{-1/2}\hat{g}_{n_{k},0}\rightarrow_{d}\chi_{d-d_{\gamma}}^{2}
\]
along all converging subsequences where $\gamma$ is identified. This
implies that the size of $\text{\ensuremath{\textrm{AR}_{\textrm{C}}\left(\beta_{n_{k},0}\right)}}$
is equal to $\alpha$ in the identified case and it follows that $\textrm{AsySz}_{\alpha}=\alpha.$
\end{proof}

\subsection{Derivation of CUE First Order Condition\label{subsec:Derivation-of-CUE}}

Consider the criterion $Q_{n}\left(\beta,\gamma\right)=n\hat{g}_{n}\left(\beta,\gamma\right)'\hat{\Omega}_{n}^{-1}\left(\beta,\gamma\right)\hat{g}_{n}\left(\beta,\gamma\right)$
evaluated at $Q_{n}\left(\beta_{n,0},\gamma\right)$. We evaluate
the first order condition for the minimization problem $\min_{\gamma}Q_{n}\left(\beta_{n,0},\gamma\right)$.
Following \citet{Magnus1988}, p. 183 as well as the discussion in
Section 13, we start by computing total derivatives of $Q_{n}\left(\beta_{n,0},\gamma\right).$
We use the shorthand notation $\hat{g}_{n}'\Omega_{n}^{-1}\hat{g}_{n}$
when no confusion arises. Using $d$ as the differential operator
We have 
\begin{align*}
d\left(\hat{g}_{n}'\Omega_{n}^{-1}\hat{g}_{n}\right) & =\left(d\hat{g}_{n}\right)'\Omega_{n}^{-1}\hat{g}_{n}+\hat{g}_{n}'\Omega_{n}^{-1}\left(d\hat{g}_{n}\right)-\hat{g}_{n}'\Omega_{n}^{-1}d\Omega_{n}\Omega_{n}^{-1}\hat{g}_{n}\\
 & =2\hat{g}_{n}'\Omega_{n}^{-1}\left(d\hat{g}_{n}\right)-\hat{g}_{n}'\Omega_{n}^{-1}d\Omega_{n}\Omega_{n}^{-1}\hat{g}_{n}.
\end{align*}
Since the LHS is a scalar it follows that 
\begin{align*}
d\left(\hat{g}_{n}'\Omega_{n}^{-1}\hat{g}_{n}\right) & =\left(2d\hat{g}_{n}-\hat{g}_{n}'\Omega_{n}^{-1}d\Omega_{n}\right)'\Omega_{n}^{-1}\hat{g}_{n}\\
 & =\left(\vec\left(2d\hat{g}_{n}-\hat{g}_{n}'\Omega_{n}^{-1}d\Omega_{n}\right)\right)'\Omega_{n}^{-1}\hat{g}_{n}
\end{align*}
where the second quality follows because $d\hat{g}_{n}$, $\hat{g}_{n}'\Omega_{n}^{-1}d\Omega_{n}$,
and $\Omega_{n}^{-1}\hat{g}_{n}$ are vectors. Then, it is easy to
see that 
\[
d\left(\hat{g}_{n}'\Omega_{n}^{-1}\hat{g}_{n}\right)=\hat{g}_{n}'\Omega_{n}^{-1}\left(2d\hat{g}_{n}-\left(I_{d}\otimes\hat{g}_{n}'\Omega_{n}^{-1}\right)d\vec\left(\Omega_{n}\right)\right).
\]
Following \citet{Magnus1988}, p. 173, Definition 1, we obtain the
vector of partial derivatives as 
\[
\frac{\partial\hat{g}_{n}'\Omega_{n}^{-1}\hat{g}_{n}}{\partial\gamma'}=\hat{g}_{n}'\Omega_{n}^{-1}\left(2\frac{\partial\hat{g}_{n}}{\partial\gamma'}-\left(I_{d}\otimes\hat{g}_{n}'\Omega_{n}^{-1}\right)\frac{\partial\vec\left(\Omega_{n}\right)}{\partial\gamma'}\right)
\]
where $\partial\hat{g}_{n}/\partial\gamma'$ is a $d\times d_{\gamma}$
dimensional matrix and $\partial\vec\left(\Omega_{n}\right)/\partial\gamma'$
is a $d^{2}\times d_{\gamma}$ dimensional matrix. Finally, consider
the term $d\Omega_{n}$ which can be evaluated as 
\begin{align*}
d\vec\Omega_{n} & =n^{-1}\sum_{t=1}^{n}\sum_{s=1}^{n}k\left(\frac{t-s}{a_{n}}\right)\left\{ g_{s}\left(\beta,\gamma\right)\otimes\left(dg_{t}\left(\beta,\gamma\right)\right)+\left(dg_{s}\left(\beta,\gamma\right)\right)\otimes g_{t}\left(\beta,\gamma\right)\right\} \\
 & =n^{-1}\sum_{t=1}^{n}\sum_{s=1}^{n}k\left(\frac{t-s}{a_{n}}\right)\left\{ \left(g_{s}\left(\beta,\gamma\right)\otimes I_{d}\right)dg_{t}\left(\beta,\gamma\right)+\left(I_{d}\otimes g_{t}\left(\beta,\gamma\right)\right)dg_{s}\left(\beta,\gamma\right)\right\} 
\end{align*}
where the second equality follows from \citet{Magnus1988}, p. 185,
(11) and therefore the derivative can be written as 
\begin{align*}
\frac{\partial\vec\Omega_{n}}{\partial\gamma'} & =n^{-1}\sum_{t=1}^{n}\sum_{s=1}^{n}k\left(\frac{t-s}{a_{n}}\right)\left\{ \left(g_{s}\left(\beta,\gamma\right)\otimes I_{d}\right)\frac{\partial g_{t}\left(\beta,\gamma\right)}{\partial\gamma'}+\left(I_{d}\otimes g_{t}\left(\beta,\gamma\right)\right)\frac{\partial g_{s}\left(\beta,\gamma\right)}{\partial\gamma'}\right\} \\
 & :=\hat{\Lambda}_{n}^{1}+\hat{\Lambda}_{n}^{2}
\end{align*}
However, note that the derivative in $\hat{g}_{n}'\Omega_{n}^{-1}d\Omega_{n}\Omega_{n}^{-1}\hat{g}_{n}$
satisfies the property 
\begin{align*}
d\Omega_{n} & =n^{-1}\sum_{t=1}^{n}\sum_{s=1}^{n}k\left(\frac{t-s}{a_{n}}\right)\left(\left(dg_{s}\left(\beta,\gamma\right)\right)g_{t}\left(\beta,\gamma\right)'+g_{s}\left(\beta,\gamma\right)\left(dg_{t}\left(\beta,\gamma\right)\right)'\right)\\
 & :=\tilde{d\Omega_{n}}+\tilde{d\Omega_{n}}'
\end{align*}
and 
\begin{align*}
\hat{g}_{n}'\Omega_{n}^{-1}d\Omega_{n}\Omega_{n}^{-1}\hat{g}_{n} & =\hat{g}_{n}'\Omega_{n}^{-1}\left(\tilde{d\Omega_{n}}+\tilde{d\Omega_{n}}'\right)\Omega_{n}^{-1}\hat{g}_{n}\\
 & =2\hat{g}_{n}'\Omega_{n}^{-1}\tilde{d\Omega_{n}}\Omega_{n}^{-1}\hat{g}_{n}
\end{align*}
because the criterion function is a scalar which is invariant to transposes.
The implication is that we can focus 
\[
\frac{\partial\hat{g}_{n}'\Omega_{n}^{-1}\hat{g}_{n}}{\partial\gamma'}=2\hat{g}_{n}'\Omega_{n}^{-1}\left(\frac{\partial\hat{g}_{n}}{\partial\gamma'}-\left(I_{d}\otimes\hat{g}_{n}'\Omega_{n}^{-1}\right)\hat{\Lambda}_{n}\right)
\]
and the factor $2$ can be dropped when evaluating first order conditions.
\end{document}